\documentclass{article}

\usepackage[utf8]{inputenc}
\usepackage{booktabs} % For formal tables
\usepackage[ruled,lined]{algorithm2e} % For algorithms
\SetInd{8pt}{5pt} 
\usepackage{fullpage}
\usepackage{amsmath, amsthm, amssymb, mathtools}
\usepackage{mathtools}
\usepackage[libertine]{newtxmath} % this is to have the upper and lowercase \mathscr
\usepackage{subfigure,tikz} % vector graphics %TC: uncomment if this is not included in pgfplots
    \usetikzlibrary{decorations.pathreplacing} % for curly braces in pictures
\usepackage{enumitem} % compact lists
\usepackage{nicefrac} % compact fractions
\usepackage{authblk} % authors with different affiliation 
\usepackage{pgfplots}
\pgfplotsset{compat=1.15}
\usetikzlibrary{arrows}
\usepackage{mleftright} % remove the extra horizontal spacing in left/right
    \mleftright
\usepackage{thmtools,thm-restate} % to reference a lemma in the appendix

\usepackage{natbib}

\usepackage[colorlinks=true,breaklinks=true,bookmarks=true,urlcolor=blue,
     citecolor=blue,linkcolor=blue,bookmarksopen=false,draft=false]{hyperref}

\usepackage[capitalize]{cleveref} %cleveref must be loaded after hyperref!

\newcommand{\cD}{\mathcal{D}}
\newcommand{\cE}{\mathcal{E}}
\newcommand{\cF}{\mathcal{F}}
\newcommand{\cG}{\mathcal{G}}
\newcommand{\cH}{\mathcal{H}}

\newcommand{\cO}{\mathcal{O}}

\newcommand{\cR}{\mathcal{I}}

\newcommand{\cU}{\mathcal{U}}
\newcommand{\cV}{\mathcal{V}}
\newcommand{\cW}{\mathcal{W}}
\newcommand{\cX}{\mathcal{X}}
\newcommand{\cY}{\mathcal{Y}}
\newcommand{\cZ}{\mathcal{Z}}
\newcommand{\E}{\mathbb{E}}
\renewcommand{\exp}[1]{\mathbb{E}\left[ #1 \right]}

\newcommand{\I}{\mathbb{I}}

\newcommand{\N}{\mathbb{N}}

\renewcommand{\P}{\mathbb{P}}
\newcommand{\Q}{\mathbb{Q}}
\newcommand{\R}{\mathbb{R}}

\newcommand{\e}{\varepsilon}
\newcommand{\fhi}{\varphi}
\newcommand{\tht}{\vartheta}
\newcommand{\lrb}[1]{\left(#1\right)}
\newcommand{\brb}[1]{\bigl(#1\bigr)}
\newcommand{\Brb}[1]{\Bigl(#1\Bigr)}
\newcommand{\bbrb}[1]{\biggl(#1\biggr)}
\newcommand{\lsb}[1]{\left[#1\right]}
\newcommand{\bsb}[1]{\bigl[#1\bigr]}
\newcommand{\Bsb}[1]{\Bigl[#1\Bigr]}
\newcommand{\lcb}[1]{\left\{#1\right\}}
\newcommand{\bcb}[1]{\bigl\{#1\bigr\}}
\newcommand{\Bcb}[1]{\Bigl\{#1\Bigr\}}
\newcommand{\bce}[1]{\bigl\lceil#1\bigr\rceil}

\newcommand{\labs}[1]{\left\lvert#1\right\rvert}
\newcommand{\babs}[1]{\bigl\lvert#1\bigr\rvert}

    \colorlet{RED}{red}
    \newcommand{\tccheck}{}

\DeclareMathOperator*{\argmax}{argmax}
\newcommand{\dif}{\,\mathrm{d}}

\DeclareMathOperator*{\gft}{GFT}
\newcommand{\GFT}{\mathrm{GFT}}
\newcommand{\lip}{Lipschitz}
\newcommand{\leb}{Lebesgue}
\newcommand{\fracc}[2]{#1/#2}
\newcommand{\s}{\subset}
\newcommand{\m}{\setminus}
\newcommand{\iop}{\infty}

\newcommand{\xs}{x^\star}
\newcommand{\ps}{p^\star}
\newcommand{\qs}{q^\star}
\newcommand{\Rs}{R^\star}

\newcommand{\fA}{\sA}
\newcommand{\fG}{\sG}

\newcommand{\fP}{\sP}

\newcommand{\sA}{\mathscr{A}}

\newcommand{\sG}{\boldsymbol{\mathscr{G}}}
\newcommand{\sP}{\mathscr{P}}
\newcommand{\sS}{\mathscr{S}}
\newcommand{\slf}{\tilde{\mathscr{f}}}
\newcommand{\slg}{\mathscr{g}}
\newcommand{\slh}{\mathscr{h}}
\newcommand{\bmu}{\boldsymbol{\mu}}
\renewcommand{\hat}[1]{\widehat{#1}}
\newcommand{\sPiid}{\sP_{\mathrm{iid}}}
\newcommand{\sPiv}{\sP_{\mathrm{iv}}}
\newcommand{\sPbd}{\sP_{\mathrm{bd}}^M}
\newcommand{\sPivbd}{\sP_{\mathrm{iv+bd}}^M}
\newcommand{\sPadv}{\sP_{\mathrm{adv}}}
\newcommand{\ber}{\mathrm{Ber}}

\newcommand{\bx}{\boldsymbol{\bx}}
\newcommand{\ks}{k^\star}
\newcommand{\uno}{(\mathrm{I})}
\newcommand{\due}{(\mathrm{II})}
\newcommand{\tre}{(\mathrm{III})}
\newcommand{\quattro}{(\mathrm{IV})}
\newcommand{\calR}{\mathcal{R}}
\renewcommand{\tilde}{\widetilde}
\newcommand{\ind}{\mathbb{I}}

\SetKwComment{Comment}{$\triangleright$\ }{}

\newtheorem{theorem}{Theorem}
\newtheorem{lemma}{Lemma}
\newtheorem*{theorem*}{Theorem}

\theoremstyle{definition}
\newtheorem{definition}{Definition}

\title{Bilateral Trade: A Regret Minimization Perspective\footnote{A preliminary version of this paper appeared at the 22nd ACM Conference on Economics and Computation \citep{cesabianchi2021regret}.} }

\author[1]{Nicol\`o Cesa-Bianchi}
\author[2,3]{Tommaso Cesari}
\author[1,4]{Roberto Colomboni}
\author[5]{Federico Fusco}
\author[5]{Stefano Leonardi}

\affil[1]{Universit\`a degli Studi di Milano, Milano, Italy}
\affil[2]{Toulouse School of Economics (TSE), Toulouse, France}
\affil[3]{Artificial and Natural Intelligence Toulouse Institute (ANITI), Toulouse, France}
\affil[4]{Istituto Italiano di Tecnologia, Genova, Italy}
\affil[5]{Sapienza Universit\`a di Roma, Roma, Italy}

\begin{document}

\maketitle

\begin{abstract}
 Bilateral trade, a fundamental topic in economics, models the problem of intermediating between two strategic agents, a seller and a buyer, willing to trade a good for which they hold private valuations. 
In this paper, we cast the bilateral trade problem in a regret minimization framework over $T$ rounds of seller/buyer interactions, with no prior knowledge on their private valuations.
Our main contribution is a complete characterization of the regret regimes for fixed-price mechanisms with different feedback models and private valuations, using as a benchmark the best fixed-price in hindsight. More precisely, we prove the following tight bounds on the regret:
\begin{itemize}
    \item $\Theta\brb{ \sqrt{T} }$ for full-feedback (i.e., direct revelation mechanisms).
    \item $\Theta\brb{ T^{2/3} }$ for realistic feedback (i.e., posted-price mechanisms) and independent seller/buyer valuations with bounded densities.
    \item  $\Theta(T)$ for realistic feedback and seller/buyer valuations with bounded densities.
    \item  $\Theta(T)$ for realistic feedback and independent seller/buyer valuations.
    \item $\Theta(T)$ for the adversarial setting.
\end{itemize}  
\end{abstract}

\section{Introduction \tccheck}

In the bilateral trade problem, two strategic agents ---a seller and a buyer--- wish to trade some good. They both privately hold a personal valuation for it and strive to maximize their respective quasi-linear utility. 
The burden of designing a mechanism to reach an agreement is usually delegated to a third party. This scenario arises naturally in many internet applications, such as ridesharing systems like Uber or Lyft, where trades between sellers (drivers) and buyers (riders) are managed by a mechanism designed by the platform.

In general, an ideal mechanism for the bilateral trade problem would optimize the efficiency, i.e., the social welfare resulting by trading the item, while enforcing incentive compatibility (IC) and individual rationality (IR). 
The assumption that makes a two-sided mechanism design more complex than its one-sided counterpart is budget balance (BB): the mechanism cannot subsidize or make a profit from the market. 
Unfortunately, as Vickrey observed in his seminal work \citep{Vickrey61}, the optimal incentive compatible mechanism maximizing social welfare for bilateral trade may not be budget balanced.

A more general result due to Myerson and Satterthwaite \citep{MyersonS83} shows that there are some problem instances where a fully efficient mechanism for bilateral trade that satisfies IC, IR, and BB does not exist. 
This impossibility result holds even if prior information on the buyer and seller's valuations is available, the truthful notion is relaxed to Bayesian incentive compatibility (BIC), and the budget balance constraint is loosened to weak budget balance (WBB).
To circumvent this obstacle, a long line of research has focused on the design of approximating mechanisms that satisfy the above requirements while being nearly efficient. 

These approximation results build on a Bayesian assumption: seller and buyer valuations are drawn from two distributions, which are both known to the mechanism designer. 
Although in some sense necessary ---without any information on the priors there is no way to extract any meaningful approximation result  \citep{Duetting20}--- this assumption is unrealistic.  
Following a recent line of research \citep{Cesa-BianchiGM15,LykourisST16,DaskalakisS16}, in this work we study this basic mechanism design problem in a regret minimization setting.
Our goal is bounding the total loss in efficiency experienced by the mechanism in the long period by learning the salient features of the prior distributions. 

At each time $t$, a new seller/buyer pair arrives. 
The seller has a private valuation $S_t \in [0,1]$ representing the smallest price she is willing to accept in order to trade. 
Similarly, the buyer has a private value $B_t \in [0,1]$ representing the highest price that she will pay for the item. 
The mechanism sets a price $P_t \in [0,1]$ which results in a trade if and only if $S_t \le P_t \le B_t$. 

There are two common utility functions that reflect the performance of the mechanism at each time step: the social welfare, which sums the utilities of the two players after the trade (and remains equal to the seller's valuation if no trade occurs), and the gain from trade, consisting in the net gain in the utilities. 
In formulae, for each price $p\in [0,1]$,
\begin{itemize}
    \item \textbf{Social Welfare}: $ \mathrm{SW}_t(p) := \mathrm{SW} ( p, S_t, B_t ) := S_t + (B_t-S_t) \I \{ S_t \leq p \leq B_t \}$.
     \item \textbf{Gain from Trade}: $
        \gft\nolimits_t(p)
    := 
        \gft(p,S_t,B_t)
    :=
        (B_t-S_t)\I \{ S_t \leq p \le B_t \}$.
\end{itemize}
We begin by investigating the standard assumption in which
$(S_1,B_1),(S_2,B_2),\ldots$ are i.i.d.\ random vectors, supported in $[0,1]^2$, representing the valuations of seller and buyer respectively (stochastic i.i.d.\ setting). 
We also consider the case where $(S_1,B_1),(S_2,B_2),\ldots$ is an arbitrary deterministic process $(s_1,b_1),(s_2,b_2),\dots$ (adversarial setting).

In our online learning framework, we aim at minimizing the \emph{regret} of the mechanism over a time horizon $T$:
\[
    \max_{p \in [0,1]} \exp{\sum_{t=1}^T \gft\nolimits_t(p) - \sum_{t=1}^T \gft\nolimits_t(P_t)} \;.
\]
Note that since $\gft\nolimits_t(P_t) = \mathrm{SW}_t(P_t) -S_t$ and $S_t$ does not depend on the choice of $p$,  gain from trade and social welfare lead to the same notion of regret. 

Hence, the regret is the difference between the expected total performance of our algorithm, which can only \emph{sequentially learn} the distribution,
and our reference benchmark, corresponding to the the best fixed-price strategy assuming \emph{full knowledge} of the distribution. 
Our main goal is to design strategies with asymptotically vanishing time-averaged regret with respect to the best fixed-price strategy or, equivalently, regret sublinear in the time horizon $T$. 

The class of fixed-price mechanisms is of particular importance in bilateral trade as they are simple to implement, clearly truthful, individually rational, budget balanced, and enjoy the desirable property of asking the agents very little information. Moreover, it can be shown that fixed prices are the {\em only} direct revelation mechanisms which enjoy budget balance, dominant strategy incentive compatibility, and ex-post individual rationality \citep{Colini-Baldeschi16}.

To complete the description of the problem, we need to specify the feedback obtained by the mechanism after each sequential round. We propose two main feedback models: 
\begin{itemize}
    \item {\em Full feedback.} In the full-feedback model, the pair $(S_t,B_t)$ is revealed to the mechanism after the $t$-th trading round. The information collected by this feedback model corresponds to {\em direct revelation mechanisms}, where the agents publicly declare their valuations in each round, and the price proposed by the mechanism at time $t$ only depends on past bids.
    \item {\em Realistic feedback.}    In the harder realistic feedback model, only the relative order between $S_t$ and $P_t$ and between $B_t$ and $P_t$ are revealed after the $t$-th round. This model corresponds to {\em posted price mechanisms}, where seller and buyer separately accept or refuse the posted price. The price computed at time $t$ only depends on past bids, and the values $S_t$ and $B_t$ are never revealed to the mechanisms. 
\end{itemize}

\subsection{Overview of our Results \tccheck}

In \cref{sec:adversarial,s:real-feddb,s:full}, we investigate the stochastic setting (under various assumptions), the adversarial setting, and how regret bounds change depending on the quality of the received feedback.
In all cases, we provide matching upper and lower bounds.
In particular, our positive results are constructive: explicit algorithms are given in each case.
More precisely, we present (see \cref{tab:general} for a summary):
\begin{table}
\centering
\begin{tabular}{l|l|l|l|l|l|}
\cline{2-6}
\multicolumn{1}{c|}{}      & \multicolumn{4}{c|}{\textbf{Stochastic (iid)}}                                                                              & \multicolumn{1}{c|}{\textbf{Adversarial}} \\ \cline{2-6} 
\multicolumn{1}{c|}{}      & \multicolumn{1}{c|}{\textbf{iid}} & \multicolumn{1}{c|}{\textbf{+iv}} & \multicolumn{1}{c|}{\textbf{+bd}} & \multicolumn{1}{c|}{\textbf{+iv+bd}} & \multicolumn{1}{c|}{\textbf{adv}}         \\ \hline
\multicolumn{1}{|l|}{\textbf{Full}} & $T^{1/2}$ (Thm~\ref{thm:upper_full})                & $T^{1/2} \vphantom{e^{e^{e^{e}}}}$                & $T^{1/2}$                & $T^{1/2}$ (Thm~\ref{thm:lower-full})                   & $T$ (Thm~\ref{thm:adv-lower})                              \\ \hline
\multicolumn{1}{|l|}{\textbf{Real}} & $T$                      & $T$ (Thm~\ref{thm:lower-real-iv})                      & $T$ (Thm~\ref{thm:lower-real-bd})                      & $T^{2/3} \vphantom{e^{e^{e^{e}}}}$ (Thms~\ref{thm:upper_real}+\ref{thm:lower-real-iv+bd})                   & $T$                              \\ \hline
\end{tabular}
\caption{Our main results for fixed-price mechanisms. 
The rates are both upper and lower bounds.
The slots without references are immediate consequences of the others.}
\label{tab:general}
\end{table}
\begin{itemize}
    \item \cref{alg:followTheBestPrice} (Follow the Best Price) for the full-feedback model achieving a $\cO(T^{1/2})$ regret in the stochastic (iid) setting (\cref{thm:upper_full});
    this rate cannot be improved, not even under some additional natural assumptions (\cref{thm:lower-full}).
    \item \cref{alg:meta} (Scouting Bandits) for the harder realistic-feedback model achieving a $\cO(T^{2/3})$ regret in a stochastic (iid) setting in which the valuations of the seller and the buyer are independent of each other (iv) and have bounded densities (bd) (\cref{thm:upper_real});
    this rate cannot be improved (\cref{thm:lower-real-iv+bd}).
    \item Impossibility results:
    \begin{itemize}
        \item For the realistic-feedback model, if either the (iv) or the (bd) assumptions are dropped from the previous stochastic setting, no strategy can achieve sublinear worst-case regret (\cref{thm:lower-real-iv,thm:lower-real-bd}).
        \item In an adversarial setting, no strategy can achieve sublinear worst-case regret, not even in the simple full-feedback model (\cref{thm:adv-lower}).
    \end{itemize}
\end{itemize}
In \cref{s:wbb}, we depart from the budget balance setting in which the learner posts the same price to both the seller and the buyer.
We consider a weak budget balance (WBB) setting in which (possibly) distinct prices $p \le p'$ can be posted, $p$ to the seller, and $p'$ to the buyer.
We design \cref{alg:etc} (Scouting Blindits) which, leveraging the higher amount of information available in the WBB setting, can break the linear lower bound of \cref{thm:lower-real-bd} (\cref{t:wbb}).

Finally, in \cref{s:one-bit} we investigate the case in which the feedback is limited to one single bit (i.e., whether or not a trade occurred), showing a striking difference between the BB and WBB cases.

\subsection{Technical Challenges \tccheck}

In this section, we sum up the technical challenges for various instances of our problem.
    
\paragraph{Full feedback.} The full-feedback model fits nicely in the learning with expert advice framework \citep{Nicolo06}.
Each price $p \in [0,1]$ can be viewed as an expert, and the revelation of $S_t$ and $B_t$ allows the mechanism to compute $\gft\nolimits_t(p)$ for all $p$, including the mechanism's own reward $\gft\nolimits_t(P_t)$. 
A common approach to reduce the cardinality of a continuous expert space is to assume some regularity (e.g., \lip{}ness) of the reward function, so that a finite grid of representative prices can be used. This approach yields a $\widetilde{\cO}(\sqrt{T})$ bound under density boundedness assumptions on the joint distribution of the seller and the buyer. 
By exploiting the structure of the reward function $\E \bsb{\gft\nolimits_t(\cdot)}$, we obtain a better regret bound (by a log factor) without any assumptions on the distribution (other than iid).
In \cref{e:decomp-one}, we show how to decompose the expression of the expected gain from trade in pieces that can be quickly learned via sampling. The full feedback received in each new round is used to refine the estimate of the actual gain from trade as a function of the price, while the posted prices are chosen so to maximize it. A {\em Follow the Leader} strategy is shown to achieve a $\cO(\sqrt{T})$ bound in the stochastic (iid) setting (\cref{thm:upper_full}).
This holds for arbitrary joint distributions of the seller and the buyer.
In particular, even when the buyer and seller have a correlated behavior.
The main issue for the lower bound is that the (expected) gain from trade cannot be chosen arbitrarily: we can only control its shape indirectly as a function of the distribution of the seller/buyer pair.
By designing a suitable family of such distributions, we build a reduction showing that the full-feedback bilateral trade problem is harder than a corresponding $2$-action partial monitoring game with a known $\Omega\brb{ \sqrt{T} }$ lower bound (\cref{thm:lower-full}).

\paragraph{Realistic Feedback.} 
Here, at each time $t$, only $\I\{S_t \le P_t\}$ and $\I\{P_t \le B_t\}$ are revealed to the learner. 
In contrast to the full-feedback model, this is not enough to reconstruct the gain from trade $\GFT_t$ at time $t$: if the trade does not occur, it is unclear which prices would have resulted in a trade.
Moreover, in contrast to bandit problems \citep{Nicolo06}, this feedback is not even enough to determine $\GFT_t(P_t)$: if the trade occurs, there is no way to infer the difference $B_t-S_t$ directly.
Thus, we cannot directly rely on known bandits tools to tackle the two competing goals of estimating the underlying distributions (exploration) while optimizing the estimated gain from trade (exploitation).
Instead, we show how to decompose the expected gain from trade at price $p$ into a \emph{global} part that can be quickly estimated by uniform sampling on the $[0,1]$ interval, and a \emph{local} part that can be learned by posting $p$. 
\cref{thm:upper_real} shows that our \cref{alg:meta} (Scouting Bandits) can take advantage of this decomposition by relying on any bandit algorithm to learn the local part of the expected gain from trade.
We derive a sublinear regret of $ \cO(T^{2/3})$ in a stochastic (iid) setting in which the valuations of the seller and the buyer are independent of each other (iv) and have bounded densities (bd).
The lower bound presents challenges similar to those of the full-feedback model, with additional hurdles due to the specific nature of the realistic feedback that lead to a harder $\Omega(T^{2/3})$ rate (\cref{thm:lower-real-iv+bd}).
Dropping the (iv) assumption leads to a pathological \emph{lack of observability} phenomenon, in which it is impossible to distinguish between two scenarios with significantly different optimal prices (\cref{thm:lower-real-bd}).
Dropping the (bd) assumption leads to a {\em needle in a haystack}, a different pathological phenomenon in which all prices but one suffer a high regret, and it is essentially impossible to find this optimal price among a continuum of suboptimal prices (\cref{thm:lower-real-iv}).

\paragraph{Adversarial setting.}
Here, the valuations of the buyer and the seller form an arbitrary deterministic process generated by an oblivious adversary.
In this setting, learning is impossible.
Indeed, using a construction inspired by the Cantor ternary set, we show that even under a full-feedback model, no strategy can lead to a sublinear worst-case regret
(\cref{thm:adv-lower}).

\paragraph{Lower Bound Techniques.}
Due to their technical nature, the proofs of the lower bounds (\Cref{thm:lower-full,thm:lower-real-bd,thm:lower-real-iv,thm:lower-real-iv+bd,thm:adv-lower}) are only sketched in the main text. 
Detailed versions of all of them are provided in the Appendix where, inspired by partial monitoring, we develop a general setting for sequential games that subsumes, in particular, all instances of our bilateral trade problem (Appendix~\ref{s:model-appe}).
Within this setting, we then build reductions by mapping instances of our problem to other known partial monitoring games.
These reductions rely on two key lemmas, introduced in Appendix~\ref{s:keylemmas}:
our Embedding and Simulation lemmas
(\cref{l:embedding,l:simulation}) are useful tools to manipulate rewards and feedbacks, allowing to build chains of progressively easier games leading to games with known minimax regrets.

\subsection{Further Related Work \tccheck}
The study of the bilateral trade problem dates back to the already mentioned seminal works of Vickrey \citep{Vickrey61} and Myerson and Satterthwaite \citep{MyersonS83}.
A more recent line of research focuses on Bayesian mechanisms that achieve the IC, BB, and IR requirements while approximating the optimal social welfare or the gain form trade. Blumrosen and Dobzinski~\citep{BlumrosenD14} proposed the \emph{median mechanism} that sets a posted price equal to the median of the seller distribution and shows that this mechanism obtains an approximation factor of $2$ to the optimal social welfare. Subsequent work by the same authors \citep{BlumrosenD16} 
improved the approximation guarantee to $e/(e-1)$ through a randomized mechanism whose prices depend on the seller distribution in a more intricate way. 
In~\citep{Colini-Baldeschi16} it is demonstrated that all DSIC mechanisms that are BB and IR  must post a fixed price to the buyer and to the seller. 
The same result has been previously proven under stronger assumptions in \citep{hagerty1987robust}.

In a different research direction aimed to characterize the information theoretical requirements of two-sided markets mechanisms, \citep{Duetting20} show that setting the price equal to a single sample from the seller distribution gives a $2$-approximation to the optimal social welfare.  
In a parallel line of work, the harder objective of approximating the \emph{gain from trade} has been considered. 
An asymptotically tight fixed-price $O\big(\log\frac 1r\big)$ approximation  bound is also achieved in \citep{Colini-Baldeschi17}, with $r$ being the probability that a trade happens (i.e., the value of the buyer is higher than the value of the seller).  A BIC $2$-approximation of the second best with a simple mechanism is obtained in \citep{BrustleCWZ17}. 
    
In the following, we discuss the relationship between the approximation results mentioned above and the regret analysis we develop in this work that compares online learning mechanisms against the best ex-ante fixed-price mechanism. First of all, in the realistic feedback setting, the approximation mechanisms for bilateral trade cannot be easily implemented.  For example, the single sample $2$-approximation to the optimal social welfare \citep{Duetting20} requires multiple rounds of interaction in order to obtain, approximately, a random sample from the distribution.  The median mechanism of \citep{BlumrosenD14} requires an even larger number of rounds in order to estimate the median of the seller distribution.  
Furthermore, here we note that these two more demanding approaches may yield worst performances than the best ex-ante fixed price\footnote{Consider a seller with value $\varepsilon>0$ or $0$ with equal probability and a buyer with value $1$. The best fixed price has welfare of $1$. For small $\varepsilon$, the median and the sample mechanism, respectively, obtains a welfare close to $1/2$ and $3/4$.}. This implies that there are instances where our online learning approach converges to a mechanism that is strictly better than the median or sample mechanisms, even assuming they have full knowledge of the underlying distributions.

There is a vast body of literature on regret analysis in (one-sided) dynamic pricing and online posted price auctions ---see, e.g., the excellent survey published by \citep{den2015dynamic} and the tutorial slides by \citep{SZ15}. In their seminal paper, Kleinberg and Leighton prove a $O(T^{2/3})$ upper bound (ignoring logarithmic factors) on the regret in the adversarial setting \citep{kleinberg2003value}. Later works show simultaneous multiplicative and additive bounds on the regret when prices have range $[1,h]$ \citep{blum2004online,blum2005near}. These bounds have the form $\varepsilon\,G_T^{\star} + O\big((h\ln h)/\varepsilon^2\big)$ ignoring $\ln\ln h$ factors, where $G_T^{\star}$ is the total revenue of the optimal price $p^{\star}$. Recent improvements on these results prove that the additive term can be made $O(p^{\star}\brb{ \ln h)/\varepsilon^2 }$, where the linear scaling is now with respect to the optimal price rather than the maximum price $h$ \citep{bubeck2017online}. Other variants consider settings in which the number of copies of the item to sell is limited \citep{agrawal2014bandits,babaioff2015dynamic,badanidiyuru2013bandits}, buyers act strategically in order to maximize their utility in future rounds \citep{amin2013learning,devanur2014perfect,mohri2014optimal,drutsa2018weakly}, or there are features associated with the goods on sale \citep{DBLP:journals/mansci/CohenLL20}. In the stochastic setting, previous works typically assume parametric \citep{broder2012dynamic}, locally smooth \citep{kleinberg2003value}, or piecewise constant demand curves \citep{cesa2019dynamic,den2020discontinuous}.

\section{The Bilateral Trade learning protocol \tccheck} 
\label{s:bil-tr-model}
In this section, we present the learning protocol for the sequential problem of bilateral trade (see Learning Protocol~\ref{a:learning-model}).
We recall that the reward collected from a trade is the gain from trade, defined for all $p,s,b \in [0,1]$, by 
\[
    \gft(p,s,b) 
:= 
    (b-s) \I \{s\le p \le b\} \;.
\]
{
\renewcommand*{\algorithmcfname}{Learning Protocol}
\begin{algorithm}
\For
{%
    time $t=1,2,\ldots$
}
{
    a new seller/buyer pair arrives with (hidden) valuations $(S_t,B_t) \in [0,1]^2$\;
    the learner posts a price $P_t \in [0,1]$\;
    the learner receives a (hidden) reward $ \GFT_t(P_t) \in [0,1]$, where $\GFT_t(\cdot):= \gft (\cdot, S_t,B_t)$\;
    feedback $Z_t$ is revealed\;
}
 \caption{Bilateral Trade}
 \label{a:learning-model}
\end{algorithm}
}

At each time step $t$, a seller and a buyer arrive with privately held valuations: $S_t\in [0,1]$ for the seller and $B_t \in [0,1]$ for the buyer.
The learner then posts a price $P_t \in [0,1]$ and a trade occurs if and only if $S_t \le P_t \le B_t$.
When this happens, the learner gains a reward $\gft (P_t, S_t,B_t)$ but, instead of observing this reward, the learner only observes the feedback $Z_t$.
The nature of the sequence of valuations $(S_1,B_1),(S_2,B_2),\ldots$ and feedback $Z_1,Z_2,\ldots$ depends on the specific instance of the problem and is described below.

The goal of the learner is to determine a strategy $\alpha$ generating the prices $P_1,P_2, \ldots$ (as in Learning~Protocol~\ref{a:learning-model}) achieving sublinear \emph{regret}
\[
    R_T(\alpha)
:=
    \max_{p \in [0,1]} \E\lsb{\sum_{t=1}^T \gft ( p, S_t,B_t)  - \sum_{t=1}^T \gft ( P_t, S_t,B_t)}
    \;,
\]
where the expectation is taken with respect to the sequence of buyer and seller valuations, and (possibly) the internal randomization of $\alpha$.
To lighten the notation, we denote by $\ps$ (one of) the $p\in[0,1]$ maximizing the previous expectation. 
Such a $\ps$ always exists (for a proof of this fact, see Appendix~\ref{s:existenceMax}).

We now introduce several instances of bilateral trade, depending on the type of the received feedback and the nature of the environment.

\subsection{Feedback} 

\begin{description}
    \item[Full feedback:] the feedback $Z_t$ received at time $t$ is the entire seller/buyer pair $(S_t,B_t)$;
    in this setting, the seller and the buyer reveal their valuations at the end of a trade.
    \item[Realistic feedback:] the feedback $Z_t$ received at time $t$ is just the pair $\brb{ \I\{S_t\le P_t\}, \, \I\{P_t \le B_t \} }$;
    in this setting, the seller and the buyer only reveal whether or not they accept the trade at price $P_t$.
\end{description}

\subsection{Environment} 
\begin{description}
    \item[Stochastic (iid):] 
    $(S_1,B_1),(S_2,B_2),\ldots$ is an i.i.d.\ sequence of seller/buyer pairs,
    while $S_t$ and $B_t$ could be (arbitrarily) correlated.
We will also investigate the (iid) setting under the following further assumptions.
    \begin{description}
        \item[Independent valuations (iv):]
        $S_1$ and $B_1$ are independent of each other.
        \item[Bounded density (bd):] 
        $(S_1,B_1)$ admits a joint density bounded by some constant $M$.
    \end{description}
    \item[Adversarial (adv):] 
    $(S_t,B_t)_{t\in\N}$ is an arbitrary deterministic sequence $(s_t,b_t)_{t\in\N} \s [0,1]^2$.
\end{description}

\section{The Decomposition Lemma}
\label{s:decomp}

In this section, we present a key lemma whose purpose is to decompose the gain from trade into terms that depend only on the outcome of yes/no questions.
This result allows leveraging DKW inequalities in the proofs of our upper bounds. 
Moreover, it shows how to use the limited feedback available to reconstruct the expected gain from trade in the realistic feedback settings.
Furthermore, it leads to an easy proof of the existence of the maximum of the expected gain from trade, under no assumptions on the seller and buyer distributions.
We defer the proofs of the results of this section to Appendix~\ref{s:decomp-appe}.

\begin{restatable}[Decomposition lemma]{lemma}{decomplemma}
\label{l:decmp}
Fix any price $p\in [0,1]$. Then, for any $s,b\in [0,1]$,
\begin{equation}
    \label{e:decomp-zero}
    \GFT(p,s,b)
=
    \int_{[p,1]} \I[s \le p \le \lambda \le b] \dif\lambda
    +
    \int_{[0,p]} \I[s \le \lambda \le p \le b] \dif\lambda \;.
\end{equation}
Furthermore, let $S$ and $B$ be two $[0,1]$-valued random variables:
\begin{itemize}
    \item Then
    \begin{equation}
    \label{e:decomp-one}
        \E \bsb{ \GFT(p,S,B) }
    =
        \int_{[p,1]} \P[S\le p \le \lambda \le B] \dif\lambda
        +
        \int_{[0,p]} \P[S\le \lambda \le p \le B] \dif\lambda \;.
    \end{equation}
    \item If $U$ is uniform on $[0,1]$ and independent of $(S,B)$, then
    \begin{equation}
        \label{e:decomp-two}
        \E \bsb{ \GFT(p,S,B) }
    =
        \P[S \le p \le U \le B]
        +
        \P[S \le U \le p \le B] \;.
    \end{equation}
    \item If $U$ is uniform on $[0,1]$ and $S,B,U$ are independent, then
    \begin{equation}
        \label{e:decomp-three}
        \E \bsb{ \GFT(p,S,B) }
    =
        \P[S \le p]\P[p \le U \le B]
        +
        \P[p \le B]\P[S \le U \le p] \;.
    \end{equation}
    \item If $U$ is uniform on $[p,1]$, $V$ is uniform on $[0,p]$ and $(U,V)$ is independent of $(S,B)$, then
    \begin{equation}
        \label{e:decomp-four}
        \E\bsb{ \GFT(p,S,B) } 
    = 
        \E\bsb{ (1-p) \I \{S \le p \le U \le B \} } + \E\bsb{ p \I \{S \le V \le p \le B \} } \;.
    \end{equation}
\end{itemize}
\end{restatable}

We now present a corollary of our Decomposition lemma relating the regularity of the distributions (specifically, the boundedness of the densities) to the regularity of the expected gain from trade (i.e., its \lip{}ness).

\begin{restatable}{corollary}{decompcorollarylip}
\label{l:decomp-lip}
If $S$ and $B$ are $[0,1]$-valued random variables such that $(S,B)$ admits joint density $f$ bounded above by some constant $M$, then $\E\bsb{\gft(\cdot,S,B)}$ is $4M$-Lipschitz.
\end{restatable}

\section{Full-Feedback Stochastic (iid) Setting \tccheck}
\label{s:full}

We begin by considering the full-feedback model (corresponding to direct revelation mechanisms) in a stochastic environment, where the seller/buyer pairs $(S_1,B_1), (S_2,B_2), \ldots$ are $[0,1]^2$-valued i.i.d.\ random vectors, without any further assumptions on their common distribution (in particular, $S_1$ and $B_1$ could be arbitrarily correlated). 
Here, at the end of each round, sellers and buyers declare their actual valuations to the learner. 
The incentive-compatibility is guaranteed by the fact that the posted prices do not depend on the declared valuations at each specific round, but only on past ones, so that there is no point in misreporting.

In \cref{sec:DKW}, we show that a {\em Follow the Leader} approach, which we call Follow the Best Price (FBP, \cref{alg:followTheBestPrice}), achieves a $O\brb{ \sqrt{T} }$ upper bound.
In \cref{sec:lower_bound_indep}, we provide a matching $\Omega\brb{ \sqrt T }$ lower bound rate.

\subsection{Follow the Best Price (FBP) \tccheck}
\label{sec:DKW}

We begin by presenting our Follow the Best Price (FBP) algorithm. It consists in posting the best price with respect to the samples that have been observed so far. Notably, it does not need preliminary knowledge of the time horizon $T$.

\begin{algorithm}
    \SetKwInput{kwInit}{initialization}
    \kwInit{let $P_1 \gets 1/2$;}
    \For{$t=1,2, \ldots$}
    {
        post price $P_t$\;
        receive feedback $(S_t,B_t)$\;
        pick $P_{t+1} \in \argmax_{p \in [0,1]} \frac 1 t \sum_{i=1}^t \GFT(p, S_i, B_i);$
    }
    \caption{Follow the Best Price (FBP)}
    \label{alg:followTheBestPrice}
\end{algorithm}

For each time $t$, given $(S_1,B_1), \dots, (S_t,B_t)$, one can reconstruct the gain from trade function $\gft(\cdot, S_i, B_i)$ at each time step $i \le t$ and compute (one of) the best price(s)
$
    P_{t+1} \in \argmax_{p \in [0,1]} \frac 1 t \sum_{i=1}^{t} \GFT(p,S_i,B_i) 
$.
Note that $\frac 1 t \sum_{i=1}^{t} \GFT(\cdot, S_i,B_i)$ is a step-wise constant function that attains its maximum at one of the observed sellers' valuations $S_1\ldots,S_t$.\footnote{By the symmetry of the problem, the maximum is also attained at one of the buyers' valuations.}
Hence, even a naive enumeration approach is computationally efficient.
On a technical note, prices $P_{t+1}$ should be defined in a measurable way in order for the regret to be well-defined.
For example, this can be done by picking $P_{t+1}$ as the $S_i$ with the smallest index among all the $S_j \in \argmax_{p\in [0,1]} \frac 1 t \sum_{t=1}^t \GFT(p,S_i,B_i)$. (For other ideas on how to break ties in a measurable way, see \cite[Section~2.4]{cesari2021nearest}.)

The main idea of the analysis of \cref{alg:followTheBestPrice} is to show that the approximation of the expected gain from trade with its empirical means is uniform over all possible seller/buyer distributions and prices. 
A possible way to achieve this result could be through a pseudo-dimension argument (e.g., see \cite[Introduction and Theorem~5]{li2001improved}). 
However, this approach requires subtle measurability considerations. 
In contrast, we will show that one could get around these measurability issues altogether by leveraging our Decomposition lemma (\cref{l:decmp}) and a bivariate DKW inequality (\cref{t:vc}). 
Our approach also yields constants that ---while still fairly high (see discussion after the proof)--- are significantly better than those guaranteed by pseudo-dimension results.
Finally, this presentation will be helpful to get the reader acquainted with the techniques that appear in the following sections for the realistic setting.

\begin{theorem}
\label{thm:upper_full}
    In the full-feedback stochastic (iid) setting, the regret of Follow the Best Price 
    satisfies, for all horizons $T$,
    \[
        R_T(\emph{FBP}) 
    \le 
        \frac 1 2 + c \sqrt{T-1}
        \;.
    \]
    where $c \in (0, 1144240)$ is a universal constant.
\end{theorem}
\begin{proof}
Without loss of generality, assume that $T\ge 2$. 
Fix any $t\in [T-1]$. 
For any $p\in[0,1]$ define the random variable
\[
    H_t (p)
:=
    \frac 1 t \sum_{i=1}^t \GFT_i(p) - \E \bsb{ \GFT_1(p) } \;,
\]
where we recall that $\GFT_i(p) := \GFT(p, S_i,B_i)$.
By definition of $P_{t+1}$ and the independence of $P_{t+1}$ and $(S_{t+1},B_{t+1})$, we have that
\begin{multline*}
    \E \bsb{ \GFT_{t+1}(\ps) } - \E \bsb{ \GFT_{t+1}(P_{t+1}) }
\le
    \E \lsb{ \frac 1 t \sum_{i=1}^t \GFT_i(P_{t+1}) } - \E \bsb{ \GFT_{t+1}(P_{t+1}) }
\\
=
    \E \lsb{ \frac 1 t \sum_{i=1}^t \GFT_i(P_{t+1})  - \E \bsb{ \GFT_{t+1}(P_{t+1}) \mid P_{t+1} } }
=
    \E \bsb{ H_t (P_{t+1}) }
=:
    (*)\;.
\end{multline*}
Then, by the Decomposition lemma \eqref{e:decomp-zero}-\eqref{e:decomp-one}, we get
\begin{align}
    H_t (P_{t+1})
& \le
    \sup_{p \in [0,1]} \lrb{ \frac 1 t \sum_{i=1}^t \GFT_i(p) - \E \bsb{ \GFT_1(p) } }
\\
& =
    \sup_{p \in [0,1]} \left( 
        \frac 1 t \sum_{i=1}^t  \lrb{ \int_{[p,1]} \I [ S_i \le p \le \lambda \le B_i ] \dif \lambda 
        + 
        \int_{[0,p]} \I [ S_i \le \lambda \le p \le B_i ] \dif \lambda 
        }
        \right.
        \nonumber
\\
&
    \qquad \qquad \qquad
    - \left. \lrb{
        \int_{[p,1]} \P [ S_1 \le p \le \lambda \le B_1 ] \dif\lambda 
        +
        \int_{[0,p]} \P [ S_1 \le \lambda \le p \le B_1 ] \dif \lambda 
        }
        \right)
        \nonumber
\\
& =
    \sup_{p \in [0,1]} \left( 
        \int_{[0,p]} \lrb{ \frac 1 t \sum_{i=1}^t \I\{S_i \le \lambda, -B_i \le -p\} - \P [ S_1 \le \lambda, -B_1 \le -p ] } \dif \lambda 
        \right.
        \nonumber
\\
&
    \qquad \qquad \qquad
    + \left.
        \int_{[p,1]} \lrb{ \frac 1 t \sum_{i=1}^t \I\{S_i \le p, -B_i \le -\lambda\} - \P[S_1 \le p, -B_1 \le - \lambda] } \dif \lambda 
        \right)
        \nonumber
\\
& \le
    2 \sup_{x,y \in \R} \labs{
    \frac 1 t \sum_{i=1}^t \I\{S_i \le x, -B_i \le y \} - \P[S_1 \le x, -B_1 \le y]
    }
    \;.
    \label{e:proof-full-info}
\end{align}
Letting $m_0, c_1, c_2$ as in \cref{t:vc}, $\e_t := \sqrt{m_0/t}$, taking expectations to the left and right hand side of \cref{e:proof-full-info}, and applying the bivariate DKW inequality (\cref{t:vc}), we get
\begin{align}
    (*) 
&
    \le
    \E \lsb{
    2 \sup_{x,y \in \R} \labs{
    \frac 1 t \sum_{i=1}^t \I\{S_i \le x, -B_i \le y \} - \P[S_1 \le x, -B_1 \le y]
    }
    }
    \nonumber
\\
& \le
    2 \e_t + 2 \int_{[\e_t, 1]} \P\lsb{ \sup_{x,y \in \R} \labs{
    \frac 1 t \sum_{i=1}^t \I\{S_i \le x, -B_i \le y \} - \P[S_1 \le x, -B_1 \le y]
    } > \e } \dif \e
    \label{e:proof-improve-constants}
\\
& \le
    2 \e_t + 2 \int_{\e_t}^1 c_1 \operatorname{exp}\brb{ -c_2 t \e^2 } \dif \e
  \le
    2 \e_t + \frac{c_1}{\sqrt{c_2 t}} \int_{0}^\infty e^{-u} u^{-1/2} \dif u 
= 
    \lrb{ 2\sqrt{m_0} + c_1 \sqrt{\frac{\pi}{c_2}} } \frac{1}{\sqrt{t}} \;. 
    \nonumber
\end{align}
Being $t$ arbitrary, using the fact that $\sum_{t=1}^{T-1} t^{-1/2} \le 2\sqrt{T-1}$, and letting $c:= 2 \lrb{ 2\sqrt{m_0} + c_1 \sqrt{\frac{\pi}{c_2}} } < 1144265$, we have that
\[
    R_T(\text{FBP}) 
\le 
    \frac 1 2 + \sum_{t=1}^{T-1} \Brb{ \E \bsb{ \GFT_{t+1}(\ps) } - \E \bsb{ \GFT_{t+1}(P_{t+1}) } } 
\le 
    \frac 1 2 + \frac{c}{2} \sum_{t=1}^{T-1} \frac{1}{\sqrt{t}} 
= 
    \frac 1 2 + c \sqrt{T-1} \;,
\]
which concludes the proof.
\end{proof}

The loose bound on the constant $c$ appearing in the statement is due to the (likely suboptimal) large constants appearing in \cref{t:vc}:
any improvement on the bivariate DKW inequality would result in an improvement of this constant.
For example, it is conjectured \cite[Section~5]{naaman2021tight} that the tightest bound for the bivariate DKW inequality is (with the same notation as \cref{t:vc}), for all $m\in \N$ and $\e>0$,
$
    \P \lsb{ \sup_{x,y \in \R} \labs{ \frac 1 m \sum_{k=1}^m \I\{ X_k \le x, Y_k \le y \} - \P[X_1 \le x, Y_1 \le y] } > \e }
    \le
    4 \operatorname{exp}\brb{-2m\e^2}
$.
If this was the case, we could replace \cref{e:proof-improve-constants} with 
\[
    (*)
\le
    2 \int_{[0, 1]} \P\lsb{ \sup_{x,y \in \R} \labs{
    \frac 1 t \sum_{i=1}^t \I\{S_i \le x, -B_i \le y \} - \P[S_1 \le x, -B_1 \le y]
    } > \e } \dif \e
\le
    2\sqrt{2\pi} \frac{1}{\sqrt{t}} \;.
\]
leading to a significantly smaller constant $c := 2\cdot 2 \sqrt{2\pi} < 11$.

\subsection{\texorpdfstring{$\sqrt{T}$}{sqrt(T)} Lower Bound (iv+bd) \tccheck}
\label{sec:lower_bound_indep}

In this section, we show that the upper bound on the minimax regret we proved in \cref{sec:DKW} is tight.
No strategy can beat the $\cO\brb{\sqrt{T}}$ rate when the seller/buyer pair $(S_t,B_t)$ is drawn i.i.d. from an unknown fixed distribution, even under the further assumptions that the valuations of the seller and buyer are independent of each other and have bounded densities.
For a full proof of the following theorem, see Appendix~\ref{s:lower-full}.

\begin{theorem}
\label{thm:lower-full}
In the full-feedback 
model, for all horizons $T$, the minimax regret $\Rs_T$ satisfies 
\[
    \Rs_T 
:=
    \inf_{\alpha} \sup_{(S_1,B_1) \sim \cD} R_T(\alpha)
=
    \Omega \brb{ \sqrt{T} } \;,
\]
where $c \ge 1/\brb{ 8\sqrt{2\pi} }$, the infimum is over all of the learner's strategies $\alpha$, and the supremum is over all distributions $\cD$ of the seller/buyer pair such that:
\begin{itemize}
    \item[\emph{(iid)}] $(S_1,B_1),(S_2,B_2),\ldots \sim \cD$ is an i.i.d.\ sequence.
    \item[\emph{(iv)}] $S_1$ and $B_1$ are independent of each other.
    \item[\emph{(bd)}] $(S_1,B_1)$ admits a joint density bounded by $M\ge4$.
\end{itemize}
\end{theorem}
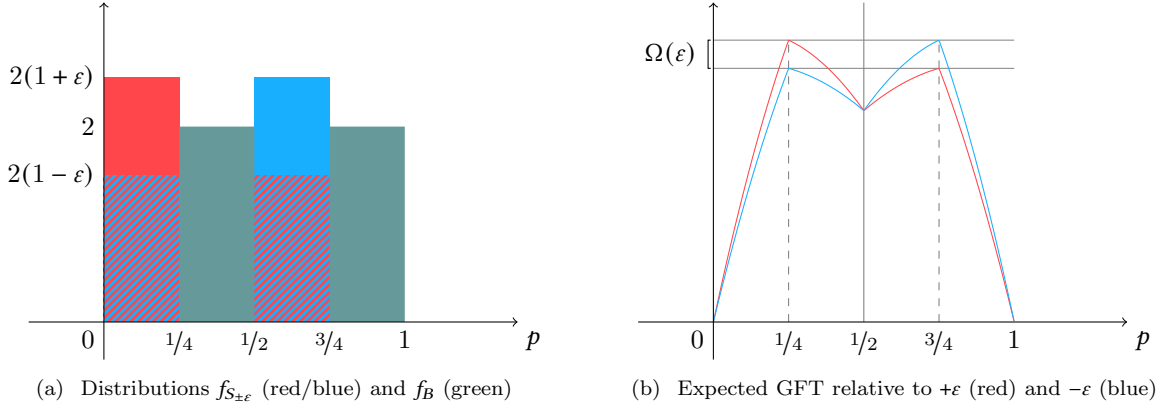
\begin{figure}
    \centering
    \subfigure[\label{f:sqrt-lower-a} Distributions $f_{S_{\pm \e}}$ (red/blue) and $f_B$ (green)]
    {
    \begin{tikzpicture}
    \def\k{0.5}
    \def\stretchY{0.65}
    \definecolor{myblue}{RGB}{25,175,255}
    \definecolor{myred}{RGB}{255,70,75}
    \definecolor{mygreen}{RGB}{64,128,128}
    % plots
    \fill[myred] (0, 0) rectangle ({2*\k}, {10*\k*\stretchY});
        \begin{scope}
        \clip (0,0) rectangle ({2*\k}, {6*\k*\stretchY});
        \foreach \x in {-20,...,20}
        {
            \draw[myblue,line width=1pt] ({-0.1-0.1*\x}, {0}) -- ({3.9-0.1*\x}, {4});
        }
        \end{scope}
    \fill[myblue] ({4*\k},0) rectangle ({6*\k}, {10*\k*\stretchY});
        \begin{scope}
        \clip ({4*\k}, 0) rectangle ({6*\k}, {6*\k*\stretchY});
        \foreach \x in {-40,...,20}
        {
            \draw[myred,line width=1pt] ({-0.1-0.1*\x}, {0}) -- ({5.9-0.1*\x}, {6});
        }
        \end{scope}
    \fill[mygreen!80] 
        ({2*\k}, 0) rectangle ({4*\k}, {8*\k*\stretchY})
        ({6*\k}, 0) rectangle ({8*\k}, {8*\k*\stretchY})
    ;
    % axes
    \draw[->] ({-2*\k}, {0*\k}) -- ({11*\k}, {0*\k}) node[below right] {$p$};
    \draw[->] ({0*\k}, -{1*\k}) -- ({0*\k}, {8.5*\k}) ;
    % labels
    \draw (0,0) node[below left] {$0$}
        ({2*\k},0) node[below] {$\nicefrac{1}{4}$}
        ({4*\k},0) node[below] {$\nicefrac{1}{2}$}
        ({6*\k},0) node[below] {$\nicefrac{3}{4}$}
        ({8*\k},0) node[below] {$1$}
        (0, {10*\k*\stretchY}) node[left] {$2(1+\e)$}
        (0, {8*\k*\stretchY}) node[left] {$2$}
        (0, {6*\k*\stretchY}) node[left] {$2(1-\e)$}
        ;
\end{tikzpicture}
}
\qquad
\subfigure[\label{f:sqrt-lower-b} Expected $\gft$ relative to $+\e$ (red) and $-\e$ (blue)]
{
    \begin{tikzpicture}[
    declare function={
        func(\x)
    = 
        (\x < 0) * (0)
        + and(\x >= 0, \x < 1/4) * ( (5/16) * \x * ( 5 - 4*\x ) )
        + and(\x >= 1/4, \x < 1/2) * ( (5/32) * ( -4*\x*\x + \x + 2 ) )
        + and(\x >= 1/2, \x < 3/4) * ( \x * ( 21/32 - (3/8)*\x ) )
        + and(\x >= 3/4, \x < 1) * ( (1/8) * ( 1 - \x ) * ( 8*\x + 3 ) )
        + (\x >= 1) * (0)
       ;
      }
    ]
    \def\k{0.5}
    \def\stretchY{3}
    \definecolor{myblue}{RGB}{25,175,255}
    \definecolor{myred}{RGB}{255,70,75}
    \def\colorOne{myblue}
    \def\colorTwo{myred}
    % dashed lines
    \draw[gray, dashed] ({\k*2}, 0) -- ({\k*2}, {(5/16)*\stretchY*\k*8});
    \draw[gray, dashed] ({\k*6}, 0) -- ({\k*6}, {(5/16)*\stretchY*\k*8});
    % vertical line
    \draw[gray, thin] ({\k*4}, 0) -- ({\k*4}, {\k*8.5});
    % horizontal lines
    \draw[gray, very thin] ({\k*0}, {(5/16)*\stretchY*\k*8}) -- ({\k*8}, {(5/16)*\stretchY*\k*8});
    \draw[gray, very thin] ({\k*0}, {(9/32)*\stretchY*\k*8}) -- ({\k*8}, {(9/32)*\stretchY*\k*8});
    % plots
    \draw[domain = 0:{\k*8}, myred, samples = 300] plot (\x, {\stretchY*\k*8*func( \x/(\k*8) )});
    \draw[domain = 0:{\k*8}, myblue, samples = 300] plot (\x, {\stretchY*\k*8*func( 1-\x/(\k*8) )});
    % axes
    \draw[->] ({-2*\k}, {0*\k}) -- ({11*\k}, {0*\k}) node[below right] {$p$};
    \draw[->] ({0*\k}, -{1*\k}) -- ({0*\k}, {8.5*\k}) ;
    % labels
    \draw (-1pt, {(5/16)*\stretchY*\k*8}) 
        -- (-2pt, {(5/16)*\stretchY*\k*8}) 
        -- (-2pt, {(9/32)*\stretchY*\k*8})
        -- (-1pt, {(9/32)*\stretchY*\k*8})
    ;
        \draw (-2pt, {((9/32 + 5/16)/2)*\stretchY*\k*8}) node[left, xshift=-1pt] {$\Omega( \e )$};
    \draw (0,0) node[below left] {$0$}
        ({2*\k},0) node[below] {$\nicefrac{1}{4}$}
        ({4*\k},0) node[below] {$\nicefrac{1}{2}$}
        ({6*\k},0) node[below] {$\nicefrac{3}{4}$}
        ({8*\k},0) node[below] {$1$}
        ;
    % arciapli
    \draw ({-2*\k},0) node {};
    \draw ({11.5*\k},0) node {};
\end{tikzpicture}
}
    \caption{The best posted price is $\nicefrac{1}{4}$ (resp., $\nicefrac{3}{4}$) in the $+\e$ (resp., $-\e$) case.
    By posting $\nicefrac{1}{4}$, the player suffers a $\Omega( \e )$ regret in the $-\e$ case, and the same is true posting $\nicefrac{3}{4}$ if in $+\e$ case.
    }
    \label{f:root-t-full}
\end{figure}
\begin{proof}[Proof sketch]
We build a family of distributions $\cD_{\pm \e}$ for the seller/buyer pair parameterized by $\e\in[0,1]$.
For the seller, for any $\e \in [0,1]$, we define the density
\[
    f_{S,\pm\e} 
:= 
    2(1 \pm\e)\ind_{\lsb{0, \frac{1}{4}}} + 2(1\mp\e)\ind_{\lsb{\frac{1}{2},\frac{3}{4}}}\;. 
    \tag{\text{\cref{f:sqrt-lower-a}, in red/blue}}
\]
For the buyer, we define a single density (independently of $\e$)
\[
    f_B
:=
    2 \I_{\lsb{ \frac{1}{4}, \frac{1}{2} } \cup \lsb{ \frac{3}{4}, 1 }} \;.
    \tag{\text{\cref{f:sqrt-lower-a}, in green}}
\]
In the $+\e$ (resp., $-\e$) case, the optimal price belongs to the region $[0,\nicefrac{1}{2}]$ (resp., $(\nicefrac{1}{2}, 1]$, see \cref{f:sqrt-lower-b}).
By posting prices in the wrong region $(\nicefrac{1}{2}, 1]$ (resp., $[0,\nicefrac{1}{2}]$) in the $+\e$ (resp., $-\e$) case, the learner incurs a $\Omega(\e)$ regret.
Thus, if $\e$ is bounded-away from zero, the only way to avoid suffering linear regret is to identify the sign of $\pm\e$ and play accordingly.

This closely resembles the construction for the lower bound of online learning with expert advice.
In fact, a technical proof (see Appendix~\ref{s:lower-full}), shows that our setting is harder (i.e., it has a higher minimax regret) than an instance of an expert problem (with two experts), which has a known lower bound on its minimax regret of $\frac{1}{8 \sqrt{2 \pi}}\sqrt{T}$ \citep{cover65}.
\end{proof}

\section{Realistic-Feedback Stochastic (iid) Setting \tccheck}
\label{s:real-feddb}

In this section, we tackle the problem in the more challenging realistic-feedback model, again under the assumption that the seller/buyer pairs $(S_1,B_1),(S_2,B_2),\ldots$ are $[0,1]^2$-valued i.i.d.\ random variables, all with the same law as some $(S,B)$.
We will first study the case in which $S$ and $B$ are independent (iv) and have bounded densities (bd), then discuss what happens if either one of the two assumptions is lifted.

We recall that in the realistic-feedback model, the only information collected by the learner at the end of each round $t$ consists of $\ind\{S_t \le P_t\}$ and $\ind\{P_t \le B_t\}$.

\subsection{Scouting Bandits: from Realistic Feedback to Multi-Armed Bandits \tccheck}
\label{sec:lipschitz}
The main challenge in designing low-regret algorithms with realistic feedback lies in the fact that posting a price \emph{does not} reveal the corresponding gain from trade.
We can observe this phenomenon by looking at the Decomposition lemma \eqref{e:decomp-three}.
While the local terms $\P[S \le p]$ and $\P[p \le B]$ can be reconstructed by simply posting the same price $p$ multiple times, the integral terms are inherently global: they depend on all values in $(p,1]$ and $[0,p)$, and thus estimating them requires posting prices that are \emph{far} from $p$.
This prevents direct applications of well-established algorithms, such as action elimination or UCB \citep{Slivkins19}, and suggests that this problem is harder than multiarmed bandits (as in fact it is: see \cref{sec:candidate}).

A naive approach to tackle this issue could be estimating the CDFs of $S$ and $B$ on a suitable grid of prices and using this information to reconstruct both the global and the local terms of $\GFT(\cdot,S,B)$. 
This would lead to an $\widetilde \cO \brb{ T^{3/4} }$ regret. 
Instead, our \cref{alg:meta} (Scouting Bandits) exploits better the decomposition in \cref{e:decomp-three} by learning \emph{separately} the global and local parts of the gain from trade.
First, a global exploration phase is run (scouting phase), in which prices uniformly sampled in $[0,1]$ are posted and used to simultaneously estimate the integral terms on a suitable grid of $K$ points. 
Once this is done, by replacing the integrals in \cref{e:decomp-three} with their approximations $\hat F _k$ and $\hat G_k$ for each price $q_k$ in the grid, we obtain the estimate
\[
    \E\bsb{ \GFT(q_k,S,B) }
\, \approx \,
    \P[S\le q_k] \, \hat F_k + \P[q_k\le B] \, \hat G_k 
= 
    \E \bsb{ \I\{S\le q_k\} \, \hat F_k + \I\{q_k\le B\} \, \hat G_k \mid H }
\, =: \,
    \E \bsb{ Z(k) \mid H}\;,
\]
where $H$ consists of the estimates $\hat F_j, \hat G_j$ (for all $j$) at the the end of the global exploration phase.
We are now only left to solve a bandit problem on $K$ arms with reward function $Z$: the only quantities to learn are the two local terms $\P[S\le q_k]$ and $\P[q_k\le B]$, which can be estimated with the available feedback by posting the price $q_k$.
    
\begin{algorithm}
    \textbf{input:} exploration time $T_0$, grid size $K$, and $K$-armed bandit algorithm $\alpha$\;
    \textbf{initialization:} $q_k \gets k/(K+1)$, $\hat{F}_k \gets 0$, $\hat{G}_k \gets 0$, for all $k\in [K]$\;
     \For(\tcp*[f]{scouting phase}){$t=1,2,\dots,T_0$ }{
     draw $U_t$ from $[0,1]$ uniformly at random\;
     post price $U_t$ and observe feedback $\brb{ \ind\{S_t \le U_t\}, \, \ind\{U_t \le B_t\} }$\;
     let $\hat F_k \gets \hat F_k + \frac{1}{T_0}\ind\{q_k \le U_t\le B_t\}$, and 
    $\hat G_k \gets \hat G_k + \frac{1}{T_0}\ind\{S_t \le U_t\le q_k\}$, for all $k\in [K]$\;
     }
     \For(\tcp*[f]{bandit phase}){$t=T_0+1,T_0+2,\dots$}{
     generate the next arm $I_t$ with $\alpha$\;
     post price $q_{I_t}$ and observe $\brb{ \ind\{S_t \le q_{I_t}\}, \, \ind\{q_{I_t} \le B_t\}}$\;
     feed $\alpha$ the reward $Z_t(I_t) \gets \ind\{S_t \le q_{I_t}\} \hat{F}_{I_t} + \ind\{q_{I_t} \le B_t\} \hat{G}_{I_t}$\;
     }
      \caption{Scouting Bandits}
      \label{alg:meta}
\end{algorithm}

The independence of $S$ and $B$ (iv) is required for applying \cref{e:decomp-three}, while the bounded density assumption (bd) implies the Lipschitzness of the expected gain from trade (\cref{l:decomp-lip}), which in turns allows to discretize the problem.
Later, we show how dropping either of these assumptions leads to linear regret (\Cref{thm:lower-real-iv,thm:lower-real-bd}).

We are now ready to state and prove the main result of this section.

\begin{theorem}
\label{thm:upper_real}
In the realistic-feedback stochastic (iid) setting where the distributions of the seller and buyer are independent (iv) and have densities bounded by some constant $M$, the regret of Scouting Bandits (SB) run with parameters $T_0$, $K$, and $\alpha$ satisfies, for any time horizon $T \ge T_0$,
\[
    R_T(\emph{SB})
\le
    T_0 + \lrb{\frac{4M}{K+1} + \sqrt \frac{2\pi}{T_0} }(T-T_0) + \calR_{T-T_0}(\alpha) \;,
\]
where $\mathcal{R}_\tau(\alpha)$ is a distribution-free upper bound on the regret after $\tau$ rounds of $\alpha$ in the stochastic i.i.d.\ setting with $[0,1]$-valued rewards.

In particular,
if for each $K$ we have a bandit algorithm $\alpha^K$ over $K$ arms such that
$\calR_\tau(\alpha^K)=\cO\brb{ \sqrt{K \tau} }$ (e.g., if $\alpha^K$ is the MOSS algorithm over $K$ arms \citep{audibert2009minimax}), then
tuning the parameters $T_0 := \bce{ T^{2/3} }$ and $K := \bce{ T^{1/3}}$ gives the regret bound $R_T(\emph{SB}) = \cO\brb{ M T^{2/3} }$.
\end{theorem}

\begin{proof}
Let $H := (\hat F_k, \hat G_k)_{k\in [K]}$ and denote its range space $[0,1]^{2K}$ by $\cH$.
For each $h = (f_k, g_k)_{k\in [K]} \in \cH$, let $(I_{h,t})_{t\ge T_0+1}$ be the sequence of arms pulled by $\alpha$ (possibly using some internal randomization) on the sequence of rewards $(Z_{h,t})_{t\ge T_0+1}$ defined for any time $t\ge T_0 + 1$ and all arms $k\in [K]$ by
\[
    Z_{h,t}(k) := \I\{S_t \le q_k\} f_k + \I\{q_k \le B_t\} g_k \;.
\]
Let $\ps \in \argmax_{p \in [0,1]} \E\bsb{\GFT(p,S_1,B_1)}$ and $\ks$ be the index of a point in the grid $\{q_1,\dots,q_K\}$ closest to $\ps$.
Let $P_t$ be the price posted by SB at each time $t$.
Similarly to previous sections, denote for all times $t$ and prices $p$, $\GFT_t(p) := \gft(p, S_t,B_t)$.
Then
\begin{align}
    R_T(\text{SB})
&
\le
    T_0
    +
    \sum_{t=T_0+1}^T \E\bsb{ \GFT_t(\ps) - \GFT_t(P_t) } \nonumber
\\
&
=
    T_0
    +
    \sum_{t=T_0+1}^T \Brb{ \E\bsb{ \GFT_t(\ps) } - \E \bsb{ \GFT_t(q_{\ks}) } }
    + 
    \sum_{t=T_0+1}^T \Brb{ \E \bsb{ \GFT_t(q_{\ks}) } - \E \bsb{ Z_{H,t}(\ks) } } \nonumber
\\
&
    \qquad
    + \E \lsb{ \sum_{t=T_0+1}^T Z_{H,t}(\ks) - \sum_{t=T_0+1}^T Z_{H,t}(I_{H,t}) } 
    + \sum_{t=T_0+1}^T \Brb{ \E \bsb{ Z_{H,t}(I_{H,t}) } - \E\bsb{ \GFT_t(P_t) } } \nonumber
\\
&
=:
    T_0 + \uno + \due + \tre + \quattro \;. \label{e:regret-decomp}
\end{align}
We bound the four terms separately.

For the term $\uno$, by the $4M$-Lipschitzness of the gain from trade (\Cref{l:decomp-lip}) and the fact that the step size of the grid is $1/(K+1)$, we get 
\[
    \uno
=
    \sum_{t=T_0+1}^T \Brb{ \E\bsb{ \GFT_t(\ps) } - \E \bsb{ \GFT_t(q_{\ks}) } }
\le 
    4M |p^{\star}-q_{k^\star}| (T-T_0)
\le 
    \frac{4M}{K+1}(T-T_0) \;.
\]
For the term $\due$, for any $t\ge T_0 +1$, by the independence of $H$ and $(S_t,B_t)$, we have
\begin{align*}
    \E \bsb{ Z_{H,t}(\ks) }
&
=
    \E \bsb{ \I\{S_t \le q_{\ks}\} \hat F_{\ks} + \I\{q_{\ks} \le B_t\} \hat G_{\ks} }
\\
&
=
    \P [S_t \le q_{\ks}] \P[q_{\ks} \le U_t \le B_t] + \P [q_{\ks} \le B_t] \P[S_t \le U_t \le q_{\ks} ]
=
    \E \bsb{ \GFT_t(q_{\ks}) } \;,
\end{align*}
where the last identity follows from \cref{e:decomp-three}, and in turn implies that $\due = 0$.

For the term $\tre$, using the fact that for $\P_H$-almost every $h\in \cH$, the sequence $(Z_{h,t})_{t\ge T_0 + 1}$ is included in $[0,1]$, we obtain
\begin{align*}
    \tre
& =
    \E \lsb{ \E \lsb{ \sum_{t=T_0+1}^T Z_{H,t}(\ks) - \sum_{t=T_0+1}^T Z_{H,t}(I_{H,t}) } \mid H }
\\
& \overset{(*)}{\le}
    \int_{\cH} \E \lsb{ \sum_{t=T_0+1}^T Z_{h,t}(\ks) - \sum_{t=T_0+1}^T Z_{h,t}(I_{h,t}) } \dif\P_H(h)
\le
    \calR_{T-T_0}(\alpha)
\end{align*}
where $(*)$ follows from the independence of $(I_{h,t},S_t,B_t)$ and $H$ (for any $h\in \cH$ and all $t \ge T_0 +1$) and in the last inequality we upper bounded (for $\P_H$-almost every $h \in \cH$) the regret of $\alpha$ when run on the sequence of rewards $(Z_{h,t})_{t\ge T_0 + 1}$ with $\calR_{T-T_0}(\alpha)$.

Finally, we upper bound the last term $\quattro$.
If the $K$-armed bandit algorithm $\alpha$ is randomized, let $V_t$ be its internal randomization of at each time step $t \ge T_0 +1$; otherwise, omit all references to $(V_t)_{t\ge T_0+1}$.
Define, for each time step $t \ge T_0+1$, 
$
    L_t 
:=
    (H,V_{T_0+1},S_{T_0+1},B_{T_0+1},\ldots,V_{t-1},S_{t-1},B_{t-1},V_t)
$,
$\P_t := \P [ \cdot \mid L_t ]$, 
and take a uniform random variable $U_t$ on $[0,1]$ independent of $(L_t,B_t,S_t)$.
Now, for all $t\ge T_0 +1$, leveraging the measurability of $q_{I_{H,t}}, \hat F_{I_{H,t}}, \hat G_{I_{H,t}}$ with respect to $\sigma(L_t)$, the independence of $L_t$ and $(S_t,B_t)$, and the Decomposition lemma \eqref{e:decomp-three}, we get
\begin{align*}
&
    \E \bsb{ Z_{H,t}(I_{H,t}) } - \E\bsb{ \GFT_t(P_t) }
 =
    \E \Bsb{
    \E \bsb{
        \brb{
        \I\{S_t \le q_{I_{H,t}}\} \hat F _{I_{H,t}} 
        +
        \I\{q_{I_{H,t}} \le B_t\} \hat G _{I_{H,t}}
        }
        -
        \GFT(q_{I_{H,t}}, S_t, B_t)
    \mid L_t
    }
    }
\\
& \qquad =
    \E \bsb{
        \P_t[ S_t \le q_{I_{H,t}} ] 
        \brb{ \hat F_{I_{H,t}} - \P_t [ q_{I_{H,t}} \le U_t \le B_t] }
        +
        \P_t [ q_{I_{H,t}} \le B_t] 
        \brb{ \hat G_{I_{H,t}} - \P_t [ S_t \le U_t \le q_{I_{H,t}} ] }
    }
\\
& \qquad \le
    \E \lsb{ \max_{k\in [K]} \babs{ \hat F_k - \P[q_k \le U_1 \le B_1 ] } }
    +
    \E \lsb{ \max_{k\in [K]} \babs{ \hat G_k - \P[S_1 \le U_1 \le q_k ] } }
=:
    (\mathrm{V}) + (\mathrm{VI})
     \;.  
\end{align*}
For the first addend, applying the univariate DKW inequality (\cref{t:dkw-massart}), we have
\begin{align*}
    (\mathrm{V})
& =
    \int_{[0,1]} \P \lsb{ \max_{k\in [K]} \babs{ \hat F_k - \P[q_k \le U_1 \le B_1 ] } > \e } \dif \e 
\\
& =
    \int_{[0,1]} \P \lsb{ \max_{k\in[K]} \labs{ \frac{1}{T_0} \sum_{i=1}^{T_0} \I\bcb{-U_i\I\{U_i\le B_i\} \le -q_k } - \P\bsb{ -U_1 \I\{U_1 \le B_1\} \le -q_k } } > \e } \dif \e
\\
& \le
    \int_{[0,1]} \P \lsb{ \sup_{x\in\R} \labs{ \frac{1}{T_0} \sum_{i=1}^{T_0} \I\bcb{-U_i\I\{U_i\le B_i\} \le x } - \P\bsb{ -U_1 \I\{U_1 \le B_1\} \le x } } > \e } \dif \e
\\
& \le
    \int_0^1 2 \operatorname{exp}\brb{ -2 T_0 \e^2 } \dif \e
\le
    \frac{1}{\sqrt{2 T_0}} \int_{0}^\infty e^{-u} u^{-1/2} \dif u
=  
    \sqrt{\frac \pi 2}\frac{1}{\sqrt{T_0}} \;.  
\end{align*}
Similarly, one can show that $(\mathrm{VI})\le \sqrt{\frac \pi 2}\frac{1}{\sqrt{T_0}}$ which in turn yields $\quattro \le \sqrt \frac{2\pi}{T_0} (T-T_0)$.

Putting the bounds on $\uno$-$\quattro$ together in \eqref{e:regret-decomp} gives the first part of the result. Substituting the stated choice of the parameters yields the second.
\end{proof}

Note that to achieve a regret of order $\cO (M T^{2/3})$ we tuned the parameters $T_0$ and $K$ of Scouting Bandits as a function of $T$.
If the time horizon is unknown, we can obtain the same order of regret with a standard doubling trick \citep{Nicolo06}.
Also, note that if we allow tuning the parameters as a function of the \lip{} constant $M$ (which is however unknown in general), the regret rate would improve to order $\cO(M^{1/3} T^{2/3})$. 
This can be achieved by taking $T_0 := \bce{ T^{2/3} }$ and $K := \bce{ M^{2/3} T^{1/3} }$.

\subsection{\texorpdfstring{$T^{2/3}$}{T\^{}(2/3)} Lower Bound Under Realistic Feedback (iv+bd) \tccheck}
\label{sec:candidate}

In this section, we show that the upper bound on the minimax regret we proved in \cref{sec:lipschitz} is tight.
No strategy can beat the $\cO( T^{2/3})$ rate when the seller/buyer pair $(S_t,B_t)$ is drawn i.i.d. from an unknown fixed distribution, even under the further assumptions that the valuations of the seller and buyer are independent of each other and have bounded densities.
For a full proof of the following theorem, see Appendix~\ref{s:proof-t-two-thrid-lower-bound-appe}.
\begin{theorem}
\label{thm:lower-real-iv+bd}
In the realistic-feedback model, for all horizons $T$, the minimax regret $\Rs_T$ satisfies 
\[
    \Rs_T 
:=
    \inf_{\alpha} \sup_{(S_1,B_1) \sim \cD} R_T(\alpha)
\ge
    c T^{2/3} \;,
\]
where $c \ge 11/672$, the infimum is over all learner's strategies $\alpha$, and the supremum is over all distributions $\cD$ of the seller/buyer pair such that:
\begin{itemize}
    \item[\emph{(iid)}] $(S_1,B_1),(S_2,B_2),\ldots \sim \cD$ is an i.i.d.\ sequence.
    \item[\emph{(iv)}] $S_1$ and $B_1$ are independent of each other.
    \item[\emph{(bd)}] $(S_1,B_1)$ admits a joint density bounded by $M\ge24$.
\end{itemize}
\end{theorem}
\begin{proof}[Proof sketch]
\begin{figure}
    \centering
    \subfigure[\label{f:t-two-third-lower-bound-a} Distributions $f_{S,\pm \e}$ (red/blue) and $f_B$ (green)]
    {
    \begin{tikzpicture}
    \def\k{0.5}
    \def\stretchY{0.5}
    \definecolor{myblue}{RGB}{25,175,255}
    \definecolor{myred}{RGB}{255,70,75}
    \definecolor{mygreen}{RGB}{64,128,128}
    % dashed lines
    \draw[gray, dashed] (0, {8*\k*\stretchY*(1+0.7)}) -- ({8*\k/6}, {8*\k*\stretchY*(1+0.7)});
    \draw[gray, dashed] (0, {8*\k*\stretchY*(1-0.7)}) -- ({8*\k/6}, {8*\k*\stretchY*(1-0.7)});
    \draw[gray, dashed] (0, {8*\k*\stretchY)}) -- ({8*\k}, {8*\k*\stretchY});
    % plots
    \fill[myred] (0, 0) rectangle ({8*\k/48}, {8*\k*\stretchY*(1+0.7)});
        \begin{scope}
        \clip (0,0) rectangle ({8*\k/48}, {8*\k*\stretchY*(1-0.7)});
        \foreach \x in {-2,...,10}
        {
            \draw[myblue,line width=1pt] ({-0.1-0.1*\x}, {0}) -- ({1.9-0.1*\x}, {2});
        }
        \end{scope}
        \draw[myblue, line width=1pt] 
        (0, {8*\k*\stretchY*(1-0.7)}) 
            -- ({8*\k/48}, {8*\k*\stretchY*(1-0.7)})
        ;
    \fill[myblue] ({8*\k/6}, 0) rectangle ({8*\k*3/16}, {8*\k*\stretchY*(1+0.7)});
        \begin{scope}
        \clip ({8*\k/6}, 0)  rectangle ({8*\k*3/16}, {8*\k*\stretchY*(1-0.7)});
        \foreach \x in {-2,...,10}
        {
            \draw[myred,line width=1pt] ({8*\k/6-0.1-0.1*\x}, {0}) -- ({8*\k/6+1.9-0.1*\x}, {2});
        }
        \end{scope}
        \draw[myred, line width=1pt] 
        ({8*\k/6}, {8*\k*\stretchY*(1-0.7)}) 
            -- ({8*\k*3/16}, {8*\k*\stretchY*(1-0.7)})
        ;
    \fill[myred] ({8*\k/4}, 0) rectangle ({8*\k*13/48}, {8*\k*\stretchY});
        \begin{scope}
        \clip ({8*\k/4}, 0) rectangle ({8*\k*13/48}, {8*\k*\stretchY});
        \foreach \x in {-2,...,20}
        {
            \draw[myblue,line width=1pt] ({8*\k/4-0.1-0.1*\x}, {0}) -- ({8*\k/4+1.9-0.1*\x}, {2});
        }
        \end{scope}
    \fill[myred] ({8*\k*2/3}, 0) rectangle ({8*\k*11/16}, {8*\k*\stretchY});
        \begin{scope}
        \clip ({8*\k*2/3}, 0) rectangle ({8*\k*11/16}, {8*\k*\stretchY});
        \foreach \x in {-2,...,20}
        {
            \draw[myblue,line width=1pt] ({8*\k*2/3-0.1-0.1*\x}, {0}) -- ({8*\k*2/3+1.9-0.1*\x}, {2});
        }
        \end{scope}
    \fill[mygreen!80] 
        ({8*\k*5/16}, 0) rectangle ({8*\k*1/3}, {8*\k*\stretchY})
        ({8*\k*35/48}, 0) rectangle ({8*\k*3/4}, {8*\k*\stretchY})
        ({8*\k*13/16}, 0) rectangle ({8*\k*5/6}, {8*\k*\stretchY})
        ({8*\k*47/48}, 0) rectangle ({8*\k}, {8*\k*\stretchY})
    ;
    % axes
    \draw[->] ({-0.5*\k}, {0*\k}) -- ({10.5*\k}, {0*\k}) node[below right] {$p$};
    \draw[->] ({0*\k}, -{0.5*\k}) -- ({0*\k}, {8.5*\k}) ;
    % labels
    \draw [decorate,decoration={brace, amplitude=3pt, mirror}, yshift=-1pt]
        (0,0) -- ({\k*8*3/16},0) node [black, midway, below, yshift=-2pt] {$a_1$};
    \draw [decorate,decoration={brace, amplitude=2.4pt, mirror}, yshift=-1pt]
        ({\k*8*13/48},0) -- ({\k*8*5/16},0) node [black, midway, below, yshift=-2pt] {$a_2$};
    \draw [decorate,decoration={brace, amplitude=2.4pt, mirror}, yshift=-1pt]
        ({\k*8*11/16},0) -- ({\k*8*35/48},0) node [black, midway, below, yshift=-2pt] {$a_3$};
    \draw (0,0) node[below left] {$0$}
        ({8*\k},0) node[below] {$1$}
        (0, {8*\k*\stretchY*(1+0.7)}) node[left] {$\frac{1+\e}{4 \tht}$}
        (0, {8*\k*\stretchY}) node[left] {$\frac{1}{4 \tht}$}
        (0, {8*\k*\stretchY*(1-0.7)}) node[left] {$\frac{1-\e}{4\tht}$}
        ;
\end{tikzpicture}
}
\qquad
\subfigure[\label{f:t-two-third-lower-bound-b} Expected $\gft$ relative to $+\e$ (red) and $-\e$ (blue)]
{\begin{tikzpicture}[
    declare function={
        func(\x,\eps)
    = 
        (\x < 0) * (0)
        + and(\x >= 0, \x < 1/48) * ( ( 3 * (46 - 32 * \x) * \x * ( \eps + 1 ) )/16 )
        + and(\x >= 1/48, \x < 1/6) * ( (17/96)*(1+\eps) )
        + and(\x >= 1/6, \x < 3/16) * ( -(69/8)*\x*(-1+\eps) + 6*\x*\x*(-1+\eps) + (1/96)*(-105+139*\eps) )
        + and(\x >= 3/16, \x < 1/4) * ( \eps/24 + 5/16 )
        + and(\x >= 1/4, \x < 13/48) * ( \eps/24 - 47/32 + 69*\x/8 - 6*\x*\x )
        + and(\x >= 13/48, \x < 5/16) * ( \eps/24 + 41/96 )
        + and(\x >= 5/16, \x < 1/3) * ( ( -3456*\x^2 + (1032 - 384*\eps)*\x + 152*\eps + 343 )/768 )
        + and(\x >= 1/3, \x < 2/3) * ( \eps/32 + 101/256 )
        + and(\x > 2/3, \x <= 11/16) * ( \eps/32 + (1/768) * (-2081 + 5880* \x - 3456*\x^2) )
        + and(\x > 11/16, \x < 35/48) * ( \eps / 32 + 41/96)
        + and(\x >= 35/48, \x < 3/4) * ( 37/32 + 27/8 * \x - 6*\x*\x + 19/48 * \eps - \x*\eps/2)
        + and(\x >= 3/4, \x < 13/16) * ( (\eps + 15) / 48 )
        + and(\x >= 13/16, \x < 5/6) * ( 49/32 + 27*\x/8 - 6*\x*\x + 41/96 * \eps - \x * \eps / 2)
        + and(\x >= 5/6, \x < 47/48) * ( ( \eps + 17) / 96 )
        + and(\x >= 47/48, \x < 1) * ( -1/8 * (-1 + \x) * (21 + 48 * \x + 4 * \eps) )
        + (\x >= 1) * (0)
       ;
      }
    ]
    \def\k{0.5}
    \def\stretchY{2}
    \definecolor{myblue}{RGB}{25,175,255}
    \definecolor{myred}{RGB}{255,70,75}
    \def\colorOne{myblue}
    \def\colorTwo{myred}
    % color regions of actions a_1, a_2, a_3
    \draw[gray!25, fill = gray!25] (0,0) rectangle ({\k*8*3/16}, {8*\k});
    \draw[myred!25, fill = myred!25] ({\k*8*13/48},0) rectangle ({\k*8*5/16}, {8*\k});
    \draw[myblue!25, fill = myblue!25] ({\k*8*11/16},0) rectangle ({\k*8*35/48}, {8*\k});
    % horizontal lines
    \draw[gray, very thin] ({\k*0}, {\stretchY*\k*8*(0.7/24 + 41/96)}) -- ({\k*8}, {\stretchY*\k*8*(0.7/24 + 41/96)});
    \draw[gray, very thin] ({\k*0}, {\stretchY*\k*8*(0.7 / 32 + 41/96)}) -- ({\k*8}, {\stretchY*\k*8*(0.7 / 32 + 41/96)});
    \draw[gray, very thin] ({\k*0}, {\stretchY*\k*8*(-0.7/32 + 41/96)}) -- ({\k*8}, {\stretchY*\k*8*(-0.7/32 + 41/96)});
    \draw[gray, very thin] ({\k*0}, {\stretchY*\k*8*(-0.7/24 + 41/96)}) -- ({\k*8}, {\stretchY*\k*8*(-0.7/24 + 41/96)});
    \draw[gray, very thin] ({\k*0}, {8*\k*\stretchY*(0.7/24 + 5/16)}) -- ({\k*8}, {8*\k*\stretchY*(0.7/24 + 5/16)});
    % plots
    \draw[myblue] (0,0)
        -- ({8*\k/48}, {8*\k*\stretchY*(17/96)*(1-0.7)})
        -- ({8*\k/6}, {8*\k*\stretchY*(17/96)*(1-0.7)})
        -- ({8*\k*3/16}, {8*\k*\stretchY*(-0.7/24 + 5/16)})
        -- ({8*\k/4}, {8*\k*\stretchY*(-0.7/24 + 5/16)})
        -- ({8*\k*13/48}, {8*\k*\stretchY*( -0.7/24 + 41/96 )})
        -- ({8*\k*5/16}, {8*\k*\stretchY*( -0.7/24 + 41/96 )})
        -- ({8*\k*1/3}, {8*\k*\stretchY*( -0.7/32 + 101/256 )})
        -- ({8*\k*2/3}, {8*\k*\stretchY*( -0.7/32 + 101/256 )})
        -- ({8*\k*11/16}, {8*\k*\stretchY*( -0.7 / 32 + 41/96)})
        -- ({8*\k*35/48}, {8*\k*\stretchY*( -0.7 / 32 + 41/96)})
        -- ({8*\k*3/4}, {8*\k*\stretchY*( (-0.7 + 15) / 48 )})
        -- ({8*\k*13/16}, {8*\k*\stretchY*( (-0.7 + 15) / 48 )})
        -- ({8*\k*5/6}, {8*\k*\stretchY*( ( -0.7 + 17) / 96 )})
        -- ({8*\k*47/48}, {8*\k*\stretchY*( ( -0.7 + 17) / 96 )})
        -- ({8*\k}, {8*\k*\stretchY*0})
    ;
    \draw[myred] (0,0)
        -- ({8*\k/48}, {8*\k*\stretchY*(17/96)*(1+0.7)})
        -- ({8*\k/6}, {8*\k*\stretchY*(17/96)*(1+0.7)})
        -- ({8*\k*3/16}, {8*\k*\stretchY*(0.7/24 + 5/16)})
        -- ({8*\k/4}, {8*\k*\stretchY*(0.7/24 + 5/16)})
        -- ({8*\k*13/48}, {8*\k*\stretchY*( 0.7/24 + 41/96 )})
        -- ({8*\k*5/16}, {8*\k*\stretchY*( 0.7/24 + 41/96 )})
        -- ({8*\k*1/3}, {8*\k*\stretchY*( 0.7/32 + 101/256 )})
        -- ({8*\k*2/3}, {8*\k*\stretchY*( 0.7/32 + 101/256 )})
        -- ({8*\k*11/16}, {8*\k*\stretchY*( 0.7 / 32 + 41/96)})
        -- ({8*\k*35/48}, {8*\k*\stretchY*( 0.7 / 32 + 41/96)})
        -- ({8*\k*3/4}, {8*\k*\stretchY*( (0.7 + 15) / 48 )})
        -- ({8*\k*13/16}, {8*\k*\stretchY*( (0.7 + 15) / 48 )})
        -- ({8*\k*5/6}, {8*\k*\stretchY*( ( 0.7 + 17) / 96 )})
        -- ({8*\k*47/48}, {8*\k*\stretchY*( ( 0.7 + 17) / 96 )})
        -- ({8*\k}, {8*\k*\stretchY*0})
    ;
    % axes
    \draw[->] ({-0.5*\k}, {0*\k}) -- ({10.5*\k}, {0*\k}) node[below right] {$p$};
    \draw[->] ({0*\k}, -{0.5*\k}) -- ({0*\k}, {8.5*\k});
    % labels
    \draw [decorate,decoration={brace, amplitude=3pt, mirror}, yshift=-1pt]
        (0,0) -- ({\k*8*3/16},0) node [black, midway, below, yshift=-2pt] {$a_1$};
    \draw [decorate,decoration={brace, amplitude=2.4pt, mirror}, yshift=-1pt]
        ({\k*8*13/48},0) -- ({\k*8*5/16},0) node [black, midway, below, yshift=-2pt] {$a_2$};
    \draw [decorate,decoration={brace, amplitude=2.4pt, mirror}, yshift=-1pt]
        ({\k*8*11/16},0) -- ({\k*8*35/48},0) node [black, midway, below, yshift=-2pt] {$a_3$};
    \draw (-1pt, {\stretchY*\k*8*(0.7/24 + 41/96)}) 
        -- (-2pt, {\stretchY*\k*8*(0.7/24 + 41/96)}) 
        -- (-2pt, {\stretchY*\k*8*(0.7 / 32 + 41/96)})
        -- (-1pt, {\stretchY*\k*8*(0.7 / 32 + 41/96)})
    ;
    \draw (-2pt, {\stretchY*\k*8*((0.7/24 + 41/96) + (0.7 / 32 + 41/96))/2}) node[left, xshift=-1pt] {$\Omega(\e)$};
    \draw (-1pt, {\stretchY*\k*8*(-0.7/24 + 41/96)}) 
        -- (-2pt, {\stretchY*\k*8*(-0.7/24 + 41/96)}) 
        -- (-2pt, {\stretchY*\k*8*(-0.7 / 32 + 41/96)})
        -- (-1pt, {\stretchY*\k*8*(-0.7 / 32 + 41/96)})
    ;
    \draw (-2pt, {\stretchY*\k*8*((-0.7/24 + 41/96) + (-0.7 / 32 + 41/96))/2}) node[left, xshift=-1pt] {$\Omega(\e)$};
    \draw ({(8*\k*1.01)}, {\stretchY*\k*8*(-0.7 / 24 + 41/96)}) 
        -- ({(8*\k*1.02)}, {\stretchY*\k*8*(-0.7 / 24 + 41/96)})
        -- ({(8*\k*1.02)}, {8*\k*\stretchY*(0.7/24 + 5/16)})
        -- ({(8*\k*1.01)}, {8*\k*\stretchY*(0.7/24 + 5/16)})
    ;
    \draw ({(8*\k*1.02)}, {\stretchY*\k*8*((-0.7 / 24 + 41/96) + (0.7/24 + 5/16))/2}) node[right, xshift=1pt] {$\Omega(1)$};
    \draw (0,0) node[below left] {$0$}
        ({8*\k},0) node[below] {$1$};
    % arciapli
    \draw ({-2*\k},0) node {};
    \draw ({12*\k},0) node {};
\end{tikzpicture}
}
    \caption{%
    The only three regions where it makes sense for the learner to post prices are $a_1, a_2, a_3$. Prices in  $a_1$ reveal information about the sign of $\pm\e$ suffering a $\Omega(1)$ regret; prices in $a_2$ are optimal if the distribution of the seller is the red one $(+\e)$ but incur $\Omega(\e)$ regret if it is the blue one $(-\e)$; the converse happens in $a_3$.%
    }
    \label{f:t-two-third-lower-bound}
\end{figure}
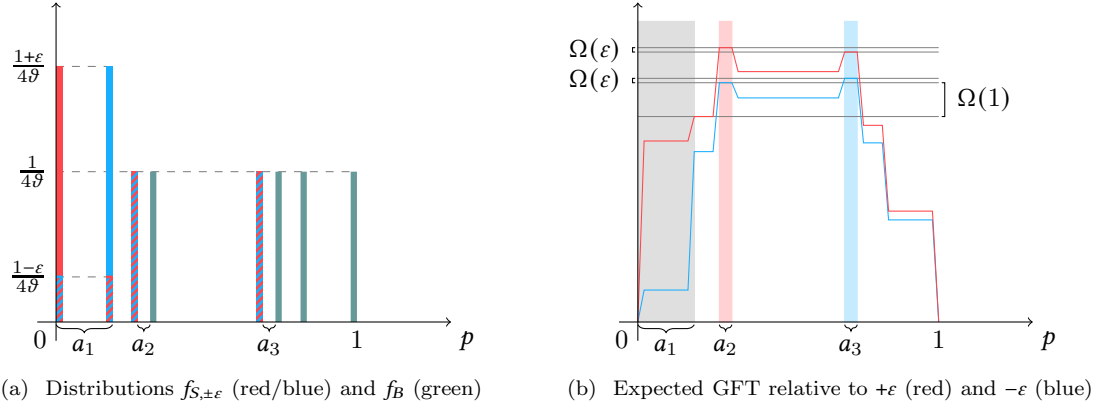
We build a family of distributions $\cD_{\pm \e}$ of the seller/buyer pair parameterized by $\e\in[0,1]$.
For the seller, for any $\e \in [0,1]$, we define the density
\[
    f_{S,\pm\e}
:=
    \frac{1}{4\tht} \lrb{
    (1\pm\e) \I_{[0,\tht]}
+ 
    (1\mp\e) \I_{\lsb{ \frac{1}{6}, \frac{1}{6} + \tht }} 
+
    \I_{\lsb{ \frac{1}{4}, \frac{1}{4} + \tht }}
+
    \I_{\lsb{ \frac{2}{3}, \frac{2}{3} + \tht }}
    } \;,
    \tag{\text{\cref{f:t-two-third-lower-bound-a}, in red/blue}}
\]
where $\tht := \nicefrac{1}{48}$ is a normalization constant.
For the buyer, we define a single density (independently of $\e$)
\[
    f_B
:=
    \frac{1}{4\tht} \lrb{
    \I_{\lsb{ \frac{1}{3}-\tht,\,\frac{1}{3} }}
+ 
    \I_{\lsb{ \frac{3}{4} - \tht,\, \frac{3}{4} }} 
+
    \I_{\lsb{ \frac{5}{6} - \tht,\, \frac{5}{6} }}
+
    \I_{\lsb{ 1-\tht,\, 1 }}
    }\;.
    \tag{\text{\cref{f:t-two-third-lower-bound-a}, in green}}
\]
In the $+\e$ (resp., $-\e$) case, the optimal price belongs to a region $a_2$ (resp., $a_3$, see \cref{f:t-two-third-lower-bound-b}).
By posting prices in the wrong region $a_3$ (resp., $a_2$) in the $+\e$ (resp., $-\e$) case, the learner incurs $\Omega(\e)$ regret.
Thus, if $\e$ is bounded away from zero, the only way to avoid suffering linear regret is to identify the sign of $\pm\e$ and play accordingly.
Clearly, the feedback received from the buyer gives no information on $\pm\e$. 
Since the feedback received from the seller at time $t$ by posting a price $p$ is $\I\{S_t \le p\}$, one can obtain information about (the sign of) $\pm\e$ only by posting prices in the costly ($\Omega(1)$-regret) sub-optimal region $a_1$.

This closely resembles the learning dilemma present in the so-called \emph{revealing action} partial monitoring game \citep{Nicolo06}.
In fact, a technical proof (see Appendix~\ref{s:proof-t-two-thrid-lower-bound-appe}), shows that our setting is harder (i.e., it has a higher minimax regret) than an instance of a revealing action problem, 
which has a known lower bound on its minimax regret of $\frac{11}{96} \brb{ \frac{1}{7} T^{2/3} }$ \citep{cesa2006regret}.
\end{proof}

\subsection{Linear Lower Bound Under Realistic Feedback (bd) \tccheck}
\label{s:lowerBoundBD}

In this section, we show that no strategy that can achieve worst-case sublinear regret when the seller/buyer pair $(S_t,B_t)$ is drawn i.i.d. from an unknown fixed distribution, even under the further assumption that the valuations of the seller and buyer have bounded densities.
This is due to a lack of observability.
For a full proof of the following theorem, see Appendix~\ref{s:lower-bd-appe}.
\begin{theorem}
\label{thm:lower-real-bd}
In the realistic-feedback model, for all horizons $T$, the minimax regret $\Rs_T$ satisfies 
\[
    \Rs_T 
:=
    \inf_{\alpha} \sup_{(S_1,B_1) \sim \cD} R_T(\alpha)
\ge
    c T \;,
\]
where $c \ge 1/24$, the infimum is over all of the learner's strategies $\alpha$, and the supremum is over all distributions $\cD$ of the seller/buyer pair such that:
\begin{itemize}
    \item[\emph{(iid)}] $(S_1,B_1),(S_2,B_2),\ldots \sim \cD$ is an i.i.d.\ sequence.
    \item[\emph{(bd)}] $(S_1,B_1)$ admits a density bounded by $M\ge24$.
\end{itemize}
\end{theorem}
\begin{proof}[Proof sketch]
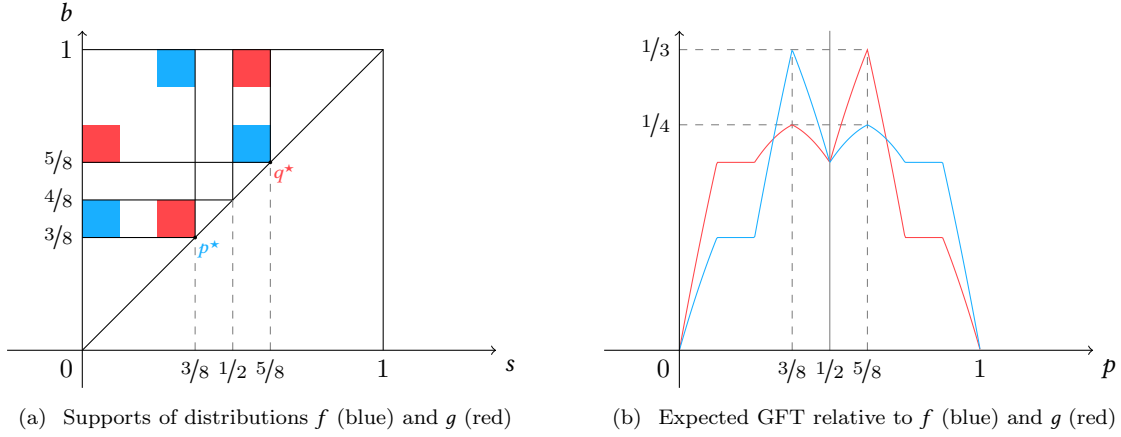
\begin{figure}
    \centering
    \subfigure[\label{f:linear-lower-bound-lip-a} Supports of distributions $f$ (blue) and $g$ (red)]
    {
    \begin{tikzpicture}
    \def\k{0.5}
    \definecolor{myblue}{RGB}{25,175,255}
    \definecolor{myred}{RGB}{255,70,75}
    \def\colorOne{myblue}
    \def\colorTwo{myred}
    % colored squares
    \fill[\colorOne] ({0*\k},{3*\k}) rectangle ({1*\k},{4*\k});
    \fill[\colorOne] ({4*\k},{5*\k}) rectangle ({5*\k},{6*\k});
    \fill[\colorOne] ({2*\k},{7*\k}) rectangle ({3*\k},{8*\k});
    \fill[\colorTwo] ({2*\k},{3*\k}) rectangle ({3*\k},{4*\k});
    \fill[\colorTwo] ({4*\k},{7*\k}) rectangle ({5*\k},{8*\k});
    \fill[\colorTwo] ({0*\k},{5*\k}) rectangle ({1*\k},{6*\k});
    % external square
    \draw ({0*\k}, {8*\k}) -- ({8*\k}, {8*\k}) -- ({8*\k}, {0*\k});
    % diagonal lines
    \draw ({0*\k}, {0*\k}) -- ({8*\k}, {8*\k});
    % p*
    \draw ({\k*0}, {3*\k}) -- ({\k*3}, {3*\k}) -- ({\k*3}, {8*\k});
    \draw[gray, dashed] ({\k*3}, {0*\k}) -- ({\k*3}, {3*\k});
    % q*
    \draw ({0*\k}, {\k*5}) -- ({\k*5}, {5*\k}) -- ({\k*5}, {8*\k});
    \draw[gray, dashed] ({5*\k}, {\k*0}) -- ({\k*5}, {5*\k});
    % 1/2
    \draw ({0*\k}, {\k*4}) -- ({\k*4}, {4*\k}) -- ({\k*4}, {8*\k});
    \draw[gray, dashed] ({4*\k}, {\k*0}) -- ({\k*4}, {4*\k});
    \draw[->] ({-2*\k}, {0*\k}) -- ({11*\k}, {0*\k}) node[below right] {$s$};
    \draw[->] ({0*\k}, -{1*\k}) -- ({0*\k}, {8.5*\k}) node[above left] {$b$};
    % labels
    \draw (0,0) node[below left] {$0$}
        ({\k*3, 0}) node[below] {$\nicefrac{3}{8}$}
        ({\k*4, 0}) node[below] {$\nicefrac{1}{2}$}
        ({\k*5, 0}) node[below] {$\nicefrac{5}{8}$}
        ({\k*8, 0}) node[below] {$1$}
        ({0, \k*3}) node[left] {$\nicefrac{3}{8}$}
        ({0, \k*4}) node[left] {$\nicefrac{4}{8}$}
        ({0, \k*5}) node[left] {$\nicefrac{5}{8}$}
        ({0, \k*8}) node[left] {$1$}
        ;
    \draw ({\k*3}, {\k*3}) node[circle, draw, fill, inner sep=0pt, minimum width=1pt] {};
        \draw ({\k*(3+0.4)}, {\k*(3-0.3)}) node[\colorOne] {$_{\ps}$};
    \draw ({\k*5}, {\k*5}) node[circle, draw, fill, inner sep=0pt, minimum width=1pt] {};
        \draw ({\k*(5+0.4)}, {\k*(5-0.3)}) node[\colorTwo] {$_{\qs}$};
    \end{tikzpicture}
    }
    \qquad
    \subfigure[\label{f:linear-lower-bound-lip-b} Expected $\gft$ relative to $f$ (blue) and $g$ (red)]
    {
    \begin{tikzpicture}[
    declare function={
        func(\x)
    = 
        (\x < 0) * (0)
        + and(\x >= 0, \x < 1/8) * ( (1/6)*\x*(7-8*\x) )
        + and(\x >= 1/8, \x < 2/8) * ( 1/8 )
        + and(\x >= 2/8, \x < 3/8) * ( (1/6)*(-8*\x*\x + 15*\x - 5/2) )
        + and(\x >= 3/8, \x < 4/8) * ( (1/24)*(-32*\x*\x + 4*\x + 11) )
        + and(\x >= 4/8, \x < 5/8) * ( (1/24)*(-32*\x*\x + 44*\x - 9) )
        + and(\x >= 5/8, \x < 6/8) * ( (1/6)*(-8*\x*\x + 9*\x - 1) )
        + and(\x >= 6/8, \x < 7/8) * ( 5/24 )
        + and(\x >= 7/8, \x < 1) * ( (1/6)*(-8*\x*\x + 5*\x +3) )
        + (\x >= 1) * (0)
      ;
      }
    ]
    \def\k{0.5}
    \def\stretchY{3}
    \definecolor{myblue}{RGB}{25,175,255}
    \definecolor{myred}{RGB}{255,70,75}
    \def\colorOne{myblue}
    \def\colorTwo{myred}
    % dashed vertical lines
    \draw[gray, dashed] ({3*\k}, {\k*0}) -- ({\k*3}, {\stretchY*\k*8*(1/3)});
    \draw[gray, dashed] ({5*\k}, {\k*0}) -- ({\k*5}, {\stretchY*\k*8*(1/3)});
    % dashed horizontal lines
    \draw[gray, dashed] ({\k*0}, {\stretchY*\k*8*(1/3)}) -- ({\k*5}, {\stretchY*\k*8*(1/3)});
    \draw[gray, dashed] ({\k*0}, {\stretchY*\k*8*(1/4)}) -- ({\k*5}, {\stretchY*\k*8*(1/4)});
    % vertical line
    \draw[gray, thin] ({\k*4}, 0) -- ({\k*4}, {\k*8.5});
    % plots
    \draw[domain = 0:{\k*8}, myred, samples = 300] plot (\x, {\stretchY*\k*8*func(1-\x/(\k*8))});
    \draw[domain = 0:{\k*8}, myblue, samples = 300] plot (\x, {\stretchY*\k*8*func(\x/(\k*8))});
    % axes
    \draw[->] ({-2*\k}, {0*\k}) -- ({11*\k}, {0*\k}) node[below right] {$p$};
    \draw[->] ({0*\k}, -{1*\k}) -- ({0*\k}, {8.5*\k});
    % labels
    \draw (0,0) node[below left] {$0$}
        ({\k*3, 0}) node[below] {$\nicefrac{3}{8}$}
        ({\k*4, 0}) node[below] {$\nicefrac{1}{2}$}
        ({\k*5, 0}) node[below] {$\nicefrac{5}{8}$}
        ({\k*8, 0}) node[below] {$1$}
        ({0, 8*\k*\stretchY/4}) node[left] {$\nicefrac{1}{4}$}
        ({0, 8*\k*\stretchY/3}) node[left] {$\nicefrac{1}{3}$}
        ;
\end{tikzpicture} 
}
    \caption{Under realistic feedback, the two densities $f$ and $g$ are indistinguishable. The optimal price $p^\star$ for $f$ gives constant regret under $g$ and $q^\star$ does the converse.}
    \label{f:linear-lower-bound-lip}
\end{figure}
Consider the two joint densities $f$ and $g$ of the seller/buyer pair as the normalized indicator functions of the red and blue squares in \cref{f:linear-lower-bound-lip-a}. Formally 
\[
    f
= 
    \frac{64}{3}  \Brb{ \I_{\lsb{ \nicefrac{0}{8}, \,\nicefrac{1}{8} } \times \lsb{ \nicefrac{3}{8}, \,\nicefrac{4}{8} } } + \I_{\lsb{ \nicefrac{2}{8}, \,\nicefrac{3}{8} } \times \lsb{ \nicefrac{7}{8}, \,\nicefrac{8}{8} }} + \I_{\lsb{ \nicefrac{4}{8}, \,\nicefrac{5}{8} } \times \lsb{ \nicefrac{5}{8}, \, \nicefrac{6}{8} }} }
\]
and $g(s,b)=f(1-b,1-s)$.
In the $f$ (resp., $g$) case, the optimal price belongs to the region $[0,\nicefrac{1}{2}]$ (resp., $(\nicefrac{1}{2}, 1]$, see \cref{f:linear-lower-bound-lip-b}).
By posting prices in the wrong region $(\nicefrac{1}{2}, 1]$ (resp., $[0,\nicefrac{1}{2}]$) in the $f$ (resp., $g$) case, the learner incurs at least a $\nicefrac{1}{3}-\nicefrac{1}{4} = \nicefrac{1}{12}$ regret.
Thus, the only way to avoid suffering linear regret is to determine if the valuations of the seller and buyer are generated by $f$ or $g$.
For each price $p \in [0,1]$, consider the four rectangles with opposite vertices $(p,p)$ and $(u_i,v_i)$, where $\lcb{ (u_i,v_i) }_{i=1,\ldots,4}$ are the four vertices of the unit square.
Note that the only information on the distribution of $(S,B)$ that the learner can gather from the realistic feedback $\brb{ \I\{S_t \le p\}, \, \I \{p\le B_t\} }$ received after posting a price $p$ is (an estimate of) the area of the portion of the support of the distribution included in each of these four rectangles.
However, these areas coincide in the cases $f$ and $g$.
Hence, under realistic feedback, $f$ and $g$ are completely indistinguishable.
Therefore, given that the optimal price in the $f$ (resp., $g$) case is $\nicefrac{3}{8}$ (resp., \nicefrac{5}{8}), the best that the learner can do is to sample prices uniformly at random in the set $\{\nicefrac{3}{8},\nicefrac{5}{8}\}$, incurring a regret of $\nicefrac{T}{24}$. 
For a formalization of this argument leveraging the techniques we described in the introduction, see Appendix~\ref{s:lower-bd-appe}.
\end{proof}

\subsection{Linear Lower Bound Under Realistic Feedback (iv) \tccheck}
\label{sec:linear_real}

In this section, we prove that in the realistic-feedback case, no strategy can achieve sublinear regret without any limitations on how concentrated the distributions of the valuations of the seller and buyer are, not even if they are independent of each other (iv).

At a high level, if the two distributions of the seller and the buyer are very concentrated in a small region, finding an optimal price is like finding a needle in a haystack.
For a full proof of the following theorem, see Appendix~\ref{sec:linear_real-appe}. 
\begin{theorem}
\label{thm:lower-real-iv}
In the realistic-feedback model, for all horizons $T$, the minimax regret $\Rs_T$ satisfies 
\[
    \Rs_T 
:=
    \inf_{\alpha} \sup_{(S_1,B_1) \sim \cD} R_T(\alpha)
\ge
    c T \;,
\]
where $c \ge 1/8$, the infimum is over all learner's strategies $\alpha$, and the supremum is over all distributions $\cD$ of the seller/buyer pair such that:
\begin{itemize}
    \item[\emph{(iid)}] $(S_1,B_1),(S_2,B_2),\ldots \sim \cD$ is an i.i.d.\ sequence.
    \item[\emph{(iv)}] $S_1$ and $B_1$ are independent of each other.
\end{itemize}
\end{theorem}
\begin{proof}[Proof sketch]
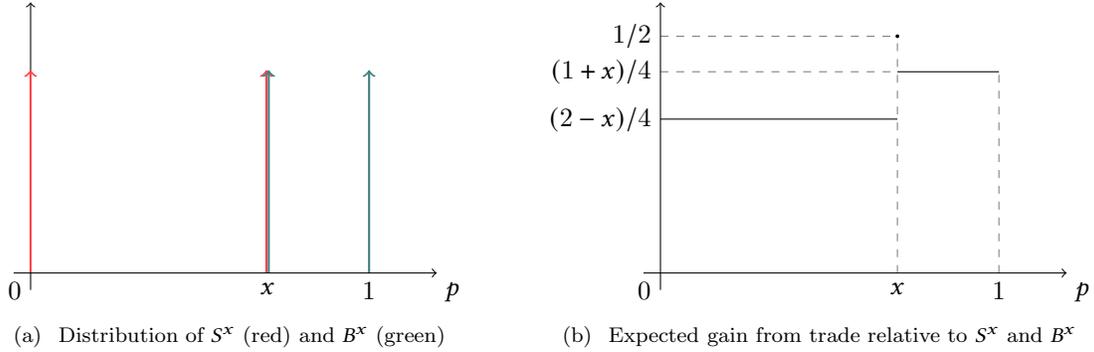
\begin{figure}
    \centering
    \subfigure[\label{fig:linear_lower_real_iv-a} Distribution of $S^x$ (red) and $B^x$ (green)]{
    \begin{tikzpicture}
    \def\k{0.45}
    \def\x{7}
    \def\stretchY{3}
    \definecolor{myblue}{RGB}{25,175,255}
    \definecolor{myred}{RGB}{255,70,75}
    \def\colorOne{myblue}
    \def\colorTwo{myred}
    \definecolor{mygreen}{RGB}{64,128,128}
    % y axis
    \draw[->] ({0*\k}, {-0.5*\k}) -- ({0*\k}, {8*\k});
    % Dirac measures
    \draw[->, myred, thick] (0,0) -- (0, {2*\k*\stretchY});
    \draw[->, myred, thick] ({\x*(0.995)*\k}, 0) -- ({\x*(0.995)*\k}, {2*\k*\stretchY});
    \draw[->, mygreen, thick] ({\x*(1.005)*\k}, 0) -- ({\x*(1.005)*\k}, {2*\k*\stretchY});
    \draw[->, mygreen, thick] ({10*\k}, 0) -- ({10*\k}, {2*\k*\stretchY});
    \draw (0,0) node[below left] {$0$}
        ({\x*\k}, 0) node [below] {$x$}
        ({10*\k}, 0) node [below] {$1$}
        ({12*\k}, 0) node [below right] {$p$}
    ;
    % x axis
    \draw[->] ({-0.5*\k}, {0*\k}) -- ({12*\k}, {0*\k});
    \end{tikzpicture}}
    \qquad
    \subfigure[\label{fig:linear_lower_real_iv-b} Expected gain from trade relative to $S^x$ and $B^x$]
    {
    \begin{tikzpicture}
    \def\k{0.45}
    \def\x{7}
    \def\stretchY{3.5}
    \definecolor{myblue}{RGB}{25,175,255}
    \definecolor{myred}{RGB}{255,70,75}
    \def\colorOne{myblue}
    \def\colorTwo{myred}
    \definecolor{mygreen}{RGB}{64,128,128}
    \draw[gray, dashed] ({\x*\k},0) -- ({\x*\k},{2*\k*\stretchY});
    \draw[gray, dashed] ({10*\k},0) -- ({10*\k},{\x+10)/10*\k*\stretchY});
    \draw[gray, dashed] (0,{(\x+10)/10*\k*\stretchY}) -- ({\x*\k},{(\x+10)/10*\k*\stretchY});
    \draw[gray, dashed] (0,{2*\k*\stretchY}) -- ({\x*\k},{2*\k*\stretchY});
    \draw (0,{(20-\x)/10*\k*\stretchY}) -- ({\x*\k},{(20-\x)/10*\k*\stretchY});
    \draw ({\x*\k},{(\x+10)/10*\k*\stretchY}) -- ({10*\k},{(\x+10)/10*\k*\stretchY});
     \draw[->] (0, {-0.5*\k}) -- (0,{8*\k});
     \draw[->] ({-0.5*\k}, 0) -- ({12*\k},0);
     \draw ({\x*\k},{2*\k*\stretchY}) node[circle, draw, fill, inner sep = 0pt, minimum width = 1pt] {};
    \draw (0,0) node[below left] {$0$}
        ({\x*\k},0) node[below] {$x$}
        ({10*\k},0) node[below] {$1$}
        (0,{(\x+10)/10*\k*\stretchY}) node[left] {$(1+x)/4$} % 1/4 + x/4 = (1 + x)/4
        (0,{(20-\x)/10*\k*\stretchY}) node[left] {$(2-x)/4$} % 1/4 + x/4 = (1 + x)/4
        (0,{2*\k*\stretchY}) node[left] {$1/2$} % 1/4 + x/4 = (1 + x)/4
        ({12*\k},0) node[below right] {$p$}
    ;
    \end{tikzpicture}
    }
    \caption{All prices but $x$ have high regret. However, under realistic feedback, finding $x$ in finite time is impossible.}
    \label{fig:linear_lower_real_iv}
\end{figure}
Consider a family of seller/buyer distributions $(S^x,B^x)$, parameterized by $x\in I$, where $I$ is a small interval centered in $\nicefrac 12$, $S^x$ and $B^x$ are independent of each other, and they satisfy
\[
S^x=\begin{cases}
    x &\text{with probability $\frac 12$}\\
    0 &\text{with probability $\frac 12$} 
\end{cases}\;,
\qquad
B^x=\begin{cases}
    x &\text{with probability $\frac 12$}\\
    1 &\text{with probability $\frac 12$} 
\end{cases}\;.
\]
The distributions and the corresponding gain from trade are represented in \cref{fig:linear_lower_real_iv-a} and \cref{fig:linear_lower_real_iv-b}, respectively. 
A direct verification shows that the best fixed price with respect to $(S^x,B^x)$ is $p=x$. 
Furthermore, by posting any other prices, the learner incurs a regret of approximately $1/2$ with probability $1/4$.
It is intuitively clear that no strategy can locate (exactly!) each possible $x\in I$ in a finite number of steps. 
This results, for any strategy, in regret of at least (approximately) $T/8$.
See Appendix~\ref{sec:linear_real-appe} for a more detailed analysis.
\end{proof}

\section{Adversarial Setting: Linear Lower Bound Under Full Feedback \tccheck}
\label{sec:adversarial}

In this section, we prove that even in the simpler full-feedback case, no strategy can achieve worst-case sublinear regret in an adversarial setting. 
Lower bounds for the adversarial setting have a slightly different structure that the stochastic ones. The idea of the proof is to build, for any strategy, a {\em hard} sequence of sellers and buyers' valuations $(s_1,b_1),(s_2,b_2),\ldots$ which causes the algorithm to suffer linear regret for any horizon $T$.
\begin{theorem}
\label{thm:adv-lower}
In the full-feedback adversarial (adv) setting, for all horizons $T$, the minimax regret $\Rs_T$ satisfies
\[
    \Rs_T 
:=
    \inf_{\alpha} \sup_{(s_1,b_1),(s_2,b_2),\ldots} R_T(\alpha)
\ge 
    c T \;,
\]
where $c \ge 1/4$, the infimum is over all of the learner's strategies $\alpha$, and the supremum is over all deterministic sequences $(s_1,b_1),(s_2,b_2),\ldots \in [0,1]^2$ of the seller and buyer's valuations.
\end{theorem}
\begin{proof}
We begin by fixing any strategy $\alpha$ of the learner.
This is a sequence of functions $\alpha_t$ mapping the past feedback $(s_1,b_1), \ldots, (s_{t-1},b_{t-1})$, together with  some internal randomization, to the price $P_t$ to be posted by the learner at time $t$.
In other words, the strategy maintains a distribution $\nu_t$ over the prices that is updated after observing each new pair $(s_t,b_t)$ and used to draw each new price $P_t$.
We will show how to constructively determine a sequence of seller/buyer valuations that is hard for $\alpha$ to learn.
This sequence is oblivious to the prices $P_1,P_2,\ldots$ posted by $\alpha$, in the sense it does not have access to the realizations of its internal randomization.
The idea is, at any time $t$, to determine a seller/buyer pair $(s_t,b_t)$ either of the form $(c_t,1)$ or $(0,d_t)$, with $c_t \approx \frac{1}{2} \approx d_t$, such that the probability $\nu_t$ that the strategy picks a price $P_t \in [s_t,b_t]$ (i.e., that there is a trade) is at most $\nicefrac{1}{2}$ and, at the same time, there is common price $\ps$ which belongs to $[s_t,b_t]$ for all times $t$.
This way, since $b_t-s_t \approx \frac{1}{2}$ for all $t$, the regret of $\alpha$ with respect to $(s_1,b_1), (s_2,b_2), \ldots$ is at least (approximately) greater than or equal to $\nicefrac{T}{4}$.

The formal construction proceeds inductively as follows. 
Let $\e \in \brb{ 0, \frac{1}{18} }$.
Let
\[
    \begin{cases}
        c_1:=\frac{1}{2}-\frac{3}{2}\varepsilon, \  d_1:=\frac{1}{2}-\frac{1}{2}\varepsilon, \  s_1 := 0,  \ b_1:=d_1,
    &
        \text{ if } \nu_{1}\bsb{\bsb{0,\frac{1}{2}-\frac{1}{2}\varepsilon}}\le\frac{1}{2} \;,
    \\
        c_1:=\frac{1}{2}+\frac{1}{2}\varepsilon, \  d_1:=\frac{1}{2}+\frac{3}{2}\varepsilon, \  s_1 := c_1, \  b_1:=1,
    &
        \text{ otherwise}.
    \end{cases}
\]
Then, for any time $t$, given that $c_i, d_i, s_i, b_i$ are defined for all $i\le t$ and recalling that $\nu_{t+1}$ is the distribution over the prices at time $t+1$ (of the strategy $\alpha$ after observing the feedback $(s_1,b_1), \ldots, (s_t,b_t)$), let
\[
    \begin{cases}
        c_{t+1}:=c_{t}, \ d_{t+1}:=d_{t}-\frac{2\varepsilon}{3^{t}}, \ s_{t+1}:=0, \ b_{t+1}:=d_{t+1},
    &
        \text{if } \nu_{t+1}\bsb{ \bsb{0,c_t+\frac{\varepsilon}{3^{t}}} }\le\frac{1}{2} \;,
    \\
        c_{t+1}:=c_{t}+\frac{2\varepsilon}{3^{t}}, \ d_{t+1}:=d_{t}, \ s_{t+1} := c_{t+1}, \ b_{t+1}:=1,
    &
        \text{otherwise}.
    \end{cases}
\]
Then the sequence of seller/buyer valuations $(s_1,b_1),(s_2,b_2),\ldots$ defined above by induction satisfies:
\begin{itemize}
    \item $\nu_t\bsb{ [s_t,b_t] } \le \frac{1}{2}$, for each time $t$.
    \item There exists $\ps \in [0,1]$ such that $\ps\in[s_t,b_t]$, for each time $t$ (e.g. $\ps:=\lim_{t \to \infty} c_t$).
    \item $b_t - s_t \ge \frac{1-3\e}{2}$, for each time $t$.
\end{itemize}
This implies, for any horizon $T$,
\[
    R_T(\alpha)
=
    \sum_{t=1}^T \gft(\ps,s_t,b_t) -  \sum_{t=1}^T \E \lsb{ \gft ( P_t, s_t, b_t ) }
\ge
    \sum_{t=1}^T (b_t-s_t)\brb{1-\nu_t \bsb{ [s_t,b_t] }}
\ge 
    \frac{1-3\e}{4}T\;.
\]
Since $\e$ and $\alpha$ are arbitrary, 
this yields immediately $\Rs \ge {T}/{4}$.
\end{proof}

\section{Breaking Linear Lower Bounds: Weakly Budget Balanced Results}
\label{s:wbb}

In this section, we show how to break the linear lower bound of \cref{s:lowerBoundBD} without requiring the independence of the valuations of the seller and the buyer. 
To do so, we move from a budget balance to a \emph{weak} budget balance mechanism.
In this setting, rather than posting a single price, the learner can post two (possibly distinct) prices $0 \le p \le p' \le 1$, $p$ to the seller, and $p'$ to the buyer. 
This condition allows the platform to extract money from the trade but not to subsidize it.
Naturally, this changes the benchmark: 
if the learner posts a pair $(p,p') \in [0,1]^2$ and the valuations of the seller and the buyer are $(s,b) \in [0,1]^2$, the net gain of the seller is $p-s$ while that of the buyer is $b-p'$. 
Thus, the gain from trade in this setting becomes
\begin{align*}
\GFT\colon [0,1]^2 \times [0,1]^2 & \to [0,1] \;,
\\
(p,p',s,b)&\mapsto (b - p' + p - s) \ind\{s\le p \le p' \le b\} \;.
\end{align*}
Note that posting the same price $p=p'$ to both the seller and the buyer leads to the old definition of gain for trade
$(b-s)\I \{s \le p \le b \}$,
which we denoted by $\GFT(p,s,b)$ in previous sections.
For this reason and to keep the notation lighter, we will denote $\GFT(p,p,s,b)$ simply by $\GFT(p,s,b)$ here.

We design a explore-then-exploit algorithm, that we call Scouting Blindits (\cref{alg:etc}).
In the exploration phase (scouting phase), the high-level idea is to leverage the Decomposition lemma \eqref{e:decomp-four} to build an accurate estimate $\hat{F}_k + \hat{G}_k$ of the gain from trade at each point $q_k$ of a suitably fine grid.
By the bounded density assumption, which implies the \lip{}ness of the expected gain from trade, this is sufficient to approximate $\E\bsb{ \GFT(\cdot,S_1,B_1) }$ uniformly, so that the price $\hat{P}^\star$ that maximizes the estimates $\hat{F}_k + \hat{G}_k$ is an approximate maximizer of gain from trade.
In the exploitation phase (blind phase) the algorithm posts $\hat{P}^\star$ blindly to both the seller and the buyer, ignoring all the feedback it receives. 
This implies that Scouting Blindits is actually budget balanced during the whole blind phase (which, after tuning, constitutes the majority of time).

\begin{algorithm}
     \textbf{input:} exploration time $T_0$, grid size $K$\;
     \textbf{initialization:} $q_i \gets i/(K+1)$, $\hat{F}_i \gets 0$, $\hat{G}_i \gets 0$, for all $i\in [K]$, $k\gets 1$\;
     \For(\tcp*[f]{scouting phase}){$t=1,\dots,2K T_0$}{
     \uIf{$t$ is odd}
     {
        draw $U_t$ from $[q_k,1]$ uniformly at random\;
        post the prices $( q_k, U_t )$ and observe feedback $\ind\{S_t \le q_k \le  U_t \le B_t\}$\;
        let $\hat{F}_k \gets \hat{F}_k + \frac{1}{T_0}(1-q_k)\ind\{S_t \le q_k \le U_t \le B_t\}$\;
     }
     \Else{
        draw $V_t$ from $[0,q_k]$ uniformly at random\;
        post the prices $(V_t, q_k)$ and observe feedback  $\ind\{S_t \le V_t \le q_k \le B_t\}$\;
        let $\hat{G}_k \gets \hat{G}_k + \frac{1}{T_0} q_k \ind\{S_t \le V_t \le q_k \le B_t\}$\;
    }
        \If{$t \ge 2k T_0$}
        {let $k\gets k+1$\;} 
    }
     compute $\hat{I}^\star \in \argmax_{k \in [K]} \brb{ \hat{F}_k + \hat{G}_k }$ and let $\hat{P}^\star \gets q_{\hat{I}^\star}$\;
     \For(\tcp*[f]{blind phase}){$t=2K T_0+1,\dots$}{
     post the price $\hat{P}^\star$ to both the seller and the buyer\;
     }
     \caption{\label{alg:etc} Scouting Blindits (SBl)}
\end{algorithm}

We will now show that the regret of suitable tuning of \cref{alg:etc} is at most $\widetilde{O}\brb{ T^{3/4} }$.

\begin{theorem}
\label{t:wbb}
If $(S_1,B_1),(S_2,B_2),\ldots $ is an i.i.d.\ sequence and $(S_1,B_1)$ has a density bounded by some constant $M$, then the regret of Scouting Blindits run with parameters $T_0$ and $K$ satisfies, for any time horizon $T\ge 2KT_0$,
\[
    \R_T(\text{\emph{SBl}}) 
\le 
    2K T_0 + 2 \lrb{ \frac{2M}{K} + \inf_{\e >0} \lrb{ \e + K \operatorname{exp}\brb{- 2 \e^2 T_0} }} (T-2K T_0) \;.
\]
In particular, tuning the parameters $K := \bce{ T^{1/4} }$ and $T_0 := \bce{ \brb{ \sqrt{T} \log{T} }/2 }$ yields 
\[
    \R_T(\text{\emph{SBl}}) 
=
    \cO \brb{ (M+\log T) \, T^{3/4} }\;.
\]
\end{theorem}

\begin{proof}
Let $p^\star \in \argmax_{p \in [0,1]} \E\bsb{\GFT(p,S_1,B_1)}$.
Fix any $\e>0$ and define the good event $\cE$ as
\[
    \cE
:= 
    \bigcap_{k=1}^K \Bcb{ \babs{ \E[\GFT(q_k,S_1,B_1)]-\brb{ \hat{F}_k + \hat{G}_k } } \le \e}.
\]
By \cref{e:decomp-four}, for each $k \in [K]$, we have that $\hat{F}_k + \hat{G}_k$ is the empirical mean of $T_0$ i.i.d.\ $[0,1]$-valued copies of a random variable whose expected value is $\E\bsb{ \GFT(q_k,S_1,B_1) }$. Then, by Chernoff-Hoeffding inequality and a union bound, we have that
\[
\P[\cE^c] \le 2K \operatorname{exp}\lrb{- 2 \e^2 T_0} \;.
\]
On the other hand, for each $p \in [0,1]$, define $k(p) \in [K]$ as the index of a point in the grid $\{q_1,\dots,q_K\}$ closest to $p$.
Then, on the good event $\cE$, for all $t \ge 2KT_0 +1$, we have that:
\begin{multline*}
    \E\bsb{ \GFT(p^\star,S_t,B_t) } - \E\bsb{ \GFT(\hat{P}^\star,S_t,B_t)\mid \hat{P}^\star }
\\
\begin{aligned}
&=
   \E\bsb{ \GFT(p^\star,S_t,B_t) } - \E\bsb{\GFT(q_{k(p^\star)},S_t,B_t) } + \E\bsb{ \GFT(q_{k(p^\star)},S_t,B_t) } - (\hat{F}_{k(p^\star)} + \hat{G}_{k(p^\star)}) \\
   & \qquad (\hat{F}_{k(p^\star)} + \hat{G}_{k(p^\star)}) - (\hat{F}_{k(\hat{P}^\star)}+ \hat{G}_{k(\hat{P}^\star)}) + (\hat{F}_{k(\hat{P}^\star)}+ \hat{G}_{k(\hat{P}^\star)}) - \E[\GFT(\hat{P}^\star,S_t,B_t)\mid \hat{P}^\star]
\\
&\le
    \frac{4M}{K} + \e + 0 + \e = \frac{4M}{K} + 2\e \;.
\end{aligned}
\end{multline*}
So, if $(P_t,Q_t)_{t\in [T]}$ are the prices posted by \cref{alg:etc}, we have that
\begin{multline*}
    \sum_{t=1}^T \Brb{ \E\bsb{ \GFT(p^\star,S_t,B_t) } - \E \bsb{ \GFT(P_t,Q_t,S_t,B_t) } }
\\
\begin{aligned}
&\le
2K T_0 + \sum_{t=2K T_0 +1}^T \Brb{ \E\bsb{ \GFT(p^\star,S_t,B_t) } - \E\bsb{ \GFT\brb{\hat{P}^\star,S_t,B_t} } }
\\
&= 2K T_0 + \sum_{t=2K T_0 +1}^T\E\Bsb{ \E[\GFT(p^\star,S_t,B_t)] - \E[\GFT\brb{\hat{P}^\star,S_t,B_t} \mid \hat{P}^\star] }
\\
&\le 2K T_0 + \P[\cE^c] (T- 2K T_0) + \sum_{t=2K T_0 +1}^T\E\Bsb{ \Brb{\E\bsb{ \GFT(p^\star,S_t,B_t)} - \E\bsb{ \GFT\brb{\hat{P}^\star,S_t,B_t} \mid \hat{P}^\star } } \I_{cE} }
\\
&\le 2K T_0 +  2K \operatorname{exp}\lrb{-2 \e^2 T_0} (T- 2K T_0) + \sum_{t=2K T_0 +1}^T\E\lsb{ \lrb{ \frac{4M}{K} + 2\e } \I_{\cE} }\\
&\le 2K T_0 + 2 \lrb{ \frac{2M}{K} + \e + K \operatorname{exp}\brb{- 2 \e^2 T_0} } (T-2K T_0) \;.
\end{aligned}
\end{multline*}
By the arbitrariness of $\e$, we have the first part of the result. Substituting the stated choice of the parameters in the last expression (doing the calculation choosing e.g. $\e = \bce{ T^{-1/4}}$) yields the second part of the result.
\end{proof}

As we noted in \cref{sec:lipschitz}, if the time horizon is unknown, we can retain the regret guarantees of the previous result with a standard doubling trick.

It is straightforward to see that the same construction of \cref{thm:lower-real-iv+bd} applies, giving a lower bound on the regret in the weakly budget balance setting of order $\Omega( T^{2/3} )$. 
Indeed, there the distribution of the buyer is known. Therefore, it is counterproductive to post two different prices to the seller and the buyer (same quality of feedback but lower gain from trade).
We leave the gap between this $\Omega( T^{2/3} )$ lower bound and the $\widetilde{\cO}(T^{3/4})$ regret of \cref{alg:etc} open for future research.

\section{Learning with One Bit}
\label{s:one-bit}

In this section, we discuss the (im)possibility of learning with less than a realistic feedback.
We start by noting that Scouting Blindits requires only one bit of feedback $\I\{S_t \le P_t \le P'_t \le B_t\}$, i.e., whether or not the trade occurred at time $t$ if $P_t$ was posted to the seller and $P'_t$ to the buyer.
In this setting, one can therefore achieve sublinear regret without observing the two bits $\I\{S_t \le P_t\}$ and $\I\{P'_t \le B_t\}$ provided by realistic feedback.
Thus, it is natural to wonder whether the single bit $\I\{S_t \le P_t \le B_t\}$ is sufficient for obtaining sublinear regret bounds also in the budget balanced setting.
\emph{This is not the case}: even under the further assumptions of bounded densities (bd) and independent valuations (iv), a single bit in the budget balance setting does not provide sufficient observability.
Indeed, consider a first instance in which the seller $S$ and buyer $B$ have uniform distributions on $[0,1]$, independent of each other.
In this case, the only maximizer of the expected gain from trade is $\ps = \nicefrac{1}{2}$.
As a second instance, consider two independent distributions of the seller $S'$ and buyer $B'$ with densities (bounded by $2$ and even infinite differentiable) $f_{S'}(s) := 4\fracc{(4-2s^3+s^2)}{(s^3 - s^2 + 4)^2}$ and $f_{B'}(b) = b(b-\nicefrac{1}{2})(b-1)+1$ respectively.
Then, for all $p\in [0,1]$, we have $\P[ S \le p \le B ] = \P[ S' \le p \le B' ]$.
Therefore, the two instances are indistinguishable under the single-bit feedback, but a direct verification shows that in the second instance, $\ps = \nicefrac{1}{2}$ is \emph{not} a maximizer of the expected gain from trade.
Leveraging these facts and the continuity of the gain from trade in the two instances leads to a linear minimax regret, using the same ideas as in \cref{thm:lower-real-bd}.

\section{Conclusions \tccheck}
This work initiates the study of the bilateral trade problem in a regret minimization framework.
We designed algorithms and proved tight bounds on the regret rates achieved under various feedback and private valuation models.  

Our work opens several possibilities for future investigation. 
One first and natural research direction is related to the more general settings of two-sided markets with multiple buyers and sellers, different prior distributions, and complex valuation functions.  
A second direction is related to the tight characterization of the regret rates for weak budget balance mechanisms (which we proved are strictly better than the budget balance rates in some cases).
Finally, we believe other classes of markets, which assume prior knowledge of the agent's preferences, could be fruitfully studied in a regret minimization framework.

\appendix

\section{Missing Details of Section~\ref{s:decomp}}
\label{s:decomp-appe}

In this section, we prove the Decomposition lemma and its corollary as stated in  Section~\ref{s:decomp}.

\decomplemma*
\begin{proof}
We begin by proving \cref{e:decomp-zero}. For any $s,b\in[0,1]$, we have
\begin{align*}
    \GFT(p,s,b)
& =
    (b-s) \I\{s \le p \le b\}
=
    \int_s^b \dif \lambda \I\{s \le p \le b\}
=
    \int_{[0,1]} \I\{s\le p \le b\} \I\{s\le \lambda \le b\} \dif \lambda
\\
& 
=
    \int_{[p,1]} \I\{s \le p \le \lambda \le b\} \dif \lambda
    +
    \int_{[0,p]} \I\{s \le \lambda \le p \le b\} \dif \lambda \;.
\end{align*}

\cref{e:decomp-one} is an immediate consequence of \cref{e:decomp-zero} and Fubini's theorem.

We now prove \cref{e:decomp-two}. 
Under the assumptions, \cref{e:decomp-one} implies
\begin{align*}
    \P[S \le p \le U \le B]
&
=   
    \P\Bsb{ \{S \le p\} \cap \{ U \le B \} \cap \bcb{ U \in [p,1] } }
\\
&
=
    \int_{[p,1]} \P\bsb{ \{S \le p\} \cap \{ U \le B \} \mid U= \lambda } \dif \P_U(\lambda)
=
    \int_{[p,1]} \P[ S \le p \le \lambda \le B ] \dif \lambda \;.
\end{align*}
The equality 
$
    \P[S \le U \le p \le B]
=
    \int_{[0,p]} \P[S\le \lambda \le p \le B] \dif\lambda
$ 
can be shown analogously, proving \cref{e:decomp-two}. 

\cref{e:decomp-three} is an immediate consequence of \cref{e:decomp-two}, leveraging independence. 

We now prove \cref{e:decomp-four}.
If $p\in \{0,1\}$, the result follows from \cref{e:decomp-three}. 
Thus, assume $p\in (0,1)$. Then
\[
    \E \bsb{ \gft(p,S,B) }
=
    \int_{[p,1]} \P \bsb{ S \le p \le \lambda \le B } \dif \lambda
    +
    \int_{[0,p]} \P \bsb{ S \le \lambda \le p \le B } \dif \lambda
     \;.
\]
For the first addend, we have,
\begin{align*}
&
    \int_{[p,1]} \P \bsb{ S \le p \le \lambda \le B } \dif \lambda
=
    (1-p) \int_{[p,1]} \P \bsb{ S \le p \le \lambda \le B } \dif \P_U(\lambda)
\\
& \hspace{0.82265pt}
=
    (1-p) \int_{[p,1]} \P [ S \le p \le U \le B \mid U = \lambda ] \dif \P_U(\lambda)
= 
    (1-p) \P [ S \le p \le U \le B ] 
=
    \E \bsb{ (1-p) \I \{ S \le p \le U \le B \} } \;.
\end{align*}
Analogously, one shows
$
    \int_{[0,p]} \P \bsb{ S \le \lambda \le p \le B } \dif \lambda
=
    \E\bsb{ p \I \{S \le V \le p \le B \} }
$,
which gives \cref{e:decomp-four}.
\end{proof}

We conclude this section by showing that the bounded-density assumption implies the \lip{}ness of the expected gain from trade.

\decompcorollarylip*
\begin{proof}
Take any two $0 \le p < q \le 1$.
We have that
\begin{multline*}
    \labs{
    \int_{[p,1]}\P [ S \le p \le \lambda \le B ] \dif \lambda
    -
    \int_{[q,1]}\P [ S \le q \le \lambda \le B ] \dif \lambda
    }
\\
\begin{aligned}
&
=
    \labs{
    \int_{[q,1]} \Brb{ \P [ S \le p \le \lambda \le B ] - \P [ S \le q \le \lambda \le B ] } \dif \lambda
    + \int_{[p,q)} \P [ S \le p \le \lambda \le B ] \dif \lambda
    }
\\
&
\le
    \sup_{\lambda\in[0,1]} \babs{ \P [ S \le p \le \lambda \le B ] - \P [ S \le q \le \lambda \le B ] }
    +
    \labs{ p - q }
\\
&
\le
    \sup_{\lambda\in[0,1]}\labs{ \int_{[\lambda,1]} \lrb{ \int_{[p,q)} f(s,b) \dif s } \dif b }
    +
    \labs{ p - q }
\le
    2M \labs{p - q} \;.
\end{aligned}
\end{multline*}
Analogously, we can prove that
\[
    \labs{
    \int_{[0,p]}\P [ S \le \lambda  \le p \le B ] \dif \lambda
    -
    \int_{[0,q]}\P [ S \le \lambda \le q \le B ] \dif \lambda
    }
\le
    2M \labs{p-q} \;.
\]
Thus, \cref{e:decomp-one} yields $\babs{ \E \bsb{\GFT(p,S,B)} - \E \bsb{\GFT(q,S,B)}} \le 4M \labs{p-q}$.
\end{proof}

\section{Existence of the Best Price}
\label{s:existenceMax}

In this section, we show that a price $\ps$ maximizing the expected regret always exists. 

\begin{lemma}
The function $p\mapsto \E\bsb{ \GFT(p,S,B) }$ is upper semicontinuous. 
In particular, there exists a maximizer $\ps \in [0,1]$.
\end{lemma}
\begin{proof}
Let $U$ be a random variable that is uniform on $[0,1]$ and independent of $(S,B)$.
By the Decomposition lemma \eqref{e:decomp-two}, it is sufficient to show that
\[
    f\colon \R \to [0,1], \ p\mapsto\P[S \le p \le U \le B]
    \qquad \text{ and } \qquad
    g\colon \R \to [0,1], \ p\mapsto\P[S \le U \le p \le B]
\]
are both upper semicontinuous.
We now prove that $f$ is upper semicontinuous, i.e., that for any $p \in \R$, we have
\[
    \limsup_{q \to p} f(q) \le f(p) \;.
\]
To do so, we show that for any $p \in \R$ and any two sequences $q_n\uparrow p$, $r_n \downarrow p$, we have that
\[
    \limsup_{q_n \uparrow p}f(q_n) \le f(p)
    \qquad \text{ and } \qquad
    \limsup_{r_n \downarrow p}f(r_n) \le f(p) \;.
\]
If $p\in \R\m[0,1]$, the result is trivially true.
Thus, let $p \in [0,1]$, $q_n\uparrow p$ and $r_n \downarrow p$.
Then, 
\begin{align*}
    &\I\{ S \le q_n \le U \le B \} \to \I \{ S < p \le U \le B \}\;, \qquad n \to \iop\;,
\\
    &\I\{ S \le r_n \le U \le B \} \to \I \{ S \le p < U \le B \}\;, \qquad n \to \iop\;,
\end{align*}
pointwise everywhere. By \leb{}'s dominated convergence theorem, it  follow that, if $n\to \iop$,
\begin{align*}
    &f(q_n) \to \P [ S < p \le U \le B ] \le \P [ S \le p \le U \le B ] = f(p) \;,
\\
        &f(r_n) \to \P [ S \le p < U \le B ] = \P [ S \le p \le U \le B ] = f(p) \;.
\end{align*}
By the arbitrariness of $p$, $(q_n)_{n\in \N}$ and $(r_n)_{n\in \N}$, $f$ is therefore upper semicontinuous.
Analogously, one can prove that $g$ is upper semicontinuous.
Hence, $p\mapsto \E\bsb{ \GFT(p,S,B) } = f(p) + g(p)$ is an upper semicontinuous function defined on the compact set $[0,1]$, so it attains its maximum at some $\ps \in [0,1]$ by the Weierstrass theorem.
\end{proof}

\section{Model and Notation \tccheck}
\label{s:model-appe}

For all $T\in \N$, we denote the set of the first $T$ integers $\{1,\ldots,T\}$ by $[T]$.
If $\P$ is a probability measure and $X$ is a random variable, we denote by $\P_X$ the probability measure defined for any (measurable) set $E$, by $\P_X[E] := \P[ X \in E]$.
We denote the expectation of a random variable $X$ with respect to the probability measure $\P$ by $\E_\P[X]$.
If a measure $\nu$ is absolutely continuous with respect to another measure $\mu$ with density $f$, we denote $\nu$ by $f \mu$, so that for any (measurable) set $E$, $(f \mu) [E] := \nu[E] = \int_E f(x) \dif \mu(x)$. 
We denote the \leb{} measure on the interval $[0,1]$ by $\mu_L$ and the product \leb{} measure on $[0,1]^\N$ by $\bmu_L$.
For any set $E$ and $x\in E$, we denote the Dirac measure on $x$ by $\delta_x$ (the dependence on $E$ will always be clear from context).

\subsection{The Learning Model \tccheck} 
\label{s:learning-model-appe}
In this section, we introduce an abstract notion of sequential games which encompasses all the settings we discussed in the main part of the paper, providing a unified perspective.
This will be especially useful when proving lower bounds.
\begin{definition}[Sequential game]
\label{d:game-appe}
A \emph{(sequential) game} is a tuple $\sG := (\cX,\cY,\cZ,\rho,\fhi,\sP)$, where:
\begin{itemize}
    \item $\cX,\cY,\cZ$ are sets called the \emph{player's action space}, \emph{adversary's action space}, and \emph{feedback space}.
    \item $\rho \colon \cX \times \cY \to [0,1]$ and $\fhi \colon \cX \times \cY \to \cZ$ are called the \emph{reward} and \emph{feedback} functions\footnote{More precisely, we need $\cX,\cY,\cZ$ to be non-empty measurable spaces and $\rho,\fhi$ to be measurable functions. To avoid clutter, in the following we will never mention explicitly these types of standard measurability assumptions unless strictly needed.}.
    \item $\sP$ is a set of probabilities on the set $\cY^\N$ of sequences in $\cY$, called the \emph{adversary's behavior}.
\end{itemize}
\end{definition}
This definition generalizes the partial monitoring games of \citep{lattimore2020bandit,bartok2014partial} to settings with infinitely many arms and is able to model adversarial, i.i.d., and more general stochastic settings all at once.
Before proceeding, we introduce another few extra handy definitions that will be used throughout the paper.
\begin{definition}
\label{d:extra-stuff-appe}
If $\sG = (\cX,\cY,\cZ,\rho,\fhi,\sP)$ is a game, then we say the following.
\emph{The sample space} is the set $\Omega := \cY^{\N} \times [0,1]^{\N}$.
\emph{The adversary's actions} $\brb{ Y_t }_{t\in \N}$ and \emph{the player's randomization} $\brb{ U_t }_{t\in\N}$ are sequences of random variables defined, for all $t \in \N$ and $\omega = \brb{ (y_n)_{n\in\N}, (u_n)_{n\in \N} } \in \Omega$, by $Y_t(\omega):=y_t$ and $U_t(\omega) := u_t$.
\emph{The set of scenarios} $\sS$ is the set of probability measures $\P$ on $\Omega$ of the form $\P = \boldsymbol \mu \otimes \boldsymbol{\mu}_L$, where $\boldsymbol \mu \in \sP$.
\end{definition}
For the sake of conciseness, whenever we fix a game $\sG$, we will assume that all the objects (sets, functions, random variables) presented in Definitions~\ref{d:game-appe}--\ref{d:extra-stuff-appe} are fixed and denoted by the same letters without declaring them explicitly each time, unless strictly needed.

Note that this setting models an \emph{oblivious} adversary since its actions are independent of the player's past randomization, i.e., for all $t\in \N$, $\P_{Y_{t+1} \mid Y_1, \ldots, Y_t, U_1, \ldots, U_t } = \P_{Y_{t+1} \mid Y_1, \ldots, Y_t}$.
Note also that we are assuming that the randomization of the player's strategy is carried out by drawing numbers in the interval $[0,1]$ independently and uniformly at random.
We can restrict ourselves to this case in light of the Skorokhod Representation Theorem \cite[Section 17.3]{williams1991probability} without losing (much) generality.
We now introduce formally the strategies of the player, the resulting played actions, and the corresponding feedback.

\begin{definition}[Player's strategies, actions, and feedback]
Given a game $\sG$, we define a \emph{player's strategy} as a sequence of functions $\alpha = (\alpha_t)_{t\in \N}$ such that, for each $t \in \N$, $\alpha_{t} \colon [0,1]^t\times\cZ^{t-1} \to \cX$.\footnote{When $t=1$, $[0,1]^t \times \cZ^{t-1} := [0,1]$. 
In the following, we will always adopt this type of convention without mention it.}
Given a player's strategy $\alpha$, we define inductively (on $t$) the corresponding sequences of \emph{player's actions} $(X_t)_{t\in \N}$ and \emph{player's feedback} $(Z_t)_{t\in \N}$ by $X_t := \alpha_t(U_1, \ldots, U_t, Z_1, \ldots, Z_{t-1})$, $Z_t := \varphi (X_t, Y_t)$.
In the sequel, we will denote the set of all strategies for a game $\sG$ by $\sA(\sG)$.
\end{definition}
To lighten the notation, we will write $\sA$ instead of $\sA(\sG)$ if it is clear from context.
We can now extend the standard notions of regret, worst-case regret, and minimax regret to our general setting.
\begin{definition}[Regret]
\label{d:regr-appe}
Given a game $\sG$ and a horizon $T\in \N$, we define the \emph{regret} (of $\alpha \in \sA$ in a scenario $\P \in \sS$), the \emph{worst-case regret} (of $\alpha \in \sA$), and the \emph{minimax regret} (of $\sG$), respectively, by
\[
R^{\P}_T(\alpha) := \sup_{x \in \cX} \E_{\P}\lsb{\sum_{t=1}^T \rho(x,Y_t) - \sum_{t=1}^T \rho(X_t,Y_t)}\;,
\quad
R_T^{\sS}(\alpha) := \sup_{\P \in \sS} R^{\P}_T(\alpha)\;,
\quad
\Rs_T(\sG) := \inf_{\alpha \in \sA(\sG)} R_T^{\sS}(\alpha)\;.
\]
\end{definition}
If $\sG$ and $\widetilde{\sG}$ are two games and $\Rs_T(\sG) \ge \Rs_T(\tilde \sG)$, we say that $\tilde{\sG}$ is \emph{easier} than $\sG$ (or equivalently, that $\sG$ is \emph{harder} than $\tilde{\sG}$).
When it is clear from the context, we will omit the dependence on $\sG$ in $\Rs_T(\sG)$.

\subsection{Bilateral Trade as a Game \tccheck}
\label{s:biltrad-setting-appe}

We now formally cast the various instances of bilateral trade we introduced in \cref{s:bil-tr-model} into our sequential game setting.\footnote{Straightforwardly, the same can be done for the weak budget balance setting we studied in \cref{s:wbb}.}
In this context, we think of the learner as the \emph{player} and the environment as the \emph{adversary}.

\subsubsection{Player's Actions, Adversary's Actions, and Reward}

The player's action space $\cX$ is the unit interval $[0,1]$.
This corresponds to the player posting the same price to both the seller and the buyer (budget balance).
The adversary's action space $\cY$ is $[0,1]^2$. 
They are the pairs of valuations of the seller and buyer.
The reward function $\rho$ is the gain from trade $\gft\colon [0,1] \times [0,1]^2 \to [0,1]$, $\brb{ p, (s,b) } \mapsto (b-s) \I\{ s \le p \le b\}$.

\subsubsection{Available Feedback} 

\begin{description}
    \item[Full:] the feedback space $\cZ$ is the unit square $[0,1]^2$ and the feedback function is $\fhi \colon [0,1] \times [0,1]^2 \to [0,1]^2$, $\brb{ p, (s,b) } \mapsto (s,b)$.
    This corresponds to the seller and the buyer revealing their valuations at the end of a trade.
    \item[Realistic:] the feedback space $\cZ$ is the boolean square $\{0,1\}^2$ and the feedback function is $\fhi \colon [0,1] \times [0,1]^2 \to \{0,1\}^2$, $\brb{ p, (s,b) } \mapsto \brb{ \I\{s \le p\}, \I\{ p\le b \} }$.
    This corresponds to the seller and the buyer accepting or rejecting a trade at a price $p$.
\end{description}

\subsubsection{Adversary's Behavior} 
\begin{description}
    \item[Stochastic (iid):] 
    the adversary's behavior $\sP = \sPiid$ consists of products of a single probability on $\cY = [0,1]^2$, i.e., $\bmu \in \sPiid$ if and only if there exists a probability measure $\mu$ on $[0,1]^2$ such that $\bmu = \otimes_{t\in \N} \, \mu$.
    This corresponds to a stochastic i.i.d.\ environment, where however the valuations of the seller and the buyer could be correlated.
    
    We will also investigate the following stronger assumptions.
    \begin{description}
        \item[Independent valuations (iv):] 
        the adversary's behavior $\sP = \sPiv$ is the subset of $\sPiid$ in which the valuations of the seller and the buyer are independent, i.e., $\bmu \in \sPiv$ if and only if there exist two probability measures $\mu_S,\mu_B$ on $[0,1]$ such that $\bmu = \otimes_{t\in \N} \, (\mu_S \otimes \mu_B)$.
        \item[Bounded density (bd):] 
        for a fixed $M \ge 1$, the adversary's behavior $\sP = \sPbd$ is the subset of $\sPiid$ in which the joint distribution of the valuations of buyer and seller has a density bounded by $M$, i.e., $\bmu \in \sPbd$ if and only if there exists a density $f\colon [0,1]^2 \to [0,M]$ such that  $\bmu = \otimes_{t\in \N} \, ( f \mu ) $, where $\mu = \mu_L \otimes \mu_L$.
        \item[Independent valuations with bounded density (iv+bd):] 
        for a fixed $M \ge 1$, the adversary's behavior $\sP = \sPivbd$ is the subset $\sPiv \cap \sPbd$ of $\sPiid$.
    \end{description}
    \item[Adversarial (adv):] 
    the adversary's behavior $\sP = \sPadv$ consists of products of Dirac measures on $\cY = [0,1]^2$, i.e., $\bmu \in \sPadv$ if and only if there exists a sequence $(s_t,b_t)_{t\in\N} \s [0,1]^2$ such that $\bmu = \otimes_{t\in \N} \, \delta_{(s_t,b_t)} $.
    This corresponds to a deterministic, oblivious, and adversarial environment.
\end{description}

\section{Two Key Lemmas on Simplifying Sequential Games \tccheck}
\label{s:keylemmas}
In this section we introduce some useful techniques that could be of independent interest for proving lower bounds in sequential games.
The idea is to give sufficient conditions for a given game to be harder than another, where the second one has a known lower bound on its minimax regret.

At a high level, the first lemma shows that if the adversary's actions are independent of each other, a game $\tilde{\sG}$ is easier than game $\sG$ if $\tilde{\sG}$ can be embedded in $\sG$ in such a way that: 
\begin{enumerate}
    \item The optimal player's actions of $\tilde{\sG}$ are no better than the ones in $\sG$.
    \item The suboptimal player's actions of $\tilde{\sG}$ no worse than the ones in $\sG$.
    \item At distributional level, the quality of the feedback in $\tilde{\sG}$ is no worse than that in $\sG$.
\end{enumerate}
The proof is deferred to Appendix~\ref{sec:lemmas}.

\begin{restatable}[Embedding]{lemma}{lembedding}
\label{l:embedding}
Let $\sG := (\cX,\cY,\cZ,\rho,\fhi,\sP)$ and $\tilde{\sG} := (\tilde{\cX},\tilde{\cY},\tilde{\cZ},\tilde{\rho},\tilde{\fhi},\tilde{\sP})$ be two games, $\sS, \tilde \sS$ their respective sets of scenarios, $(Y_t)_{t\in\N}, (\tilde Y_t)_{t\in\N}$ their adversaries' actions, and $T\in \N$ a horizon.
Assume that
$Y_1,\ldots,Y_T$ are $\P$-independent for any scenario $\P \in \sS$, $\tilde Y_1, \ldots, \tilde Y_T$ are $\tilde \P$-independent for any scenario $\tilde \P \in \tilde \sS$, and 
that there exist 
$
    \slf \colon \cX \to \tilde{\cX}
$, 
$\slg \colon \tilde{\cZ} \to \cZ
$,
and 
$\slh \colon \tilde{\sS} \to \sS
$
satisfying:
\begin{enumerate}
    \item \label{bullone} $\sup_{\tilde{x} \in \tilde{\cX}} \sum_{t=1}^T \E_{\tilde{\P}} \bsb{ \tilde{\rho}(\tilde{x},\tilde{Y}_t) } \le \sup_{x \in \cX}\sum_{t=1}^T \E_{\slh(\tilde{\P})} \bsb{ \rho(x,Y_t) }$ for any scenario $\tilde{\P} \in \tilde{\sS}$.
    \item \label{bulltwo} $\E_{\tilde{\P}}\bsb{ \tilde{\rho} \brb{ \slf(x),\tilde{Y}_t } } \ge \E_{\slh(\tilde{\P})}\bsb{ \rho(x,Y_t) }$ for any time $t\in [T]$, scenario $\tilde \P \in \tilde \sS$, and action $x \in \cX$.
    \item \label{bullthree} $\tilde{\P}_{\slg\lrb{\tilde{\fhi}\brb{\slf(x),\tilde{Y}_t}}} = \brb{\slh(\tilde{\P})}_{\fhi(x,Y_t)}$ for any time $t\in [T]$, scenario $\tilde{\P} \in \tilde{\sS}$, and action $x \in \cX$.
\end{enumerate}
Then 
$
\Rs_T(\sG) \ge \Rs_T(\tilde{\sG})
$.
\end{restatable}

The second lemma addresses feedback with uninformative (i.e., scenario-independent) components.
At a high level, if the feedback of some of the player's actions has one or more uninformative components, the game can be simplified by getting rid of them.
The player can achieve this by simulating the uninformative parts of the feedback using her randomization.
The proof is deferred to Appendix~\ref{sec:lemmas}.

\begin{restatable}[Simulation]{lemma}{lsimulation}
\label{l:simulation}
Let $\cV,\cW$ be two sets, $\sG := (\cX,\cY,\cZ,\rho,\fhi,\sP)$ a game with $\cZ = \cV \times \cW$, $\sS$ its set of scenarios, $(Y_t)_{t\in \N}$ its adversary's actions, $\pi \colon \cZ \to \cV$ the projection on $\cV$, and $T \in \N$ a horizon.
Assume that
$Y_1,\ldots,Y_T$ are $\P$-independent for any scenario $\P \in \sS$ and that there exist disjoint sets $\cR, \cU \subset \cX$ such that $\cR \cup \cU = \cX$ and
\begin{enumerate}
    \item \label{buone} For any time $t \in [T]$ and action $x \in \cR$ there exists $\psi_{t,x} \colon [0,1] \to \cW$ such that, for all $\P \in \sS$,
    \[
        \P_{\fhi(x,Y_t)} = \P_{\pi\brb{\fhi(x,Y_t)}} \otimes (\mu_L)_{\psi_{t,x}} \;.
    \]
    \item \label{butwo} For any time $t \in [T]$ and action $x \in \cU$, there exists $\gamma_{t,x}\colon[0,1] \to \cZ$ such that, for all $\P \in \sS$,
    \[
        \P_{\fhi(x,Y_t)} = (\mu_L)_{\gamma_{t,x}} \;.
    \]
\end{enumerate}
Let $* \in \cV$ and define
\[
\tilde{\fhi} \colon \cX \times \cY \to \cV, \ (x,y) \mapsto
\begin{cases}
\pi\brb{\varphi(x,y)} \;, &\text{ if } x \in \cR,\\
* \;, &\text{ if } x \in \cU.
\end{cases}
\]
Define the game $\tilde{\sG} := (\cX,\cY,\cV,\rho,\tilde{\fhi},\sP)$.
Then
$
\Rs_T(\sG) \ge \Rs_T(\tilde{\sG})
$.
\end{restatable}

\subsection{Proofs of the lemmas \tccheck}
\label{sec:lemmas}

In this section, we will give a full proof of the two useful Embedding and Simulation lemmas introduces in Appendix~\ref{s:keylemmas}.
To lighten the notation, for any $m,n \in \N$, with $m\le n$ and a family $(\lambda_k)_{k\in\N}$ we let $\lambda_{m:n} := (\lambda_m, \lambda_{m+1}, \ldots, \lambda_n)$ and similarly $\lambda_{n:m} := (\lambda_{n}, \lambda_{n-1} \ldots, \lambda_m)$.

We begin by proving the Embedding lemma, that we restate for ease of reading.

\lembedding*

\begin{proof}
Fix any strategy $\alpha \in \sA(\sG)$. For each time $t \in \N$, define
\[
\tilde{\alpha}_t \colon [0,1]^t \times \tilde{\cZ}^{t-1} \to \tilde \cX, (u_1, \dots, u_t, \tilde{z}_1, \dots, \tilde{z}_{t-1}) \mapsto \slf\Brb{\alpha_t\brb{u_1,\dots, u_t, \slg(\tilde{z}_1), \dots, \slg(\tilde{z}_{t-1})}}.
\]
Then $\tilde{\alpha} := (\tilde{\alpha}_t)_{t \in \N} \in \sA(\tilde \sG)$. As usual, let $(Y_t)_{t \in \N}$ and $(U_t)_{t \in \N}$ be the adversary's actions and the player's randomization in game $\sG$ and $(X_t)_{t \in \N}$ and $(Z_t)_{t \in \N}$ the player's actions and the feedback according to the strategy $\alpha$. Let $(\tilde{Y}_t)_{t \in \N}, (\tilde{U}_t)_{t \in \N},(\tilde{X}_t)_{t \in \N},(\tilde{Z}_t)_{t \in \N}$ be the corresponding objects for the game $\tilde{\sG}$ and the strategy $\tilde{\alpha}$. Furthermore, define
\[
\hat{X}_1 = \alpha_1(\tilde{U}_1), \quad 
\hat{Z}_1 = \slg\brb{\tilde{\fhi}(\tilde{X}_1, \tilde{Y}_1)}, \quad 
\hat{X}_2 = \alpha_2(\tilde{U}_1,\tilde{U}_2,\hat{Z}_1), \quad 
\hat{Z}_2 = \slg\brb{\tilde{\fhi}(\tilde{X}_2, \tilde{Y}_2)}, \dots \;.
\]
Fix $\tilde{\P} \in \tilde{\sS}$, where $\tilde{\sS}$ is the set of scenarios of the game $\tilde \sG$. 
Then $\tilde{\P}_{\tilde{U}_1} = \brb{\slh(\tilde{\P})}_{U_1}$. Now, since $X_1 = \alpha_1(U_1)$ and $\hat{X}_1 =\alpha_1 (\tilde{U}_1)$, we also have that $\tilde{\P}_{\hat{X}_1, \tilde{U}_1} = \brb{\slh(\tilde{\P})}_{X_1,U_1} =: \Q_1$. 
Now, up to a set with $\Q_{1}$-probability zero, if $x_1 \in \cX$ and $u_1 \in [0,1]$, we get, using \cref{bullthree}:
\begin{align*}
    \tilde{\P}_{\hat{Z}_1 \mid \hat{X}_1=x_1, \tilde{U}_1=u_1} &= 
    \tilde{\P}_{\slg\Brb{\tilde{\fhi}\brb{\slf(\hat{X}_1), \tilde{Y}_1} } \mid \hat{X}_1=x_1, \tilde{U}_1=u_1} =
    \tilde{\P}_{\slg\Brb{\tilde{\fhi}\brb{\slf(x_1), \tilde{Y}_1} } }
    \\
    &=
    \brb{\slh(\tilde{\P})}_{\fhi(x_1,Y_1)}=
    \brb{\slh(\tilde{\P})}_{\fhi(X_1,Y_1) \mid X_1=x_1, U_1=u_1} = \brb{\slh(\tilde{\P})}_{Z_1 \mid X_1=x_1, U_1=u_1} \;.
\end{align*}
So, if $A_1 \subset \cZ$ and $D \subset \cX \times [0,1]$, then
\begin{align*}
\tilde{\P}_{\hat{Z}_1, \brb{\hat{X}_1, \tilde{U}_1}}(A_1 \times D) &= \int_{D} \P_{\hat{Z}_1 \mid \hat{X}_1=x_1, \tilde{U}_1 = u_1}(A_1) \dif\P_{\hat{X}_1, \tilde{U}_1}(x_1,u_1)
\\
&=
\int_{D} \brb{\slh(\tilde{\P})}_{Z_1 \mid X_1=x_1, U_1=u_1}(A_1) \dif\brb{\slh(\tilde{\P})}_{X_1, U_1}(x_1,u_1) = \brb{\slh(\tilde{\P})}_{Z_1, (X_1, U_1)}(A_1 \times D) \,,
\end{align*}
from which it follows that $\tilde{\P}_{\hat{Z}_1, \hat{X}_1, \tilde{U}_1} = \brb{\slh(\tilde{\P})}_{Z_1, X_1, U_1}$.
By induction, suppose that for time $t\in [T-1]$ we have that
\[
\tilde{\P}_{\hat{Z}_t,\dots,\hat{Z}_1,\hat{X}_t,\dots,\hat{X}_1, \tilde{U}_t,\dots,\tilde{U}_1} = \brb{\slh(\tilde{\P})}_{Z_t,\dots,Z_1,X_t,\dots,X_1, U_t,\dots,U_1} \;.
\]
Then, using independence we have that
\[
\tilde{\P}_{\hat{Z}_t,\dots,\hat{Z}_1,\hat{X}_t,\dots,\hat{X}_1, \tilde{U}_{t+1}, \tilde{U}_t,\dots,\tilde{U}_1} = \brb{\slh(\tilde{\P})}_{Z_t,\dots,Z_1,X_t,\dots,X_1, U_{t+1} , U_t,\dots,U_1} \;.
\]
Furthermore, since $X_{t+1} = \alpha_{t+1}(U_1,\dots, U_{t+1}, Z_1, \dots, Z_t)$ and $\hat{X}_{t+1} = \alpha_{t+1} (\tilde{U}_1,\dots, \tilde{U}_{t+1}, \hat{Z}_1, \dots, \hat{Z}_t)$, we have that
\[
\tilde{\P}_{\hat{Z}_t,\dots,\hat{Z}_1,\hat{X}_{t+1},\hat{X}_t,\dots,\hat{X}_1, \tilde{U}_{t+1}, \tilde{U}_t,\dots,\tilde{U}_1} = \brb{\slh(\tilde{\P})}_{Z_t,\dots,Z_1,X_{t+1},X_t,\dots,X_1, U_{t+1}, U_t,\dots,U_1} =: \Q_{t+1}\;.
\]
Now, up to a set with $\Q_{t+1}$-probability zero, if $x_1, \dots, x_{t+1} \in \cX$, $u_1, \dots, u_{t+1} \in [0,1]$, and $z_1, \dots, z_t \in \cZ$, by the $\tilde{\P}$-independence of $\tilde{Y}_1, \ldots, \tilde Y_{t+1}$, \cref{bullthree}, and the $\slh(\tilde \P)$-independence of $Y_1, \ldots, Y_{t+1}$, we have
\begin{multline*}
    \tilde{\P}_{\hat{Z}_{t+1} \mid \hat{Z}_t = z_t,\dots,\hat{Z}_1 = z_1,\hat{X}_{t+1} =x_{t+1},\dots,\hat{X}_1 = x_1, \tilde{U}_{t+1} = u_{t+1},\dots,\tilde{U}_1=u_1}
\\
\begin{aligned}
&=
    \tilde{\P}_{\slg\Brb{\tilde{\fhi}\brb{\slf(\hat{X}_{t+1}),\tilde{Y}_{t+1}}}\mid \hat{Z}_t = z_t,\dots,\hat{Z}_1 = z_1,\hat{X}_{t+1} =x_{t+1},\dots,\hat{X}_1 = x_1, \tilde{U}_{t+1} = u_{t+1},\dots,\tilde{U}_1=u_1} 
=
    \tilde{\P}_{\slg\Brb{\tilde{\fhi}\brb{\slf(x_{t+1}),\tilde{Y}_{t+1}}}}
\\
&= \brb{\slh(\tilde{\P})}_{\fhi(x_{t+1},Y_{t+1})} =
\brb{\slh(\tilde{\P})}_{\fhi(X_{t+1},Y_{t+1}) \mid Z_t = z_t,\dots,Z_1 = z_1,X_{t+1} =x_{t+1},\dots,X_1 = x_1, U_{t+1} = u_{t+1},\dots,U_1=u_1}
\\
&= \brb{\slh(\tilde{\P})}_{Z_{t+1} \mid Z_t = z_t,\dots,Z_1 = z_1,X_{t+1} =x_{t+1},\dots,X_1 = x_1, U_{t+1} = u_{t+1},\dots,U_1=u_1} \;.
\end{aligned}
\end{multline*}
So, if $A_{t+1}\subset \cZ, D \subset \cZ^t \times \cX^{t+1}\times[0,1]^{t+1}$, we have that
\begin{multline*}
\tilde{\P}_{\hat{Z}_{t+1},\brb{\hat{Z}_{t:1},\hat{X}_{t+1:1},\tilde{U}_{t+1:1} } }(A_{t+1}\times D)
\\
\begin{aligned}
&=
\int_{D} \tilde{\P}_{\hat{Z}_{t+1} \mid \hat{Z}_{t:1} = z_{t:1},\hat{X}_{t+1:1} =x_{t+1:1}, \tilde{U}_{t+1:1} = u_{t+1:1}}(A_{t+1}) \dif\tilde{\P}_{\hat{Z}_{t:1},\hat{X}_{t+1:1}, \tilde{U}_{t+1:1}}(z_{t:1}, x_{t+1:1}, u_{t+1:1})
\\
&=
\int_{D} \brb{\slh(\tilde{\P})}_{Z_{t+1} \mid Z_{t:1} = z_{t:1},C_{t+1:1} =x_{t+1:1}, U_{t+1:1} = u_{t+1:1}}(A_{t+1}) \dif\brb{\slh(\tilde{\P})}_{Z_{t:1},X_{t+1:1}, U_{t+1:1}}(z_{t:1}, x_{t+1:1}, u_{t+1:1})
\\
&=
\brb{\slh(\tilde{\P})}_{Z_{t+1},\brb{ Z_{t:1},X_{t+1:1}, U_{t+1:1} }}(A_{t+1}\times D) \;,
\end{aligned}
\end{multline*}
from which follows that $\tilde{\P}_{\hat{Z}_{t+1},\dots,\hat{Z}_1,\hat{X}_{t+1},\dots,\hat{X}_1, \tilde{U}_{t+1},\dots,\tilde{U}_1} = \brb{\slh(\tilde{\P})}_{Z_{t+1},\dots,Z_1,X_{t+1},\dots,X_1, U_{t+1},\dots,U_1}$.
In particular, for each $t \in [T]$ we have that $\tilde{\P}_{\hat{X}_t} = \brb{\slh(\tilde{\P})}_{X_t}$. 
Hence, using the $\slh(\tilde \P)$-independence of $Y_1,\ldots,Y_T$, Item~(\ref{bulltwo}), and the $\tilde{\P}$-independence of  $\tilde{Y}_1,\ldots,\tilde{Y}_T$, we get
\begin{align*}
\sum_{t=1}^T \E_{\slh(\tilde{\P})} \bsb{\rho(X_t,Y_t) } &= \sum_{t=1}^T \int_{\cX}\E_{\slh(\tilde{\P})} \bsb{\rho(x,Y_t)} \dif  \brb{\slh(\tilde{\P})}_{X_t}(x)
\\
&\le
\sum_{t=1}^T \int_{\cX}\E_{\tilde{\P}}[\tilde{\rho}(\slf(x),\tilde{Y}_t)] \dif  \brb{\slh(\tilde{\P})}_{X_t}(x)
\\
&= 
\sum_{t=1}^T \int_{\cX}\E_{\tilde{\P}}[\tilde{\rho}(\slf(x),\tilde{Y}_t)] \dif  \tilde{\P}_{\hat{X}_t}(x)
\\
&=
\sum_{t=1}^T \E_{\tilde{\P}} \bsb{\tilde \rho(\slf(\hat{X}_t),\tilde{Y}_t) } =
\sum_{t=1}^T \E_{\tilde{\P}} \bsb{\tilde \rho(\tilde{X}_t,\tilde{Y}_t) } \;.
\end{align*}
Then, using Item~(\ref{bullone}), we have
\begin{align*}
R_T^{\slh(\tilde{\P})}(\alpha) &= \sup_{x \in \cX} \bbrb{ \sum_{t=1}^T \E_{\slh(\tilde{\P})} \bsb{\rho(x,Y_t) } - \sum_{t=1}^T \E_{\slh(\tilde{\P})} \bsb{\rho(X_t,Y_t) }}
\\
&
\ge
\sup_{\tilde x \in \tilde \cX} \bbrb{ \sum_{t=1}^T \E_{\tilde{\P}} \bsb{\tilde \rho(\tilde x,\tilde{Y}_t) } - \sum_{t=1}^T \E_{\tilde{\P}} \bsb{\tilde \rho(\tilde{X}_t,\tilde{Y}_t) }} =
R_T^{\tilde{\P}}(\tilde{\alpha}) \;.
\end{align*}
Since $\tilde \P$ was arbitrary, we get
\[
    \Rs_T(\tilde{\sG}) 
= 
    \inf_{\beta \in \sA(\tilde{\sG})} R_T^{\tilde \sS}(\beta) 
\le 
    R_T^{\tilde \sS}(\tilde{\alpha}) 
= 
    \sup_{\tilde{\P} \in \tilde{\sS} } R_T^{\tilde{\P}}(\tilde{\alpha}) \le \sup_{\tilde{\P} \in \tilde{\sS} } R_T^{\slh(\tilde{\P})}(\alpha) \le \sup_{\P \in \sS } R_T^{\P}(\alpha) = R_T^{\sS}(\alpha) \;,
\]
and since $\alpha$ was arbitrary, we get
\[
\Rs_T(\tilde{\sG}) \le \inf_{\alpha \in \sA(\sG)} R_T^{\sS}(\alpha) = \Rs_T(\sG) \;.
\]
\end{proof}

We now prove the Simulation lemma we introduced in Appendix~\ref{s:keylemmas} showing how to get rid of uninformative feedback.

\lsimulation*

\begin{proof}
For each number $a\in [0,1]$, fix a binary representation $0.a_1 a_2 a_3 \ldots$ of $a$
and define $\xi(a) := 0. a_1 a_3 a_5 \ldots$, $\zeta(a) := 0. a_2 a_4 a_6 \ldots$.
Note that the two resulting functions $\xi,\zeta \colon [0,1] \to [0,1]$ are $\mu_L$-independent with common (uniform) push-forward distribution $(\mu_L)_\xi = \mu_L = (\mu_L)_\zeta$.

Let $(Y_t)_{t \in \N}, (U_t)_{t \in \N}$ be the sequences of adversary's actions and player's randomization for the sequential game $\sG$ and note that they are also the same for the sequential game $\tilde{\sG}$.
For each $t \in \N$ define $\beta_t \colon \cX \times \cV \times [0,1] \to \cZ$ via
\[
(x,v,u) \mapsto
\begin{cases}
\brb{v,\psi_{t,x}(u)} \;, &\text{ if } x \in \cR \;, \\
\gamma_{t,x}(u) \;,       &\text{ if } x \in \cU \;,
\end{cases}
\]
if $t \le T$, and in an arbitrary manner if $t\ge T+1$. Fix $\alpha = (\alpha_t)_{t\in\N} \in \sA(\sG)$. Let $(X_t)_{t \in \N}, (Z_t)_{t \in \N}$ be the sequences of player's actions and feedback associated to the strategy $\alpha$.

Fix $(u_t)_{t \in \N} \subset [0,1]$ and $(v_t)_{t\in \N} \subset \cV$. Define by induction (on $t$) the sequences $(x_t)_{t \in \N}$ and $(z_t)_{t\in \N}$ via the relationships
\[
x_t = \alpha_t\brb{\xi(u_1),\dots,\xi(u_{t}), z_1, \dots, z_{t-1}}, \qquad z_t = \beta_t\brb{x_t, v_t, \zeta(u_t)}.
\]
Note that for each $t \in \N$, we have that $x_t$ depends only on $u_1,\dots, u_t, v_1,\dots, v_{t-1}$, so we can define
\[
\tilde{\alpha}_{t}(u_1,\dots, u_t, v_1,\dots, v_{t-1}) := x_t.
\]
Being $(u_t)_{t \in \N}$ and $(v_t)_{t\in \N}$ arbitrary, this defines a sequence of functions $(\tilde{\alpha}_t)_{t \in \N}$ such that, for all $t\in \N$,
\[
\tilde{\alpha}_{t} \colon [0,1]^t \times \cV^{t-1} \to \cX
\]
i.e., $\tilde{\alpha}:=(\tilde{\alpha}_t)_{t \in \N} \in \sA(\tilde{\sG})$. Let $(\tilde{X}_t)_{t \in \N}$ and $(\tilde{V}_t)_{t \in \N}$ be respectively the sequence of player's actions and the feedback sequence associated with the strategy $\tilde{\alpha}$. For each $t \in \N$, define also $\tilde{Z}_t := \beta_{t}\brb{\tilde{X}_t, \tilde{V}_t, \zeta(U_t)}$. Note that for each $t \in \N$ it holds that $\tilde{X}_t = \alpha_t \brb{\xi(U_1), \dots, \xi(U_t), \tilde{Z}_1, \dots, \tilde{Z}_{t-1}}$.

Fix a scenario $\P \in \sS$. 
Note first that $\P_{\xi(U_1)} = \P_{U_1}$, and since $X_1 = \alpha_1(U_1)$ and $\tilde{X}_1 = \tilde{\alpha}_1(U_1) = \alpha_1\brb{\xi(U_1)}$, we also have that $\P_{\tilde{X}_1, \xi(U_1)} = \P_{X_1, U_1} =: \Q_1$. 
Now, up to a set with $\Q_{1}$-probability zero, if $x_1 \in \cX$ and $u_1 \in [0,1]$, using Items~(\ref{buone})~and~(\ref{butwo}), we have that
\begin{multline*}
    \P_{\tilde{Z}_1 \mid \tilde{X}_1=x_1, \xi(U_1) = u_1} 
=
    \P_{\beta_1\brb{\tilde{X}_1, \tilde{\fhi}(\tilde{X}_1,Y_1), \zeta(U_1)} \mid \tilde{X}_1=x_1, \xi(U_1) = u_1}
=
    \P_{\beta_1\brb{x_1, \tilde{\fhi}(x_1,Y_1), \zeta(U_1)}}
\\
\begin{aligned}
&=
\begin{cases}
\P_{\beta_1\Brb{x_1 , \pi\brb{\fhi(x_1,Y_1)}, \zeta(U_1)} } &\text{if } x_1 \in \cR \\
\P_{\beta_1\brb{ x_1, *, \zeta(U_1) }} &\text{if } x_1 \in \cU \\
\end{cases}
=
\begin{cases}
\P_{\Brb{ \pi\brb{\fhi(x_1,Y_1)} , \psi_{1,x_1} \brb{\zeta(U_1)} } } &\text{if } x_1 \in \cR \\
\P_{\gamma_{1,x_1}\brb{\zeta(U_1)}} &\text{if } x_1 \in \cU \\
\end{cases}
\\
&=
\begin{cases}
\P_{ \pi\brb{\fhi(x_1,Y_1)} } \otimes \P_{ \psi_{1,x_1} \brb{\zeta(U_1)}} &\text{if } x_1 \in \cR \\
\P_{\gamma_{1,x_1}\brb{\zeta(U_1)}} &\text{if } x_1 \in \cU \\
\end{cases}
=
\begin{cases}
\P_{ \pi\brb{\fhi(x_1,Y_1)} } \otimes \brb{\P_{\zeta(U_1)}}_{ \psi_{1,x_1} } &\text{if } x_1 \in \cR \\
\brb{\P_{\zeta(U_1)}}_{\gamma_{1,x_1}} &\text{if } x_1 \in \cU \\
\end{cases}
\\
&=
\begin{cases}
\P_{ \pi\brb{\fhi(x_1,Y_1)} } \otimes \brb{\mu_L}_{ \psi_{1,x_1} } &\text{if } x_1 \in \cR \\
\brb{\mu_L}_{\gamma_{1,x_1}} &\text{if } x_1 \in \cU \\
\end{cases}
=
\P_{\fhi(x_1,Y_1)} 
=
    \P_{\fhi(X_1,Y_1) \mid X_1=x_1, U_1=u_1} = \P_{Z_1 \mid X_1=x_1, U_1=u_1} \;.
\end{aligned}
\end{multline*}
So, if $A_1 \subset \cZ$ and $D \subset \cX \times [0,1]$, then
\begin{align*}
\P_{\tilde{Z}_1, \brb{\tilde{X}_1, \xi(U_1)}}(A_1 \times D) &= \int_{D} \P_{\tilde{Z}_1 \mid \tilde{X}_1=x_1, \xi(U_1) = u_1}(A_1) \dif\P_{\tilde{X}_1, \xi(U_1)}(x_1,u_1)
\\
&=
\int_{D} \P_{Z_1 \mid X_1=x_1, U_1 = u_1}(A_1) \dif\P_{X_1, U_1}(x_1,u_1) = \P_{Z_1, (X_1, U_1)}(A_1 \times D) \;,
\end{align*}    
from which it follows that $\P_{\tilde{Z}_1, \tilde{X}_1, \xi(U_1)} = \P_{Z_1, X_1, U_1}$.
By induction, suppose that for $t\in [T-1]$ we have that
\[
\P_{\tilde{Z}_t,\dots,\tilde{Z}_1,\tilde{X}_t,\dots,\tilde{X}_1, \xi(U_t),\dots,\xi(U_1)} = \P_{Z_t,\dots,Z_1,X_t,\dots,X_1, U_t,\dots,U_1} \;.
\]
Then, using independence we have that
\[
\P_{\tilde{Z}_t,\dots,\tilde{Z}_1,\tilde{X}_t,\dots,\tilde{X}_1, \xi(U_{t+1}),\xi(U_t),\dots,\xi(U_1)} = \P_{Z_t,\dots,Z_1,X_t,\dots,X_1, U_{t+1},U_t,\dots,U_1} \;.
\]
Furthermore, since $X_{t+1} = \alpha_{t+1}(U_1,\dots, U_{t+1}, Z_1, \dots, Z_t)$ and 

\[
    \tilde{X}_{t+1} 
= 
    \tilde{\alpha}_{t+1} (U_1,\dots, U_{t+1}, \tilde{V}_1, \dots, \tilde{V}_t) = \alpha_{t+1} (\xi(U_1),\dots, \xi(U_{t+1}), \tilde{Z}_1, \dots, \tilde{Z}_t)
\]
we have that
\[
\P_{\tilde{Z}_t,\dots,\tilde{Z}_1,\tilde{X}_{t+1},\tilde{X}_t,\dots,\tilde{X}_1, \xi(U_{t+1}),\xi(U_t),\dots,\xi(U_1)} = \P_{Z_t,\dots,Z_1,X_{t+1},X_t,\dots,X_1, U_{t+1},U_t,\dots,U_1} =: \Q_{t+1} \;.
\]
Now, up to a set with $\Q_{t+1}$-probability zero, if $x_1, \dots, x_{t+1} \in \cX$, $u_1, \dots, u_{t+1} \in [0,1]$ and $z_1, \dots, z_t \in \cZ$, using the $\P$-independence of $Y_1,\ldots,Y_{t+1}$ and Items~\eqref{buone}--\eqref{butwo}, we have that
\begin{multline*}
    \P_{\tilde{Z}_{t+1} \mid \tilde{Z}_t = z_t,\dots,\tilde{Z}_1 = z_1,\tilde{X}_{t+1} =x_{t+1},\dots,\tilde{X}_1 = x_1, \xi(U_{t+1}) = u_{t+1},\dots,\xi(U_1)=u_1}
\\
\begin{aligned}
&=
    \P_{\beta_{t+1}\brb{\tilde{X}_{t+1}, \tilde{\fhi}(\tilde{X}_{t+1},Y_{t+1}), \zeta(U_{t+1})} \mid \tilde{Z}_t = z_t,\dots,\tilde{Z}_1 = z_1,\tilde{X}_{t+1} =x_{t+1},\dots,\tilde{X}_1 = x_1, \xi(U_{t+1}) = u_{t+1},\dots,\xi(U_1)=u_1}
\\
&=
    \P_{\beta_{t+1}\brb{x_{t+1}, \tilde{\fhi}(x_{t+1},Y_{t+1}), \zeta(U_{t+1})}}
=
\begin{cases}
\P_{\beta_{t+1}\Brb{x_{t+1} , \pi\brb{\fhi(x_{t+1},Y_{t+1})}, \zeta(U_{t+1})} } &\text{if } x_{t+1} \in \cR \\
\P_{\beta_{t+1}\brb{ x_{t+1}, *, \zeta(U_{t+1}) }} &\text{if } x_{t+1} \in \cU \\
\end{cases}
\\
&=
\begin{cases}
\P_{\Brb{ \pi\brb{\fhi(x_{t+1},Y_{t+1})} , \psi_{{t+1},x_{t+1}} \brb{\zeta(U_{t+1})} } } &\text{if } x_{t+1} \in \cR \\
\P_{\gamma_{{t+1},x_{t+1}}\brb{\zeta(U_{t+1})}} &\text{if } x_{t+1} \in \cU \\
\end{cases}
\\
&
=
\begin{cases}
\P_{ \pi\brb{\varphi(x_{t+1},Y_{t+1})} } \otimes \P_{ \psi_{{t+1},x_{t+1}} \brb{\zeta(U_{t+1})}} &\text{if } x_{t+1} \in \cR \\
\P_{\gamma_{{t+1},x_{t+1}}\brb{\zeta(U_{t+1})}} &\text{if } x_{t+1} \in \cU \\
\end{cases}
\\
&=
\begin{cases}
\P_{ \pi\brb{\varphi(x_{t+1},Y_{t+1})} } \otimes \brb{\P_{\zeta(U_{t+1})}}_{ \psi_{{t+1},x_{t+1}} } &\text{if } x_{t+1} \in \cR \\
\brb{\P_{\zeta(U_{t+1})}}_{\gamma_{{t+1},x_{t+1}}} &\text{if } x_{t+1} \in \cU \\
\end{cases}
\\
&
=
\begin{cases}
\P_{ \pi\brb{\varphi(x_{t+1},Y_{t+1})} } \otimes \brb{\mu_L}_{ \psi_{{t+1},x_{t+1}} } &\text{if } x_{t+1} \in \cR \\
\brb{\mu_L}_{\gamma_{{t+1},x_{t+1}}} &\text{if } x_{t+1} \in \cU \\
\end{cases}
\\
&=
\P_{\varphi(x_{t+1},Y_{t+1})} 
=
\P_{\varphi(X_{t+1},Y_{t+1}) \mid Z_t = z_t,\dots,Z_1 = z_1,X_{t+1} =x_{t+1},\dots,X_1 = x_1, U_{t+1} = u_{t+1},\dots,U_1=u_1} 
\\
&= \P_{Z_{t+1} \mid Z_t = z_t,\dots,Z_1 = z_1,X_{t+1} =x_{t+1},\dots,X_1 = x_1, U_{t+1} = u_{t+1},\dots,U_1=u_1} \;.
\end{aligned}
\end{multline*}

So, if $A_{t+1}\subset \cZ, D \subset \cZ^t \times \cX^{t+1}\times[0,1]^{t+1}$, we have that
\begin{multline*}
\P_{\tilde{Z}_{t+1},\brb{\tilde{Z}_{t}\dots,\tilde{Z}_1,\tilde{X}_{t+1},\dots,\tilde{X}_1, \xi(U_{t+1}),\dots,\xi(U_1)}}(A_{t+1}\times D)
\\
\begin{aligned}
&=
\int_{D} \P_{\tilde{Z}_{t+1} \mid \tilde{Z}_{t:1} = z_{t:1},\tilde{X}_{t+1:1} =x_{t+1:1}, \brb{\xi(U_{t+1}),\dots, \xi(U_{1})} = u_{t+1:1}}(A_{t+1}) \dif \Q_{t+1}
(z_{t:1}, x_{t+1:1}, u_{t+1:1})
\\
&=
\int_{D} \P_{Z_{t+1} \mid Z_{t:1} = z_{t:1},X_{t+1:1} =x_{t+1:1}, U_{t+1:1} = u_{t+1:1}}(A_{t+1})
\dif \Q_{t+1}
(z_{t:1}, x_{t+1:1}, u_{t+1:1})
\\
&=
\P_{Z_{t+1},\brb{Z_{t},\dots,Z_1,X_{t+1},\dots,X_1, U_{t+1},\dots,U_1}}(A_{t+1}\times D)
\\
\end{aligned}
\end{multline*}
from which it follows that $\P_{\tilde{Z}_{t+1:1}, \tilde{X}_{t+1:1}, \brb{\xi(U_{t+1}),\dots,\xi(U_1)} } = \P_{Z_{t+1:1}, X_{t+1:1}, U_{t+1:1}}$.
In particular, for each $t \in [T]$ we have that $\P_{X_t} = \P_{\tilde{X}_t}$. So, for each $t \in [T]$, using the $\P$-independence of $Y_1,\ldots,Y_{t}$, we have that
\[
\P_{X_t,Y_t} = \P_{X_t}\otimes\P_{Y_t} = \P_{\tilde{X}_t}\otimes\P_{Y_t} = \P_{\tilde{X}_t,Y_t} \;,
\]
and then
\[
\E_\P \bsb{\rho(X_t,Y_t)} = \E_{\P_{X_t,Y_t}} \bsb{\rho} = \E_{\P_{\tilde{X}_t,Y_t}} \bsb{\rho} = \E_\P \bsb{\rho(\tilde{X}_t,Y_t)} \;.
\]
In conclusion
\begin{multline*}
    R_T^\P(\alpha) = \sup_{x \in \cX} \E_\P \lsb{\sum_{t=1}^T \rho(x,Y_t) - \sum_{t=1}^T \rho(X_t,Y_t)} = \sup_{x \in \cX} \bbrb{\sum_{t=1}^T \E_\P \lsb{ \rho(x,Y_t)}-\sum_{t=1}^T\E_\P \lsb{\rho(X_t,Y_t)}} 
\\
= \sup_{x \in \cX} \bbrb{\sum_{t=1}^T \E_\P \lsb{ \rho(x,Y_t)}-\sum_{t=1}^T\E_\P \lsb{\rho(\tilde{X}_t,Y_t)}} =  \sup_{x \in \cX} \E_\P \lsb{\sum_{t=1}^T \rho(x,Y_t) - \sum_{t=1}^T \rho(\tilde{X}_t,Y_t)} = R_T^\P(\tilde{\alpha}) \;.
\end{multline*}
Since $\P$ was arbitrary, it follows that $R_T^{\sS}(\alpha) = R_T^{\sS}(\tilde{\alpha})$. Since $\alpha$ was arbitrary, it follows that
\[
\Rs_T(\sG) = \inf_{\alpha \in \sA(\sG)} R_T^{\sS}(\alpha) = \inf_{\alpha \in \sA(\sG)} R_T^{\sS}(\tilde{\alpha}) \ge \inf_{\alpha' \in \sA(\tilde{\sG})} R_T^{\sS}(\alpha') = \Rs_T(\tilde{\sG}) \;.
\]
\end{proof}

\section{\texorpdfstring{$\sqrt{T}$}{sqrt(T)} Lower Bound Under Full-Feedback (iv+bd) \tccheck}
\label{s:lower-full}

In this section, we prove that in the full-feedback case, no strategy can beat the $\sqrt{T}$ rate that we proved in 
\Cref{thm:upper_full} when the seller/buyer pair $(S_t,B_t)$ is drawn i.i.d. from an unknown fixed distribution, not even under the further assumptions that the valuations of the seller and buyer are independent of each other and have bounded densities.

The idea of the proof is to build a family of scenarios $\P^{\pm\e}$ parameterized by $\e \in [0,1]$, like in \cref{f:root-t-full}.
The only way to avoid suffering linear regret in a scenario $\P^{\pm\e}$ is to identify the sign of $\pm \e$.
Leveraging the Embedding and Simulation lemmas (\cref{l:embedding,l:simulation}), this construction leads to a reduction to a two-action expert problem, which has a know lower bound on the regret of order $\sqrt{T}$.

%%%% This command restate the theorem with the same numbering
\begin{theorem}[\cref{thm:lower-full}, restated]
In the full-feedback stochastic (iid) setting with independent valuations (iv) and densities bounded by a constant $M\ge 4$ (bd), for all horizons $T\in \N$, the minimax regret satisfies
\[
    \Rs_T = \Omega \brb{ \sqrt{T} } \;.
\]
\end{theorem}

\begin{proof}
Fix any horizon $T\in \N$ and any $M \ge 4$. 
Recalling Appendix~\ref{s:biltrad-setting-appe}, the full-feedback stochastic (iid) setting with independent valuations (iv) and densities bounded (bd) by $M$ is a game $\sG := (\cX, \cY, \cZ, \rho, \varphi, \sP)$, where $\cX = [0,1]$, $\cY = [0,1]^2$, $\cZ = [0,1]^2$, $\rho = \gft$, $\fhi\colon \brb{p, (s,b)} \mapsto (s,b)$, and $\sP = \sPivbd$.
Define, for each $\e \in [-1,1]$, the densities $f_{S,\e} = 2(1+\e)\I_{[0,\frac{1}{4}]} + 2(1-\e)\I_{[\frac{1}{2},\frac{3}{4}]}$ and $f_B = 2\I_{[\frac{1}{4},\frac{1}{2}]\cup[\frac{3}{4},1]}$. 
Fix the adversary's behavior $\sP_1$ as the subset of $\sP$ whose elements have the form $\bmu_\e := \otimes_{t \in \N} (f_{S,\e}\mu_L \otimes f_B\mu_L)$, for some $\e \in [-1,1]$. 
Since $\sP_1 \subset \sP$, the game $\sG_1 := (\cX, \cY, \cZ, \rho, \varphi, \sP_1)$ is easier than $\sG$ (i.e., $\Rs_T(\fG) \ge \Rs_T(\fG_1)$) by the Embedding lemma (\cref{l:embedding}) with $\slf$ and $\slg$ as the identities, and $\slh$ as the inclusion. 
Now, define $\rho_1 \colon \cX \times \cY \to [0,1]$, $\lrb{p,(s,b)} \mapsto (b-s)\I\lcb{s \le \frac{1}{4} \le b}\I\lcb{p \le \frac{1}{2}} + (b-s)\I\lcb{s \le \frac{3}{4} \le b}\I\lcb{p > \frac{1}{2}}$ and note that, defining $\fG_2 := (\cX, \cY, \cZ, \rho_1, \varphi, \fP_1)$, by the Embedding lemma with $\slf,\slg,\slh$ as the identities, we have that the game $\fG_2$ is easier than the game $\fG_1$ (i.e., $\Rs_T(\fG_1) \ge \Rs_T(\fG_2)$). 
Then, let $\cZ_3 := \{0,1\} \times \bsb{ 0, \frac{1}{4} } \times [0,1]$ and $\fhi_3 \colon \cX \times \cY \to \cZ_3$, $\brb{ p, (s,b) } \mapsto \brb{ \I\{s \le \nicefrac{1}{4}\}, \,  s\I\{s \le \nicefrac{1}{4}\} + (s-\nicefrac{1}{2})\I \{ \nicefrac{1}{2} \le s \le \nicefrac{3}{4} \}, \, b }$.
Define the game $\sG_3 := (\cX, \cY, \cZ_3, \rho_1, \fhi_3, \sP_1)$. 
By the Embedding lemma with $\slf,\slh$ as the identities and $\slg \colon \cZ_3 \to \cZ$, $(i, \tilde{s}, b) \mapsto \brb{ \tilde s i + (\nicefrac{1}{2}+ \tilde s)(1-i) , \, b }$, we have that the game $\sG_3$ is easier than the game $\sG_2$ (i.e., $\Rs_T(\fG_2) \ge \Rs_T(\fG_3)$). 
Next, let $\varphi_4 \colon \cX \times \cY \to \cZ_3$, $\brb{p,(s,b)} \mapsto \I\{s \le \frac{1}{4}\}$, and define the game $\fG_4 := (\cX, \cY, \cZ_3, \rho_1, \varphi_4, \fP_1)$.
Let $(Y_t)_{t\in \N}$ be the adversary's actions in $\sG_4$.
A tedious computation verifies that for all $t\in \N$, $p\in \cX$, and scenarios $\P$ of game $\sG_3$, $\P_{\fhi_3 (p, Y_t)} = \P_{\pi ( \fhi_3(p, Y_t) ) } \otimes (\nu \otimes f_B \mu_L)$, where $\pi \colon \cZ_3 \to \{0,1\}$ is the projection on the first component $\{0,1\}$ of $\cZ_3$ and $\nu$ is the uniform distribution on $[0,\nicefrac{1}{4}]$.
By the well-known Skorokhod representation \cite[Section 17.3]{williams1991probability}, there exists $\psi \colon [0,1] \to [0,\nicefrac{1}{4}] \times [0,1]$ such that $\nu \otimes f_B \mu_L = (\mu_L)_\psi$.
Thus, by the Simulation lemma (\cref{l:simulation}) with $\cR = \cX$ and $\cU = \varnothing$, the game $\fG_4$ is easier than $\fG_3$ (i.e., $\Rs_T(\fG_3) \ge \Rs_T(\fG_4)$).
Finally, consider the game $\fG_5 := \brb{ \{1,2\}, \{1,2\}, \{0,1\}, \rho_5, \fhi_5, \sP_5 }$, where in matrix notation, $\rho_5 = \bsb{ \rho_5(i,j) }_{i,j \in \{1,2\}}$ and $\fhi_5 = \bsb{ \fhi_5(i,j) }_{i,j \in \{1,2\}}$ are given by
\[
\rho_5:=\begin{bmatrix}
1/2 & 3/8   \\
3/8 & 1/2 
\end{bmatrix} \;,
\qquad
\fhi_5:=\begin{bmatrix}
1 & 0   \\
1 & 0 
\end{bmatrix} \;,
\]
and $\sP_5$ is the set of all measures $\tilde \bmu_\e$ of the form $\tilde \bmu_\e = \otimes_{t=1}^\infty \brb{ \frac{1+\e}{2}\delta_{1}+\frac{1-\e}{2} \delta_{2}}$ for some $\e \in [-1,1]$, where $\delta_i$ is the Dirac measure at $i \in \{1,2\}$.
Thus, letting $\sS_4$ and $\sS_5$ be the two sets of scenarios in games $\sG_4$ and $\sG_5$ respectively (note that $\sS_4$ coincides with the set of scenarios of $\sG_1$) and using again the Embedding lemma, this time with $\slf\colon [0,1] \to \{1,2\}$, $p\mapsto \I\{p\le \nicefrac{1}{2}\} + 2 \I\{p > \nicefrac{1}{2}\}$, $\slg\colon \{0,1\} \to \{0,1\}$, $i\mapsto i$, and $\slh \colon \sS_5 \to \sS_4$, $\tilde{\bmu}_\e \otimes \bmu_L \mapsto \bmu_\e \otimes \bmu_L$, we obtain that $\fG_5$ is easier than $\sG_4$ (i.e., $\Rs_T(\fG_4) \ge \Rs_T(\fG_5)$).
This last game $\sG_5$ is an online learning problem with full information (also known as learning with expert advice), whose minimax regret is known to be lower bounded by $\frac{1}{8 \sqrt{2 \pi}}\sqrt{T}$ \citep{cover65}.
In conclusion, we proved that $\Rs_T(\sG)\ge \Rs_T(\sG_5) \ge \frac{1}{8 \sqrt{2 \pi}}\sqrt{T}$.
\end{proof}

\section{Proof of \texorpdfstring{$T^{2/3}$}{T\^{}(2/3)} Lower Bound Under Realistic Feedback (iv+bd) \tccheck}
\label{s:proof-t-two-thrid-lower-bound-appe}

In this section we give a detailed proof of our $T^{2/3}$ lower bound of \cref{sec:candidate}
which hinges in a non-trivial way on our Embedding and Simulation lemmas (\cref{l:embedding,l:simulation}).
We denote Bernoulli distributions with parameter $\lambda$ by $\ber_\lambda$. 

%%% To restate the theorem
\begin{theorem}[\cref{thm:lower-real-iv+bd}, restated]
In the realistic-feedback stochastic (iid) setting with independent valuations (iv) and densities bounded by a constant $M\ge 24$ (bd), for all horizons $T \in \N$, the minimax regret satisfies
\[
    \Rs_T \ge \frac{11}{672} T^{2/3} \;.
\]
\end{theorem}

\begin{proof}
Fix an arbitrary horizon $T\in \N$ and any $M \ge 24$. 
Recalling Appendix~\ref{s:biltrad-setting-appe}, the realistic-feedback stochastic (iid) setting with independent valuations (iv) and densities bounded (bd) by $M$ is a game $\sG := (\cX, \cY, \cZ, \rho, \varphi, \sP)$, where $\cX = [0,1]$, $\cY = [0,1]^2$, $\cZ = \{0,1\}^2$, $\rho = \gft$, $\fhi\colon \brb{p, (s,b)} \mapsto \brb{ \I\{s\le p\},\,\I\{p\le b\} }$, and $\sP = \sPivbd$.
The idea of the proof is to build a sequence of games, each one easier than the former, the last of which has a known lower bound on its minimax regret.
In the first step we limit the adversary's behavior to a parametric family which is easily manageable and well-represents the difficulty of the problem (see \cref{f:t-two-third-lower-bound}).
In the second step, we increase the reward of suboptimal actions in order to have only three possible expected-reward values in each scenario.
In the third and fifth steps we increase the feedback, presenting it in a way that highlights that only its first component is informative.
In step four and six, we simulate-away the uninformative parts of the feedback.
Finally, in step 7 we show that the resulting game is harder than a known partial monitoring game with minimax regret of order at least $T^{2/3}$.

\paragraph{Step 1.}
Let $\tht := \nicefrac{1}{48}$. 
Define the following densities of the seller and buyer, respectively, by
\begin{align*}
    f_{S,\e}
&
:=
    \frac{1}{4\tht} \lrb{
    (1+\e) \I_{[0,\tht]}
+ 
    (1-\e) \I_{\lsb{ \frac{1}{6}, \frac{1}{6} + \tht }} 
+
    \I_{\lsb{ \frac{1}{4}, \frac{1}{4} + \tht }}
+
    \I_{\lsb{ \frac{2}{3}, \frac{2}{3} + \tht }}
    }\;,
    \forall \e \in [-1,1] \;,
    \tag{\text{red/blue in \cref{f:t-two-third-lower-bound}}}
\\
    f_B
&
:=
    \frac{1}{4\tht} \lrb{
    \I_{\lsb{ \frac{1}{3}-\tht,\,\frac{1}{3} }}
+ 
    \I_{\lsb{ \frac{3}{4} - \tht,\, \frac{3}{4} }} 
+
    \I_{\lsb{ \frac{5}{6} - \tht,\, \frac{5}{6} }}
+
    \I_{\lsb{ 1-\tht,\, 1 }}
    } \;.
    \tag{\text{green in \cref{f:t-two-third-lower-bound}}}
\end{align*}
Define $\fP_1$ as the subset of $\fP$ whose elements have the form $\bmu_\e := \otimes_{t \in \N} (f_{S,\e}\mu_L \otimes f_B\mu_L)$ for $\e \in [-1,1]$.
Since $\sP_1 \subset \sP$, the game $\sG_1 := (\cX, \cY, \cZ, \rho, \varphi, \sP_1)$ is easier than $\sG$ (i.e., $\Rs_T(\fG) \ge \Rs_T(\fG_1)$) by the Embedding lemma (\cref{l:embedding}) with $\slf$ and $\slg$ as the identities, and $\slh$ as the inclusion.
\paragraph{Step 2.}
Define 
$\rho_2\colon \cX \times \cY \to [0,1]$, 
$\brb{ p, (s,b) } \mapsto 
    \gft\brb{ \frac{1}{6} + \tht, (s,b) } \I\bcb{ p < \frac{1}{4} } 
    + \gft\brb{ \frac{1}{4} + \tht, (s,b) } \I\bcb{ \frac{1}{4} \le p < \frac{1}{3} } 
    + \gft\brb{ \frac{2}{3} + \tht, (s,b) } \I\bcb{ \frac{1}{3} < p  }
$.
By the Embedding lemma with $\slf$, $\slg$, and $\slh$ as the identities, we have that the game $\sG_2 := (\cX, \cY, \cZ, \rho_2, \varphi, \sP_1)$ is easier than $\sG_1$ (i.e., $\Rs_T(\fG_1) \ge \Rs_T(\fG_2)$).

\paragraph{Step 3.}
Define $\cZ_3 := \bcb{0, \frac{1}{6}, \frac{1}{4}, \frac{2}{3}} \times [0,\tht] \times \{0,1\} \times \{0,1\} \times \cX$ and
$\fhi_3 \colon \cX \times \cY \to \cZ_3$, 
\[
    \brb{ p, (s,b) }
    \mapsto
    \begin{cases}
        \brb{ \eta(s), s-\eta(s), 0, \I\{p \le b\}, p } \;,
    & 
        \text{ if } p < \frac{1}{4} \;,
    \\
        \brb{ 0, 0, \I\{s \le p\}, \I\{p \le b\}, p } \;,
    & 
        \text{ if } p \ge \frac{1}{4} \;,
    \end{cases}
\]
where $\eta \colon [0,1] \to \bcb{0, \frac{1}{6}, \frac{1}{4}, \frac{2}{3}}$,
$
    s 
    \mapsto 
    \frac{1}{6} \I \bcb{ \frac{1}{6} \le s \le \frac{1}{6} +\tht }
    +
    \frac{1}{4} \I \bcb{ \frac{1}{4} \le s \le \frac{1}{4} +\tht }
    +
    \frac{2}{3} \I \bcb{ \frac{2}{3} \le s \le \frac{2}{3} +\tht }
$.
Define the game $\sG_3 := (\cX, \cY, \cZ_3, \rho_2, \fhi_3, \sP_1 )$.
By the Embedding lemma with $\slf,\slh$ as the identities and 
\[
    \slg \colon \cZ_3 \to \cZ\;,
    \quad
    (v,u,i,j,p)
    \mapsto
    \begin{cases}
    \brb{ \I\{v+u\le p\}, j }
    &
        \text{ if } p < \frac{1}{4} \;,
    \\
        (i,j) \;,
    & 
        \text{ if } p \ge \frac{1}{4} \;,
    \end{cases}
\]
we have that the game $\sG_3$ is easier than $\sG_2$ (i.e., $\Rs_T(\fG_2) \ge \Rs_T(\fG_3)$).

\paragraph{Step 4.}
Let $\cZ_4 := \bcb{0, \frac{1}{6}, \frac{1}{4}, \frac{2}{3}}$
and
$
    \fhi_4 \colon \cX \times \cY \to \cZ_4
$, 
$
    \brb{ p, (s,b) }
    \mapsto
    \eta(s) \I \lcb{ p < \frac{1}{4} } 
$.
Define the game $\sG_4 := (\cX, \cY, \cZ_4, \rho_2, \fhi_4, \sP_1)$.
Let $(Y_t)_{t\in \N} = (S_t,B_t)_{t\in \N}$ be the adversary's actions in $\sG_4$,
$E := \bsb{0, \tht} \cup \bsb{\frac{1}{6}, \frac{1}{6} + \tht} \cup \bsb{\frac{1}{4}, \frac{1}{4} + \tht} \cup \bsb{\frac{2}{3}, \frac{2}{3} + \tht}$ and $F := \bsb{\frac{1}{3} - \tht, \frac{1}{3}} \cup \bsb{\frac{3}{4} - \tht, \frac{3}{4}} \cup \bsb{\frac{5}{6} - \tht, \frac{5}{6}} \cup \bsb{1 - \tht, 1}$. 
A long and tedious computation verifies that for all $t\in \N$, 
\begin{itemize}
    \item For each $p\in [0,\nicefrac{1}{4})$ and any scenario $\P$ of game $\sG_3$, $\P_{\fhi_3 (p, Y_t)} = \P_{ \eta(S_t) } \otimes (\nu \otimes \delta_0 \otimes \ber_{\lambda_{F,p}} \otimes \delta_p)$, where $\nu$ is the uniform distribution on $[0,\tht]$ and $ \lambda_{F,p} := \frac{1}{4\tht} \mu_L \bsb{ [p,1] \cap F } $.
    By the well-known Skorokhod representation \cite[Section 17.3]{williams1991probability}, there exists $\psi_p \colon [0,1] \to [0,\tht] \times \{0,1\} \times \{0,1\} \times \cX$ such that $\nu \otimes \delta_0 \otimes \ber_{\lambda_{F,p}} \otimes \delta_p = (\mu_L)_{\psi_p}$.
    \item For each $p\in [\nicefrac{1}{4},1]$ and any scenario $\P$ of game $\sG_3$, $\P_{\fhi_3 (p, Y_t)} = \delta_0 \otimes \delta_0 \otimes \ber_{\lambda_{E,p}} \otimes \ber_{\lambda_{F,p}} \otimes \delta_p$, where $\lambda_{E,p} := \frac{1}{4\tht} \mu_L \bsb{ [0,p] \cap E }$ and $ \lambda_{F,p} := \frac{1}{4\tht} \mu_L \bsb{ [p,1] \cap F } $.
    By the Skorokhod representation, there exists $\gamma_p \colon [0,1] \to \cZ_3$ such that $\delta_0 \otimes \delta_0 \otimes \ber_{\lambda_{E,p}} \otimes \ber_{\lambda_{F,p}} \otimes \delta_p = (\mu_L)_{\gamma_p}$.
\end{itemize}
Thus, by the Simulation lemma (\cref{l:simulation}) with $\cR = [0, \nicefrac{1}{4})$ and $\cU = [\nicefrac{1}{4}, 1]$, the game $\fG_4$ is easier than $\fG_3$ (i.e., $\Rs_T(\fG_3) \ge \Rs_T(\fG_4)$).

\paragraph{Step 5.}
Let 
$\cY_5 := \cY^\N$, 
$\cZ_5 := \{0,1\} \times \brb{ \N \cup \{\iop\} } \times \{0,1\} \times \cX$, 
$\rho_5 \colon \cX \times \cY_5 \to [0,1]$, 
$\brb{ p, (s_k, b_k)_{k\in\N} } \mapsto \rho_2 (p,s_1,b_1)$,
\[
    \fhi_5 \colon \cX \times \cY_5 \to \cZ_5\;,
    \quad
    \brb{ p, (s_k, b_k)_{k\in\N} } 
    \mapsto 
    \begin{cases}
        \Brb{ \I \bcb{ \eta(s_\tau) = 0 }, \tau, \I \bcb{ \eta(s_1) = \frac{1}{4} }, p } \;,
    &
        \text{ if } p \in \bigl[ 0, \frac{1}{4} \bigr) \;,
    \\
        (0,1,0,p) \;,
    &
        \text{ if } p \in \bigl[ \frac{1}{4}, 1 \bigr]\;,
    \end{cases}
\]
where $\eta$ is defined in game $\sG_3$, $\tau := \inf \bcb{ k\in \N \mid \eta( s_k ) \in \{0, \nicefrac{1}{6} \} } \in \N \cup \{\iop\}$, and $s_\iop := 0$.
Let $\sP_5$ be the set of measures on $\cY_5^\N$ of the form $\tilde \bmu_\e := \otimes_{t\in \N} \brb{ \otimes_{k\in\N} ( f_{S,\e} \mu_L \otimes f_B \mu_L ) }$ for $\e \in [-1,1]$, and define the game $\sG_5 := ( \cX, \cY_5, \cZ_5, \rho_5, \fhi_5, \sP_5 )$.
By the Embedding lemma with $\slf$ as the identity,
\[
    \slg \colon \cZ_5 \to \cZ_4 \;,
    \quad
    (z,k,j,p)
    \mapsto
    \frac{1}{6}(1-z) \I \lcb{ p < \frac{1}{4}, k = 1 }
    + \lrb{ \frac{1}{4}j + \frac{2}{3} (1-j) } \I \lcb{ p < \frac{1}{4}, k > 1 } \;,
\]
and $\slh \colon \tilde{ \bmu }_\e \otimes \bmu_L \mapsto \bmu_\e \otimes \bmu_L$, 
we have that the game $\sG_5$ is easier than $\sG_4$ (i.e., $\Rs_T(\fG_4) \ge \Rs_T(\fG_5)$).

\paragraph{Step 6.}
Now, define $\pi\colon \cZ_5 \to \{0,1\}$ as the projection on the first component $\{0,1\}$ of $\cZ_5$, $\cZ_6 := \{0,1\}$, $\fhi_6 := \pi \circ \fhi_5$, and the game $\sG_6 := ( \cX, \cY_5, \cZ_6, \rho_5, \fhi_6, \sP_5)$.
Let $(\tilde Y_t)_{t\in \N}$ be the adversary's actions in $\sG_5$. 
A straightforward verification shows that for all $t\in \N$, 
\begin{itemize}
    \item For each $p\in [0,\nicefrac{1}{4})$ and any scenario $\P$ of game $\sG_5$, $\P_{\fhi_5 (p, \tilde Y_t)} = \P_{ \pi \brb{ \fhi_5 (p, \tilde Y_t) } } \otimes (\nu \otimes \delta_p)$, where $\nu$ is the unique distribution on $\brb{ \N \cup \{\iop\} } \times \{0,1\}$ such that, for all $k \in\N\cup\{\iop\}$, $j\in \{0,1\}$, $\nu \bsb{ \{ (k,j) \} } = \frac{1}{2} \I \{ k=1, j=0\} + \frac{1}{2^{k+1}} \I \{1<k<\iop\}$.
    Using again the Skorokhod representation, there exists $\psi_p \colon [0,1] \to \brb{ \N \cup \{\iop\} } \times \{0,1\} \times [0,1]$ such that $\nu \otimes \delta_p = (\mu_L)_{\psi_p}$.
    \item For each $p\in [\nicefrac{1}{4},1]$ and any scenario $\P$ of game $\sG_5$, $\P_{\fhi_5 (p, \tilde Y_t)} = \delta_{(0,1,0,p)} = (\mu_L)_{\gamma_p}$, where $\gamma_p \colon [0,1] \to \cZ_5$, $\lambda \mapsto (0,1,0,p)$.
\end{itemize}
Thus, by the Simulation lemma with $\cR = [0, \nicefrac{1}{4})$ and $\cU = [\nicefrac{1}{4}, 1]$, the game $\fG_6$ is easier than $\fG_5$ (i.e., $\Rs_T(\fG_5) \ge \Rs_T(\fG_6)$).

\paragraph{Step 7.}

Finally, consider the game 
$\sG_7 := \brb{ \{1,2,3\}, \{1,2\}, \{0,1\}, \rho_7, \fhi_7, \sP_7 } $, where in matrix notation, $\rho_7 = \bsb{ \rho(i,j) }_{i \in \{1,2,3\}, j \in \{1,2\}}$ and $\fhi_7 = \bsb{ \fhi(i,j) }_{i \in \{1,2,3\}, j \in \{1,2\}}$ are given by
\[
\rho_7:= \frac{1}{96}\begin{bmatrix}
34 & 34   \\
45 & 37   \\
38 & 44
\end{bmatrix} \;,
\qquad
\fhi_7:=\begin{bmatrix}
1 & 0   \\
0 & 0   \\
0 & 0
\end{bmatrix} \;,
\]
and $\sP_7$ is the set of all measures of the form $\otimes_{t\in \N} \brb{ \frac{1+\e}{2}\delta_1 + \frac{1-\e}{2} \delta_2 }$, for $\e \in [-1,1]$.
Thus, using again the Embedding lemma, this time with $\slf\colon [0,1] \to \{1,2,3\}$, $p\mapsto \I\{p < \nicefrac{1}{4}\} + 2 \I\{ \nicefrac{1}{4} \le p \le \nicefrac{1}{3}\} + 3 \I \{ \nicefrac{1}{3} < p \}$, $\slg\colon \{0,1\} \to \{0,1\}$, $i\mapsto i$, and $\slh \colon \otimes_{t\in \N} \brb{ \frac{1+\e}{2}\delta_1 + \frac{1-\e}{2} \delta_2 } \otimes \bmu_L \mapsto \tilde \bmu_\e \otimes \bmu_L$, we obtain that $\fG_7$ is easier than $\sG_6$ (i.e., $\Rs_T(\fG_6) \ge \Rs_T(\fG_7)$).
This last game is an instance of the so-called revealing action partial monitoring game, whose minimax regret is known to be lower bounded by $\frac{11}{96} \brb{ \frac{1}{7} T^{2/3} }$ \citep{cesa2006regret}.
In conclusion, we proved that $\Rs_T(\sG)\ge \Rs_T(\sG_7) \ge \frac{11}{672} T^{2/3}$.
\end{proof}

\section{Linear Lower Bound Under Realistic Feedback (bd) \tccheck}
\label{s:lower-bd-appe}

In this section, we prove that in the realistic-feedback case, no strategy can achieve sublinear regret in the worst case if the valuations of the buyer and the seller may be dependent, not even if they have a bounded density.

The idea of the proof is to exploit the lack of observability in this setting, building a family of scenarios $\P^\lambda$ (parameterized by $\lambda \in [0,1]$) as convex combinations of the two measures in \cref{f:linear-lower-bound-lip}.
If $\lambda < \nicefrac{1}{2}$, the optimal action is $\nicefrac{3}{8}$, while if $\lambda > \nicefrac{1}{2}$, the optimal action becomes $\nicefrac{5}{8}$.
This family is built is such a way that the feedback gives no information on $\lambda$, making it impossible to distinguish between the two cases.
Leveraging the Embedding and Simulation lemmas (\cref{l:embedding,l:simulation}), this construction leads to a reduction to an instance of a non-observable partial monitoring game, whose regret is trivially lower bounded by $T/24$.

%%%% To restate the Theorem
\begin{theorem}[\cref{thm:lower-real-bd}, restated]
In the realistic-feedback stochastic (iid) setting with joint density bounded by a constant $M\ge \nicefrac{64}{3}$ (bd), for all horizons $T \in \N$, the minimax regret satisfies
\[
    \Rs_T \ge \frac{1}{24} T \;.
\]
\end{theorem}

\begin{proof}
Fix any horizon $T \in \N$ and $M\ge \nicefrac{64}{3}$.
Recalling Appendix~\ref{s:biltrad-setting-appe}, the realistic-feedback stochastic (iid) setting with joint density bounded by $M$ (bd) is a game $\sG := (\cX, \cY, \cZ, \rho, \varphi, \sP)$, where $\cX = [0,1]$, $\cY = [0,1]^2$, $\cZ = \{0,1\}^2$, $\rho = \gft$, $\fhi\colon \brb{p, (s,b)} \mapsto \brb{ \I\{s\le p\},\,\I\{p\le b\} }$, and $\sP = \sPbd$.
Define the two joint densities
$
    f 
= 
    \frac{64}{3} 
    \brb{ 
        \I_{[\nicefrac{0}{8}, \nicefrac{1}{8}]\times[\nicefrac{3}{8}, \nicefrac{4}{8}]} 
    +
        \I_{[\nicefrac{2}{8}, \nicefrac{3}{8}]\times[\nicefrac{7}{8}, \nicefrac{8}{8}]} 
    +
        \I_{[\nicefrac{4}{8}, \nicefrac{5}{8}]\times[\nicefrac{5}{8}, \nicefrac{6}{8}]} 
    }
$
and $g\colon [0,1]^2\to [0,M]$, $(s,b) \mapsto f(1-b,1-s)$ (see~\cref{f:linear-lower-bound-lip}, left).
Let $\fP_1$ be the subset of $\sPbd$ whose elements have the form $\bmu_\lambda := \otimes_{t \in \N} \brb{ \brb{ (1-\lambda)f + \lambda g }  (\mu_L\otimes\mu_L) }$ for $\lambda \in [0,1]$.
Since $\fP_1 \subset \fP$ the game $\fG_1 := (\cX, \cY, \cZ, \rho, \fhi, \fP_1)$ is easier than $\sG$ (i.e., $\Rs_T(\fG) \ge \Rs_T(\fG_1)$) by the Embedding lemma (\cref{l:embedding}) with $\slf$ and $\slg$ as the identities, and $\slh$ as the inclusion.
Define $\cZ_1 := \{0\}$ and $\fhi_1 \colon \cX \times \cY \to \cZ_1 \ , \brb{p,(s,b)} \mapsto 0$.
Let $(Y_t)_{t\in \N}$ be the adversary's actions in $\cG_1$.
Now, since for all $t\in \N$, any two scenarios $\P$ and $\Q$ of game $\sG_1$, and each $p \in [0,1]$,  $\P_{\fhi(p,Y_t)} = \Q_{\fhi(p,Y_t)}$, 
then by the well-known Skorokhod representation \cite[Section 17.3]{williams1991probability}, for each $t\in \N$ and each $p \in [0,1]$ there exists $\gamma_{t,p} \colon [0,1] \to \{0,1\}^2$ such that for any scenario $\P$ of game $\sG_1$, $\P_{\varphi(x,Y_t)} = (\mu_L)_{\gamma_{t,x}}$.
Thus, the Simulation lemma (\cref{l:simulation}) with $\cR = \varnothing$ and $\cU = \cX$ implies that the game $\fG_2 := (\cX, \cY, \cZ_2, \rho, \fhi_2, \fP_1)$ is easier than $\fG_1$ (i.e., $\Rs_T(\fG_1) \ge \Rs_T(\fG_2)$). 
Define $\rho_3 \colon \cX \times \cY \to [0,1] \ , \brb{p,(s,b)} \mapsto (b-s)\I\bcb{s \le \frac{3}{8} \le b}\I\lcb{p \le \frac{1}{2}} + (b-s)\I\lcb{s \le \frac{5}{8} \le b}\I\lcb{p > \frac{1}{2}}$ and $\fG_3 := (\cX, \cY, \cZ_2, \rho_3, \varphi_2, \fP_1)$. 
By the Embedding lemma with $\slf,\slg,\slh$ as the identities, we have that the game $\fG_3$ is easier than the game $\fG_2$ (i.e., $\Rs_T(\fG_2) \ge \Rs_T(\fG_3)$). 
Finally, consider the game 
$\sG_4 := \brb{ \{1,2\}, \{1,2\}, \{0\}, \rho_4, \fhi_4, \sP_4 } $, where in matrix notation, $\rho_4 = \bsb{ \rho(i,j) }_{i,j \in \{1,2\}}$ and $\fhi_4 = \bsb{ \fhi(i,j) }_{i,j \in \{1,2\}}$ are given by
\[
\rho_4:= \begin{bmatrix}
\nicefrac{1}{3} & \nicefrac{1}{4}   \\
\nicefrac{1}{4} & \nicefrac{1}{3} 
\end{bmatrix} \;,
\qquad
\fhi_4:=\begin{bmatrix}
0 & 0   \\
0 & 0 
\end{bmatrix} \;,
\]
and $\sP_4$ is the set of all measures of the form $(1-\lambda)\delta_1 + \lambda\delta_2$, for $\lambda \in [0,1]$.
Using again the Embedding lemma, this time with $\slf\colon [0,1] \to \{1,2\}$, $p\mapsto \I\{p \le \nicefrac{1}{2}\} + 2 \I\{ \nicefrac{1}{2} < p \}$, $\slg\colon \{0\} \to \{0\}$, $i\mapsto i$, and $\slh \colon \otimes_{t\in \N} \brb{ (1-\lambda)\delta_1 + \lambda \delta_2 } \otimes \bmu_L \mapsto 
\bmu_\lambda \otimes \bmu_L$, we obtain that $\fG_4$ is easier than $\sG_3$ (i.e., $\Rs_T(\fG_3) \ge \Rs_T(\fG_4)$).
This last game has (trivially) minimax regret at most $\brb{ \frac{1}{3}-\frac{1}{4} } \frac{T}{2}$.
In conclusion, we proved that $\Rs_T(\sG) \ge \Rs_T(\sG_4) \ge \frac{1}{24} T$.
\end{proof}

\section{Linear Lower Bound Under Realistic Feedback (iv) \tccheck}
\label{sec:linear_real-appe}

In this section, we prove that in the realistic-feedback case, no strategy can achieve sublinear regret without any limitations on how concentrated the distributions of the valuations of the seller and buyer are, not even if they are independent of each other (iv).

The idea of the proof is that if the two distributions are very concentrated in a small region, finding an optimal price is like finding a needle in a haystack.
Each strategy that (at each time step) receives as feedback only a finite number of bits, as in our realistic setting, can assign positive probability to at most a countable set of points.
Thus one could find concentrated distributions of the buyer and seller that have a unique optimal point in which the strategy has zero probability of posting prices at all time steps, and such that \emph{all} other prices suffer large regret.

\begin{theorem}[\cref{thm:lower-real-iv}, restated]
    In the realistic-feedback stochastic (iid) setting  with independent valuations (iv), for all horizons $T \in \N$, the minimax regret satisfies
\[
    \Rs_T \ge \frac{1}{8} T \;.
\]
\end{theorem}
\begin{proof}
To lighten the notation, for any $n \in \N$ and a family $(\lambda_k)_{k\in\N}$, we let $\lambda_{1:n} := (\lambda_1, \ldots, \lambda_{n})$.
Fix an arbitrary horizon $T \in \N$.
Recalling Appendix~\ref{s:biltrad-setting-appe}, the realistic-feedback stochastic (iid) setting with independent valuations (iv) is a game $\sG := (\cX, \cY, \cZ, \rho, \varphi, \sP)$, where $\cX = [0,1]$, $\cY = [0,1]^2$, $\cZ = \{0,1\}^2$, $\rho = \gft$, $\fhi\colon \brb{p, (s,b)} \mapsto \brb{ \I\{s\le p\},\,\I\{p\le b\} }$, and $\sP = \sPiv$.
Let $\sS$ be the set of scenarios of $\sG$.
Fix a strategy $\alpha$ for game $\sG$ and let $\e \in (0, 1)$.
Define $\bar \alpha_1 := \alpha_1$, $\nu_1 := (\mu_L)_{\bar \alpha _1}$, and for each $t \in \N$ and $z_1, \dots, z_{t} \in \{0,1\}^2$, 
\[
    \bar{\alpha}_{t+1,z_{1:t}} \colon [0,1]^{t+1} \to [0,1], 
    \quad
    u_{1:t+1} \mapsto \alpha_{t+1}(u_{1:t+1}, z_{1:t})
\qquad
\text{ and }
\qquad
    \nu_{t+1,z_{1:t}} 
:= 
    (\otimes_{s=1}^{t+1}\mu_L)_{\bar{\alpha}_{t+1,z_{1:t}}} \;.
\]
Define also the set $A_1 := \bcb{x \in \lsb{0,1} \mid \nu_{1}[\{x\}] > 0 }$ and, for each $t \in \N$, the union
$
    A_{t+1} 
:= 
    \bigcup_{z_{1:t}\in \{0,1\}^2} \bcb{x \in \lsb{0,1} \mid \nu_{t,z_{1:t}}[\{x\}] > 0 }
$.
Note that, for each $t\in \N, A_t$ is countable, being the union of $4^{t-1}$ countable sets.
Then $A := \bigcup_{t \in \N} A_t$ is countable.
Since $B:=[\frac{1-\e}{2}, \frac{1+\e}{2}]$ has the power of continuum, we have that the same holds for $B \m A$.
In particular, $B\m A$ is non-empty.
Pick $\xs \in B \backslash A$ and define $\mu_S := \frac{1}{2}\delta_{0}+\frac{1}{2}\delta_{\xs}$, 
$\mu_B := \frac{1}{2}\delta_{\xs} + \frac{1}{2}\delta_{1}$, 
and $\P := \lrb{ \otimes_{t \in \N} (\mu_S \otimes \mu_B) }\otimes \bmu_L  \in \sS$. Then for each $t \in \N$, we have that
\[
\E_{\P} \bsb{\rho(\xs,Y_t)} =  \frac{\xs + (1-\xs) + 1}{4}\;.
\]
On the other hand, $\P[X_1 = \xs] = \nu_1[\{\xs\}] = 0$ and for each $t \in \N$, we have that
\begin{multline*}
\P[X_{t+1} = \xs] = \P\lsb{\alpha_{t+1} (U_1, \dots, U_{t+1}, Z_1, \dots, Z_{t}) = \xs}
\\
\begin{aligned}
&= \sum_{z_1, \dots, z_{t} \in \{0,1\}^2} \P\lsb{\alpha_{t+1} (U_1, \dots, U_{t+1}, z_1, \dots, z_{t}) = \xs \cap Z_1 = z_1 \cap \dots \cap Z_{t} = z_{t}}
\\
&\le \sum_{z_1,\dots, z_{t} \in \{0,1\}^2} \P\lsb{\alpha_{t+1} (U_1, \dots, U_{t+1}, z_1, \dots, z_{t}) = \xs}
= \sum_{z_1,\dots, z_{t} \in \{0,1\}^2} \nu_{t+1,z_1,\dots, z_{t}} \lsb{\{\xs\}} = 0 \;,
\end{aligned}
\end{multline*}
which in turn gives
\begin{align*}
&
\E_{\P} \bsb{\rho(X_t,Y_t)} 
= \frac{\E_\P \lsb{\rho\brb{X_t, (0,\xs)}} + \E_\P \lsb{\rho\brb{X_t, (\xs,1)}} + \E_\P \lsb{\rho\brb{X_t, (0,1)}} + \E_\P \lsb{\rho\brb{X_t, (\xs,\xs)}}}{4}
\\
& \hspace{29.25865pt} = \frac{\xs \P_{X_t} \bsb{[0,\xs]} + (1-\xs) \P_{X_t} \bsb{[\xs,1]} + 1}{4}
= \frac{\xs \P_{X_t} \bsb{[0,\xs)} + (1-\xs) \P_{X_t} \bsb{(\xs,1]} + 1}{4}
\\
& \hspace{29.25865pt}
\le \frac{ \max(\xs, 1-\xs) + 1 }{4}
= \frac{\xs + (1-\xs) + 1 - \min(\xs, 1-\xs)}{4} \;.  
\end{align*}
So, if $T \in \N$ we get
\[
R^{\P}_T(\alpha) = \E_{\P}\lsb{\sum_{t=1}^T \rho(x^\star,Y_t) - \sum_{t=1}^T \rho(X_t,Y_t)} \ge  \frac{\min(x^\star, 1-x^\star)}{4} T \ge \frac{1-\e}{8} T.
\]
Since $\e$ was arbitrary, we get, for all $T \in \N$,
$
R_T^{\sS}(\alpha) = \sup_{\P \in \sS} R^{\P}_T(\alpha) \ge \sup_{\e \in (0, 1)} \frac{1-\e}{8} T = \nicefrac{T}{8}
$.
Since $\alpha$ was arbitrary we get, for each $T \in \N$,
$
\Rs_T = \inf_{\alpha \in \fA} R_T^{\sS}(\alpha) \ge \nicefrac{T}{8}
$.
\end{proof}

\section{Adversarial Setting: Linear Lower Bound Under Full Feedback \tccheck}
\label{sec:adversarial-appe}

In this section, we give a more detailed proof of \cref{thm:adv-lower} with a notation consistent with our abstract setting of sequential games.

\begin{theorem}[\cref{thm:adv-lower}, restated]
In the full-feedback adversarial (adv) setting, for all horizons $T \in \N$, we have
\[
    \Rs_T \ge \frac{1}{4} T \;.
\]
\end{theorem}
\begin{proof}
Recalling Appendix~\ref{s:biltrad-setting-appe}, the full-feedback adversarial (adv) bilateral trade setting is a game $\sG := (\cX, \cY, \cZ, \rho, \varphi, \sP)$, where $\cX = [0,1]$, $\cY = [0,1]^2$, $\cZ = [0,1]^2$, $\rho = \gft$, $\fhi\colon \brb{p, (s,b)} \mapsto (s,b)$, and $\sP = \sPadv$.
Let $\sS$ be the set of scenarios of $\sG$.
Fix a strategy $\alpha \in \fA$ and an $\e \in (0, 1/18)$. 
Define $\bar{\alpha}_1 := \alpha_1$, $\nu_1 := (\mu_L)_{\bar{\alpha}_1}$, and
\[
    \begin{cases}
        c_1:=\frac{1}{2}-\frac{3}{2}\varepsilon, \  d_1:=\frac{1}{2}-\frac{1}{2}\varepsilon, \  s_1 := 0,  \ b_1:=d_1,
    &
        \text{ if } \nu_{1}\bsb{\bsb{0,\frac{1}{2}-\frac{1}{2}\varepsilon}}\le\frac{1}{2} \;,
    \\
        c_1:=\frac{1}{2}+\frac{1}{2}\varepsilon, \  d_1:=\frac{1}{2}+\frac{3}{2}\varepsilon, \  s_1 := c_1, \  b_1:=1,
    &
        \text{ otherwise}.
    \end{cases}
\]
If $t \in \N$, suppose we defined $\bar{\alpha}_t, \nu_t, c_t, d_t, s_t, b_t$ and let
\[
\bar{\alpha}_{t+1} : [0,1]^{t+1} \to [0,1], (u_1,\dots,u_{t+1}) \mapsto \alpha_{t+1} \lrb{ u_1, \dots, u_{t+1}, (s_1,b_1), \dots, (s_t,b_t) }, 
\]
$\nu_{t+1} := \brb{\otimes_{s=1}^{t+1}\mu_L}_{\bar{\alpha}_{t+1}}$, and
\[
    \begin{cases}
        c_{t+1}:=c_{t}, \ d_{t+1}:=d_{t}-\frac{2\varepsilon}{3^{t}}, \ s_{t+1}:=0, \ b_{t+1}:=d_{t+1},
    &
        \text{if } \nu_{t+1}\bsb{ \bsb{0,c_t+\frac{\varepsilon}{3^{t}}} }\le\frac{1}{2} \;,
    \\
        c_{t+1}:=c_{t}+\frac{2\varepsilon}{3^{t}}, \ d_{t+1}:=d_{t}, \ s_{t+1} := c_{t+1}, \ b_{t+1}:=1,
    &
        \text{otherwise}.
    \end{cases}
\]
Then $(\bar{\alpha}_t)_{t \in \N}, (\nu_t)_{t \in \N}, (c_t)_{t \in \N}, (d_t)_{t \in \N}, (s_t)_{t \in \N}, (b_t)_{t \in \N}$ are well-defined by induction and satisfy:
\begin{itemize}
    \item For each $t \in \N$, $d_t-c_t = \frac{\varepsilon}{3^{t-1}}$.
    \item For each $t \in \N$, $c_1\le c_2 \le c_3 \le \dots \le c_t \le d_t \le \dots \le d_3 \le d_2 \le d_1$.
    \item $\exists! x^\star \in \bigcap_{t=1}^{\infty}[c_t,d_t]$.
    \item For each $t \in \N$, $\rho\lrb{x^\star, (s_t,b_t)} = b_t - s_t \ge \frac{1-3\e}{2}$.
    \item For each $t\in \N$, $\P\bsb{\alpha_t \brb{U_1,\dots, U_t, (s_1,b_1), \dots, (s_{t-1},b_{t-1})} \in [s_t, b_t] }\le\frac{1}{2}$.
\end{itemize}
Now, define $\P := \lrb{ \otimes_{t \in \N} \delta_{(s_t,b_t)}}  \otimes \bmu_L \in \sS$. Then, for each $t \in \N$,
\begin{align*}
    \E_{\P}[\rho \lrb{X_t, Y_t} ] &= \E_{\P}\Bsb{ \rho\Brb{\alpha_t \brb{U_1,\dots, U_t, (s_1,b_1), \dots, (s_{t-1},b_{t-1})} , (s_t,b_t)} }
    \\
    &\le \lrb{\frac{1}{2}+\frac{3\varepsilon}{2}} \P\bsb{\alpha_t \brb{U_1,\dots, U_t, (s_1,b_1), \dots, (s_{t-1},b_{t-1})} \in [s_t, b_t] } \le  \frac{1}{4}+\frac{3\varepsilon}{4} \;,
\end{align*}
and so, for each $T \in \N$
\begin{align*}
    R^{\P}_T(\alpha) & = \E_{\P}\lsb{\sum_{t=1}^T \rho(x^\star,Y_t) - \sum_{t=1}^T \rho(X_t,Y_t)} = \sum_{t=1}^T \rho(x^\star,(s_t,b_t)) -  \sum_{t=1}^T \E_{\P} \lsb{  \rho \lrb{X_t, Y_t} }
    \\
    &\ge
    \sum_{t=1}^T (b_t-s_t)\brb{1 - \P\bsb{\alpha_t \brb{U_1,\dots, U_t, (s_1,b_1), \dots, (s_{t-1},b_{t-1})} \in [s_t, b_t] } }
    \ge \frac{1-3\e}{4}T \;.
\end{align*}
Since $\e$ was arbitrary, we get, for all $T \in \N$,
$
R_T^{\sS}(\alpha) = \sup_{\P \in \sS} R^{\P}_T(\alpha) \ge \sup_{\e \in (0, 1/18)} \frac{1-3\e}{4}T = \frac{T}{4}
$.
Since $\alpha$ arbitrarity, we get, for each $T \in \N$,
$
\Rs_T = \inf_{\alpha \in \fA} R_T^{\sS}(\alpha) \ge \frac{T}{4}
$.
\end{proof}

\section{DKW Inequalities}

We begin this section by presenting the univariate DKW inequality as proved in \citep{massart1990tight}.

\begin{theorem}
\label{t:dkw-massart}
If $(\Omega,\cF,\P)$ is a probability space and $(X_n)_{n\in \N}$ is a $\P$-i.i.d.\ sequence of random variables, then, for any $\e>0$ and all $m\in \N$, it holds
\[
    \P \lsb{ \sup_{x\in\R} \labs{ \frac 1 m \sum_{k=1}^m \I\{ X_k \le x \} - \P[X_1 \le x] } > \e }
    \le
    2 \operatorname{exp}\brb{-2m\e^2} \;.
\]
\end{theorem}

We now present a bivariate DKW inequality which can be proved by applying the VC-type bound of \cite[Theorem~4.9; see also Lemmas~4.4, 4.5, and 4.11 for the explicit constants]{anthony2009neural}.

\begin{theorem}
\label{t:vc}
There exist positive constants $m_0 \le 1200$, $c_1 \le 13448$, $c_2 \ge 1/576$ such that, if $(\Omega,\cF,\P)$ is a probability space, $(X_n,Y_n)_{n\in \N}$ is a $\P$-i.i.d.\ sequence of two-dimensional random vectors, then, for any $\e>0$ and all $m\in \N$ such that $m \ge m_0/\e^2$, it holds
\[
    \P \lsb{ \sup_{x,y \in \R} \labs{ \frac 1 m \sum_{k=1}^m \I\{ X_k \le x, Y_k \le y \} - \P[X_1 \le x, Y_1 \le y] } > \e }
    \le
    c_1 \operatorname{exp}\brb{-c_2m\e^2} \;.
\]
\end{theorem}

\section*{Acknowledgments.}

This work was partially supported by: the ERC Advanced Grant 788893 AMDROMA ``Algorithmic and Mechanism Design Research in Online Markets'', the MIUR PRIN project ALGADIMAR ``Algorithms, Games, and Digital Markets'', the AI Interdisciplinary Institute ANITI (funded by the French ``Investing for the Future -- PIA3'' program under the Grant agreement n. ANR-19-PI3A-0004), the COST Action CA16228 ``European Network for Game Theory'' (GAMENET), the EU Horizon 2020 ICT-48 research and innovation action under grant agreement 951847, -- project ELISE (European Learning and Intelligent Systems Excellence).

\bibliographystyle{plainnat} 
\bibliography{references.bib}

\end{document}